\pdfoutput=1
\documentclass[a4paper,10pt,numbers=noenddot,twoside=on,headinclude=true,footinclude=false,headsepline=true]{scrartcl}

\usepackage[utf8]{inputenc}
\usepackage[T1]{fontenc}
\usepackage{textcomp}
\usepackage[english]{babel}
\usepackage{color}
\usepackage{amssymb}
\usepackage{amsmath}
\usepackage{amsfonts}
\usepackage{amsopn}
\usepackage{amsbsy}
\usepackage{enumerate}
\usepackage[plainpages=false,pdfpagelabels,breaklinks=true,pdfborderstyle={/S/U/W 1.2}]{hyperref}
\hypersetup{pdfauthor={Giuseppe De Nittis & Max Lein},pdftitle={Derivation of Ray Optics Equations in Photonic Crystals Via a Semiclassical Limit}}
\usepackage{graphics}
\usepackage[activate={true,nocompatibility},final,tracking=true,kerning=true,spacing=true,stretch=10,shrink=10]{microtype}

\numberwithin{equation}{section} 
\numberwithin{table}{section} 
\numberwithin{figure}{section} 

\usepackage[amsmath,thmmarks,thref,hyperref]{ntheorem}

\theoremstyle{plain}
\newtheorem{theorem}{Theorem}[section]
\newtheorem{definition}[theorem]{Definition}
\newtheorem{lemma}[theorem]{Lemma}
\newtheorem{corollary}[theorem]{Corollary}
\newtheorem{proposition}[theorem]{Proposition}
\newtheorem{assumption}[theorem]{Assumption}

\theorembodyfont{\upshape}
\newtheorem{remark}[theorem]{Remark}

\theoremstyle{nonumberplain}

\theoremsymbol{\ensuremath{_\Box}}
\newtheorem{proof}{Proof}

\usepackage[style=alphabetic,
	citestyle=alphabetic,
	firstinits=true,
	backend=biber,
	hyperref=auto,
	sorting=nyt,
	maxcitenames=3,
	maxbibnames=99,
	minnames=3,
	arxiv=pdf,
	url=false,
	isbn=false,
	dateabbrev=true,
	datezeros=true,
	bibencoding=utf8]{biblatex}
\newbibmacro{string+doi}[1]{%
	\iffieldundef{doi}{#1}{\href{http://dx.doi.org/\thefield{doi}}{#1}}}
\DeclareFieldFormat
	{title}{\usebibmacro{string+doi}{\mkbibemph{#1}}\isdot}
\DeclareFieldFormat
	[article,inbook,incollection,inproceedings,patent,thesis,unpublished]
	{title}{\usebibmacro{string+doi}{\mkbibemph{#1}}\isdot}
\DeclareFieldFormat
	[article,inbook,incollection,inproceedings,patent,thesis,unpublished]
	{pages}{#1}
\DeclareFieldFormat
	[article]
	{journaltitle}{#1\isdot}
\DeclareFieldFormat
	[article]
	{volume}{\textbf{#1}}
\DefineBibliographyStrings{english}{pages={}}
\DeclareBibliographyDriver{article}{%
	\usebibmacro{bibindex}%
	\usebibmacro{begentry}%
	\usebibmacro{author/translator+others}%
	\setunit{\labelnamepunct}\newblock
	\usebibmacro{title}%
	\usebibmacro{journal+issuetitle}%
	\setunit*{\addcomma\space}
	\printfield{pages}
	\setunit*{\addcomma\space}
	\printfield{year}
	\usebibmacro{finentry}%
}

\newbibmacro*{journal+issuetitle}{%
	\usebibmacro{journal}%
	\setunit*{\addspace}%
	\iffieldundef{series}
	{}
	{\newunit
		\printfield{series}%
		\setunit{\addspace}}%
	\iffieldundef{volume}
		{}
		{\printfield{volume}%
		\setunit{\addspace}}%
	\setunit*{\addcolon\space}%
	\usebibmacro{issue}%
	\newunit}

\usepackage[scaled]{beramono}
\usepackage{helvet}
\usepackage[charter]{mathdesign} 
\providecommand{\ie}{i.~e.~}
\providecommand{\eg}{e.~g.~}
\providecommand{\cf}{cf.~}

\providecommand{\R}{\mathbb{R}}

\providecommand{\C}{\mathbb{C}}
\renewcommand{\C}{\mathbb{C}}

\providecommand{\T}{\mathbb{T}}
\renewcommand{\T}{\mathbb{T}}
\providecommand{\N}{\mathbb{N}}
\providecommand{\Z}{\mathbb{Z}}
\providecommand{\ii}{\mathrm{i}}
\providecommand{\e}{\mathrm{e}}
\renewcommand{\Re}{\mathrm{Re} \,}
\renewcommand{\Im}{\mathrm{Im} \,}
\providecommand{\Hil}{\mathcal{H}}

\providecommand{\eps}{\varepsilon}
\providecommand{\Cont}{\mathcal{C}}

\providecommand{\ker}{\mathrm{ker} \, }
\providecommand{\ran}{\mathrm{ran} \, }

\providecommand{\ker}{\mathrm{ker} \,}

\providecommand{\trace}{\mathrm{Tr} \,}

\providecommand{\dd}{\mathrm{d}}
\providecommand{\id}{\mathrm{id}}

\providecommand{\order}{\mathcal{O}}

\providecommand{\trace}{\mathrm{Tr}}

\providecommand{\abs}[1]{\left \lvert #1 \right \rvert}
\providecommand{\sabs}[1]{\lvert #1 \vert}

\providecommand{\norm}[1]{\left \lVert #1 \right \rVert}
\providecommand{\snorm}[1]{\lVert #1 \rVert}
\providecommand{\bnorm}[1]{\bigl \lVert #1 \bigr \rVert}
\providecommand{\Bnorm}[1]{\Bigl \lVert #1 \Bigr \rVert}
\providecommand{\scpro}[2]{\left \langle #1 , #2 \right \rangle}
\providecommand{\sscpro}[2]{\langle #1 , #2 \rangle}
\providecommand{\bscpro}[2]{\bigl \langle #1 , #2 \bigr \rangle}
\providecommand{\Bscpro}[2]{\Bigl \langle #1 , #2 \Bigr \rangle}
\providecommand{\ket}[1]{\left \vert #1 \right \rangle}

\providecommand{\bra}[1]{\left \langle #1 \right \vert}

\providecommand{\sopro}[2]{\vert #1 \rangle \langle #2 \vert}

\providecommand{\ncint}{\mathrel{{\ooalign{$\int$\cr\kern+.07em\raise.15ex\hbox{$\pmb{\scriptstyle-}$}\cr}}}}           \providecommand{\ncpartial}{\mathrel{{\ooalign{$\partial$\cr\kern+.29em\raise.79ex\hbox{$\pmb{\scriptstyle-}$}\cr}}}}

\renewcommand{\Hil}{\mathfrak{H}}

\providecommand{\Zak}{\mathcal{Z}}

\providecommand{\Msymb}{\mathcal{M}}
\providecommand{\Op}{\mathfrak{Op}}
\providecommand{\Weyl}{\sharp}

\providecommand{\Maxwell}{M}

\providecommand{\Mper}{\Maxwell_0}

\providecommand{\HperT}{\mathfrak{h}_0}

\providecommand{\BZ}{\mathbb{B}}

\providecommand{\Hoer}[1]{S^{#1}_{\rho}}
\providecommand{\Hoerm}[2]{S^{#1}_{#2}}
\providecommand{\Hoereq}[1]{S^{#1}_{\rho,\mathrm{eq}}}
\providecommand{\Hoermeq}[2]{S^{#1}_{#2,\mathrm{eq}}}
\providecommand{\Hoermper}[2]{S^{#1}_{#2,\mathrm{per}}}

\providecommand{\SemiHoereq}[1]{A S^{#1}_{\rho,\mathrm{eq}}}
\providecommand{\SemiHoermeq}[2]{A S^{#1}_{#2,\mathrm{eq}}}
\providecommand{\SemiHoerper}[1]{A S^{#1}_{\rho,\mathrm{per}}}
\providecommand{\SemiHoermper}[2]{A S^{#1}_{#2,\mathrm{per}}}

\providecommand{\Rot}{\mathbf{Rot}}

\providecommand{\Weyl}{\sharp}

\providecommand{\domainT}{\mathfrak{d}}

\providecommand{\Mext}{M^{\mathrm{ext}}}

\usepackage{units}

\bibliography{./bibliography}

\title{Derivation of Ray Optics \\ Equations in Photonic Crystals \\ Via a Semiclassical Limit}
\author{Giuseppe De Nittis${}^1$ \& Max Lein${}^2$}

\begin{document}

\maketitle
\vspace{-9mm}
\begin{center}
	${}^1$ Facultad de Matemáticas, 
	Pontificia Universidad Católica de Chile \linebreak
	Avenida Vicuña Mackenna 4860, 
	Santiago, 
	Chile \linebreak
	{\footnotesize \href{mailto:denittis@math.fau.de}{\texttt{denittis@math.fau.de}}}
	\medskip
	\\
	${}^2$ Advanced Institute of Materials Research,  
	Tohoku University \linebreak
	2-1-1 Katahira, Aoba-ku, 
	Sendai, 980-8577, 
	Japan \linebreak
	{\footnotesize \href{mailto:maximilian.lein.d2@tohoku.ac.jp}{\texttt{maximilian.lein.d2@tohoku.ac.jp}}}
\end{center}
\begin{abstract}
	In this work we present a novel approach to the ray optics limit: we rewrite the dynamical Maxwell equations in Schrödinger form and prove Egorov-type theorems, a robust \emph{semiclassical} technique. We implement this scheme for periodic light conductors, \emph{photonic crystals}, thereby making the quantum-light analogy between semiclassics for the Bloch electron and ray optics in photonic crystals rigorous. 
	One major conceptual difference between the two theories, though, is that electromagnetic fields are real, and hence, we need to add one step in the derivation to reduce it to a single-band problem. 
	Our main results, Theorem~\ref{ray_optics:thm:ray_optics} and Corollary~\ref{ray_optics:cor:phase_space_average}, give a ray optics limit for quadratic observables and, among others, apply to local averages of energy density, the Poynting vector and the Maxwell stress tensor. 
	Ours is the first rigorous derivation of ray optics equations which include all sub-leading order terms, some of which are also new to the physics literature. 
	The ray optics limit we prove applies to photonic crystals of \emph{any topological class}. 
\end{abstract}
%

\newpage
\tableofcontents

\section{Introduction} 
\label{intro}
The main idea of ray optics is to approximate full electrodynamics as given by the source-free Maxwell equations in a medium 
\begin{subequations}
	\label{intro:eqn:Maxwell}
	\begin{align}
		\left (
		\begin{matrix}
			\eps & \chi \\
			\chi^* & \mu \\
		\end{matrix}
		\right ) 
		\, \frac{\dd}{\dd t} \left (
		\begin{matrix}
			\mathbf{E} \\
			\mathbf{H} \\
		\end{matrix}
		\right ) &= \left (
		\begin{matrix}
			+ \nabla \times \mathbf{H} \\
			- \nabla \times \mathbf{E} \\
		\end{matrix}
		\right )
		,
		&& \mbox{(dynamical eqns.)}
		\label{intro:eqn:dynamical_Maxwell}
		\\
		\mathrm{Div}
		\left ( \left (
		\begin{matrix}
			\eps & \chi \\
			\chi^* & \mu \\
		\end{matrix}
		\right ) 
		\left (
		\begin{matrix}
			\mathbf{E} \\
			\mathbf{H} \\
		\end{matrix}
		\right ) \right ) &= 0 
		, 
		&& \mbox{(no sources eqns.)}
		\label{intro:eqn:source_Maxwell}
	\end{align}
\end{subequations}
by simpler hamiltonian equations of motion of the form 
\begin{subequations}
	\label{intro:eqn:naive_ray_optics}
	\begin{align}
		\dot{r} &= + \nabla_k \Omega + \order(\lambda) 
		, 
		\\
		\dot{k} &= - \nabla_r \Omega + \order(\lambda) 
		. 
	\end{align}
\end{subequations}
Here, $\mathrm{Div} = (\nabla \cdot) \otimes \id_{\C^2}$ consists of two copies of the divergence and the material weights \emph{electric permittivity} $\eps = \eps(\lambda)$, \emph{magnetic permeability} $\mu = \mu(\lambda)$ and \emph{bi-aniso\-tropic tensor} $\chi = \chi(\lambda)$ are $3 \times 3$-matrix-valued functions which describe the response of the medium to the impinging electromagnetic waves; the presence of the perturbation parameter $\lambda \ll 1$ indicates that the material weights are modulated compared to their unperturbed counterparts (see Assumption~\ref{setup:assumption:modulated_weights} for the case considered in this paper). 
While \eqref{intro:eqn:Maxwell} only describes \emph{non-gyrotropic} media where $\eps$, $\mu$ and $\chi$ are all real-valued, our ideas also apply to Maxwell's equations describing \emph{gyrotropic} media (\cf equations~\eqref{setup:eqn:generic_Maxwell_equations}). 
In both cases the material weights enter \eqref{intro:eqn:naive_ray_optics} implicitly via the \emph{dispersion relation} $\Omega(r,k)$, and indeed, one of the main tasks in justifying a ray optics limit is to determine $\Omega$ from the weights for \emph{suitable} initial states. 

The advantage of ray optics equations~\eqref{intro:eqn:naive_ray_optics} is that they provide a simpler, effective description of the propagation of light in a medium, \ie we can study solutions of an ODE to understand the behavior of a PDE. Ray optics are used in a wide variety of circumstances, and newfound applications to fields such as computer vision and image processing (see \eg \cite{Siddiqi_Tannenbaum_Zucker:hamiltonian_approach_eikonal_equation:1999,Rangarajan_Gurumoorthy:Schroedinger_Eikonal_equation:2009}) mean it still is an area of active research. One may also think of more sophisticated ray optics equations which include polarization as a classical spin degree of freedom. Instead of having to solve \eqref{intro:eqn:Maxwell} for $\bigl ( \mathbf{E}(t) , \mathbf{H}(t) \bigr )$, ray optics equations describe a light wave by its position $r$ and its \emph{wave vector} $k$, and the wave front propagates with \emph{group velocity} $\dot{r}$ along the trajectory $\bigl ( r(t) , k(t) \bigr )$. However, \emph{a priori} it is not at all clear in what sense \eqref{intro:eqn:naive_ray_optics} approximates \eqref{intro:eqn:Maxwell}, and how to quantify the error. 

The purpose of this paper is to derive the ray optics limit in a \emph{novel way} by rewriting the dynamical Maxwell equations~\eqref{intro:eqn:dynamical_Maxwell} in Schrödinger form and proving an Egorov theorem, a well-known and robust semiclassical technique. While most derivations of ray optics (see \eg \cite[Chapter~5.4]{Someda:electromagnetic_waves:1998}, \cite[Chapter~2]{Perlick:ray_optics:2000} and \cite{Onoda_Murakami_Nagaosa:geometrics_optical_wave-packets:2006}) employ what would be called “semiclassical wavepacket methods” in the context of quantum mechanics, our technique does not rely on the localization of $(\mathbf{E},\mathbf{H})$ around some $(r_0,k_0)$ in phase space. 

Instead, we will \emph{prove a ray optics limit for} a class of \emph{observables} that includes local averages of the field energy, the Poynting vector and components of the Maxwell stress tensor. Conceptually, there are two major differences to quantum mechanics we will need to deal with: 
\begin{enumerate}[(i)]
	\item Electromagnetic fields — unlike quantum mechanical wave functions — are \emph{real}. 
	\item Observables in electromagnetism are not selfadjoint operators, but functionals on the fields. 
\end{enumerate}
The reason the reality of electromagnetic fields complicates matters is that real electromagnetic fields are \emph{necessarily} a linear combination of states associated to $N$ positive and $N$ negative frequency bands, \ie at least \emph{two}. 
While there are multiband semiclassical techniques available \cite{Bellissard_Rammal:algebraic_semi-classical_approach:1990,Littlejohn_Flynn:geometric_phases_Bohr_Sommerfeld_quantization_multicomponent_wavefunctions:1991}, we rely on a result by Teufel and Stiepan \cite{Stiepan_Teufel:semiclassics_op_valued_symbols:2012} which works only for single bands. 
Because of the reality of electromagnetic fields, we first use 
symmetry arguments to reduce everything to the positive frequency bands (\cf Proposition~\ref{ray_optics:prop:rewriting_electromagnetic_observables}), and then apply the single-band technique from \cite{Stiepan_Teufel:semiclassics_op_valued_symbols:2012}. We do that by projecting the real electromagnetic field $(\mathbf{E},\mathbf{H})$ onto the positive frequencies via the orthogonal projection $P_{+,\lambda}$. The original real electromagnetic field $(\mathbf{E},\mathbf{H})$ can be recovered by taking the real part of $P_{+,\lambda} (\mathbf{E},\mathbf{H})$. Hereinafter, it is useful to think of the real part 
\begin{align*}
	\Re := \tfrac{1}{2} \bigl ( \id + C \bigr )
\end{align*}
as an $\R$-linear projection; any operator which commutes with $C$ also commutes with $\Re$. Just like in quantum mechanics, not all quantum observables have a good semiclassical limit; The same is true in electromagnetism. Our results hold for “quadratic” observables which come in pairs, an electromagnetic observable $\mathcal{F} : L^2(\R^3,\C^6) \longrightarrow \R$ (a functional on the electromagnetic field) and a ray optics observable $f$ (\ie a function of $(r,k)$). The former can be seen as the “quantum expectation value” of the pseudodifferential operator associated to $f$ with respect to the electromagnetic field, 
\begin{align}
	\mathcal{F}[(\mathbf{E},\mathbf{H})] &= \mathbb{E}_{f} [(\mathbf{E},\mathbf{H})]
	\notag \\
	:\negmedspace &= 2 \, \Re \, \bscpro{P_{+,\lambda} (\mathbf{E},\mathbf{H}) \,}{\, \Op_{\lambda}^{S \Zak}(f) \, P_{+,\lambda} (\mathbf{E},\mathbf{H})}_{L^2_{W_{+,\lambda}}(\R^3,\C^6)}
	. 
	\label{intro:eqn:ray_optics_observables}
\end{align}
We will explain the notation in detail later in Section~\ref{ray_optics:observables}. 

Now assume the fields are associated to a given (positive) frequency band $\omega(k)$. More specifically, $\omega(k)$ determines a projection $\Pi_{+,0}$ (different from $P_{+,\lambda}$), and the electromagnetic fields $(\mathbf{E},\mathbf{H})$ of interest lie in $\Re \ran \Pi_{+,0}$. Then for observables of the form~\eqref{intro:eqn:ray_optics_observables}, we can approximate the observable at time $t$ by transporting $f$ along the ray optics flow $\Phi^0$, 
\begin{align}
	\mathcal{F} \bigl [ \bigl ( \mathbf{E}(t) , \mathbf{H}(t) \bigr ) \bigr ] &= \mathbb{E}_{f \circ \Phi^0_t}[(\mathbf{E},\mathbf{H})] + \order(\lambda)
	\notag \\
	&
	= 2 \, \Re \, \int_{\R^3} \dd r \int_{\R^3} \dd k \, \bigl ( f \circ \Phi^0_t \bigr )(r,k) \; \mathrm{w}_{P_{+,\lambda} (\mathbf{E},\mathbf{H})}(r,k) + \order(\lambda)
	. 
	\label{intro:eqn:ray_optics_limit}
\end{align}
The dispersion $\Omega(r,k) = \tau(r)^2 \, \omega(k)$ which enters the ray optics equations~\eqref{intro:eqn:naive_ray_optics} consist of a factor $\tau(r)^2$ that is due to the slow modulation and the periodic frequency band function $\omega(k)$. Moreover, we can express $\mathbb{E}_{f \circ \Phi^0_t}(\mathbf{E},\mathbf{H})$ as a phase space average where we integrate $f \circ \Phi^0_t$ against the Wigner transform $\mathrm{w}_{P_{+,\lambda} (\mathbf{E},\mathbf{H})}$ of the positive frequency part of $(\mathbf{E},\mathbf{H})$ at time $0$ (\cf Corollary~\ref{ray_optics:cor:phase_space_average}). 

Our two main results, Theorem~\ref{ray_optics:thm:ray_optics} and Corollary~\ref{ray_optics:cor:phase_space_average}, are in fact stronger than \eqref{intro:eqn:ray_optics_limit} because after careful analysis we have been able us to reduce the error by one order of magnitude to $\order(\lambda^2)$. This is done by modifying the ray optics flow $\Phi^{\lambda} = \Phi^0 + \order(\lambda)$, the projection onto the relevant states $\Pi_{+,\lambda} = \Pi_{+,0} + \order(\lambda)$ and potentially also the ray optics observable $f$. 

Apart from the local energy density, our results also cover local averages of the Poynting vector, the field amplitudes and the components of the Maxwell stress tensor (see Section~\ref{ray_optics:examples}). Note that the error term in \eqref{intro:eqn:ray_optics_limit} can be estimated uniformly in $(\mathbf{E},\mathbf{H})$ as long as we keep the field energy fixed. Our approach overcomes two major limitation of “wavepacket techniques”: Mathematically, these are notoriously hard to justify. And physically, given that they depend on a judicious choice of initial state, it is hard to go beyond leading order and compute the $\order(\lambda)$ corrections which often contain novel physical effects.

We will illustrate how to implement a ray optics limit via semiclassical techniques for \emph{photonic crystals}, periodically patterned light conductors. Just as in case of the Bloch electron the periodic structure modifies the dispersion relation: whereas in quantum mechanics $\tfrac{1}{2m} k^2 + V(r)$ has to be replaced by an energy band function $E_n(k) + V(r)$, the so-called semiclassical limit of the Bloch electron (see \eg \cite{PST:effective_dynamics_Bloch:2003,DeNittis_Lein:Bloch_electron:2009} and references therein), also in case of photonic crystals $c \abs{k}$ has to be substituted by $\Omega(r,k) = \tau^2(r) \, \omega_n(k) + \order(\lambda)$ where $\omega_n(k)$ is a frequency band function and the modulation $\tau(r)^2$ is due to the external perturbation. And just like in the case of the Bloch electron, we rely on the presence of a \emph{spectral gap}, \ie $\omega_n(k)$ is a non-degenerate frequency band which does not intersect or merge with other bands. The choice of band not only enters the dispersion relation, but also determines the subspace $\Re \ran \Pi_{+,\lambda} \ni (\mathbf{E},\mathbf{H})$ on which \eqref{intro:eqn:ray_optics_limit} holds. Moreover, finding the form of the $\order(\lambda)$ terms in \eqref{intro:eqn:naive_ray_optics} is crucial, because these first-order corrections are believed to explain geometric and topological effects \cite{Raghu_Haldane:quantum_Hall_effect_photonic_crystals:2008,Bliokh_Bliokh:Berry_curvature_optical_Magnus_effect:2004,Onoda_Murakami_Nagaosa:geometrics_optical_wave-packets:2006}. 

Our first main result, Theorem~\ref{ray_optics:thm:ray_optics}, rigorously establishes the ray optics limit for two classes of observables, scalar and non-scalar quadratic observables (\cf Definition~\ref{ray_optics:defn:admissible_observables}). Apart from generic conditions on the material weights, no restrictions such as topological triviality of the frequency band $\omega_n$ or the presence of symmetries needs to be imposed, in the parlance of \cite{DeNittis_Lein:symmetries_Maxwell:2014,DeNittis_Lein:symmetries_electromagnetism:2016} our main Theorem~\ref{ray_optics:thm:ray_optics} applies to photonic crystals of any topological class. We follow the ideas of Stiepan and Teufel, but it is necessary to generalize their procedure to include \emph{non-scalar} observables to cover prominent examples such as the Poynting vector and the Maxwell stress tensor. 

Up until this work the exact form of the ray optics equations had been an open problem, even on the level of physics the exact form of the ray optics equations had not yet been established: Raghu and Haldane proposed their ray optics equations \emph{by analogy} to the corresponding quantum system, the Bloch electron. 
Subsequently, only three works attempted to derive ray optics equations systematically: Onoda et al \cite{Onoda_Murakami_Nagaosa:geometrics_optical_wave-packets:2006} used variational techniques developed by Sundaram and Niu \cite{Sundaram_Niu:wave_packet_dynamics_Bloch_electron:1999}, and  their ray optics equation differ to sub-leading order (where all topological contributions enter) from those of Raghu and Haldane. The second work is by Esposito and Gerace \cite{Esposito_Gerace:photonic_crystals_broken_TR_symmetry:2013} who derive only the equation for $\dot{r}$ via standard perturbation theory. None of these equations coincide with the equations we have found, though (\cf Proposition~\ref{ray_optics:prop:ray_optics_limit}). The only rigorous work we are aware of is \cite{Allaire_Palombaro_Rauch:diffractive_Bloch_wave_packets_Maxwell:2012}, and they justify the eikonal approximation via a multiscale WKB ansatz. However, Allaire et al crucially assume in \cite[Hypothesis~1.1]{Allaire_Palombaro_Rauch:diffractive_Bloch_wave_packets_Maxwell:2012} that the perturbation of the material weights is a \emph{second}-order effect, \eg $\eps(\lambda) = \eps(0) + \order(\lambda^2)$, meaning the perturbation is of the same order of magnitude as the error in \eqref{intro:eqn:ray_optics_limit}. We refer to Section~\ref{discussion:previous_results} for a more in-depth discussion of these previous results and a comparison to ours. 

The equations we have derived are \emph{one}-band equations, and in principle, one may wonder whether \emph{degenerate} bands are a generic feature of a certain class of photonic crystals? Fortunately, for most the answer is no. There are two symmetries which lead to globally degenerate bands, and neither of them are present in most photonic crystals: 
\begin{enumerate}[(i)]
	\item Light comes in two chiralities, left- and right-hand circularly polarized light, and in many materials the light dynamics are independent of the polarization. The associated symmetry operator can be written as a function of $- \ii \nabla$ \cite[equation~(22)]{Bliokh_Kivshar_Nori:magneto_electric_effects:2014}, and hence, position-dependent material weights break this symmetry. Nevertheless, in periodic waveguide arrays where $\eps$ and $\mu$ are scalar, $\chi = 0$, and the contrast is very low (of the order of $10^{-4} \sim 10^{-3}$ \cite{Longhi:quantum_optical_ananlog_review:2009,Rechtsman_Zeuner_et_al:photonic_topological_insulators:2013}), the degeneracy of the two polarization states is broken only at the subleading order. Here, we reckon one needs to include a classical spin degree of freedom in the ray optics equations using semiclassical techniques for a particle with spin \cite{Gat_Lein_Teufel:semiclassical_dynamics_particle_spin:2013}. 
	\item Materials where the roles of electric and magnetic field are “symmetric” possess the “dual symmetry” \cite{Bliokh_Bekshaev_Nori:dual_electromagnetism:2013}; this symmetry generates “rotations”, 
	\begin{align*}
		(\mathbf{E},\mathbf{H}) \mapsto \bigl ( \cos \alpha \, \mathbf{E} + \sin \alpha \, \mathbf{H} \; , \; - \sin \alpha \, \mathbf{E} + \cos \alpha \, \mathbf{H} \bigr )
		, 
	\end{align*}
	mixing electric and magnetic fields. Periodic light conductors made up of dual symmetric materials exist: in case $\eps = c \, \mu$ and $\chi = 0$ (\eg vacuum or certain YIG 2d photonic crystals \cite{Pozar:microwave_engineering:1998,Wang_et_al:edge_modes_photonic_crystal:2008}) each band is two-fold degenerate due to this dual symmetry. 
\end{enumerate}

\paragraph{Outline} 
\label{intro:outline}
The essential ingredient for the derivation of ray optics is to bring the Maxwell equations~\eqref{intro:eqn:Maxwell} in Schrödinger form and to extend them to include gyrotropic media, something which we explain in Section~\ref{setup}. There we also introduce other necessary objects and notation, and state all assumptions. Because the adiabatically perturbed Maxwell operator (which takes the place of the hamilton operator) is a \emph{pseudodifferential} operator \cite[Theorem~1.3]{DeNittis_Lein:adiabatic_periodic_Maxwell_PsiDO:2013}, standard semiclassical techniques can be applied to yield ray optics equations. Those approximate full electrodynamics in the sense of an Egorov theorem (Section~\ref{ray_optics}), the proof of which is the content of Section~\ref{egorov}. Our work closes with a discussion of our results in Section~\ref{discussion}. Some auxiliary results are put into an appendix. 

\paragraph{Acknowledgements} 
\label{intro:acknowledgements}
G.~D.~research is supported by the  grant \emph{Iniciaci\'{o}n en Investigaci\'{o}n 2015} - $\text{N}^{\text{o}}$ 11150143 funded  by FONDECYT. We would like to take the opportunity to thank Stefan Teufel for useful feedback and friendly discussions. 

\section{Schrödinger formalism of the Maxwell equations} 
\label{setup}
Let us proceed to clearly define the mathematical problem. For the purpose of this paper we restrict ourselves to \emph{linear}, \emph{lossless} media meaning that the \emph{material weights} 
\begin{align}
	W^{-1}(x) := \left (
	\begin{matrix}
		\eps(x) & \chi(x) \\
		\chi^*(x) & \mu(x) \\
	\end{matrix}
	\right ) 
	\label{setup:eqn:Winverse}
\end{align}
which quantify the response of the medium are frequency-\emph{in}dependent and take values in the \emph{hermitian} $6 \times 6$-matrices. 
We will always make the following assumptions:
\begin{assumption}[Material weights]\label{setup:assumption:material_weights}
	Assume that $W^{-1} \in L^{\infty} \bigl ( \R^3 , \mathrm{Mat}_{\C}(6) \bigr )$ is positive, selfadjoint, bounded and has a bounded inverse $W$. We say that the weights are \emph{real} if and only if $[C , W] = 0$ where $C$ denotes complex conjugation. 
\end{assumption}
Throughout the main body of the paper, we will make a conscious attempt to cut down on technical details which are not necessary to understand the strategy of the proofs.

\subsection{Materials with real material weights} 
\label{setup:real}
Let us start by considering light conductors whose material weights are real (as opposed to complex). Here, the reality of electromagnetic fields is preserved by Maxwell's equations~\eqref{intro:eqn:Maxwell} — which simplifies the mathematical description. The case where $W \neq \overline{W}$ is complex will be discussed in Section~\ref{setup:complex}. In both cases the first goal is to rewrite Maxwell's equations in Schrödinger form as that allows us to adapt techniques initially developed for quantum mechanics and apply them to classical electromagnetism.

\subsubsection{First-order Schrödinger framework of electromagnetism} 
\label{setup:real:Schroedinger_framework}
As our starting point we recast the Maxwell equations as a \emph{Schrödinger equation} 
\begin{align}
	\ii \partial_t \Psi = M_w \Psi
	\label{setup:eqn:Maxwell_Schroedinger_eqn}
\end{align}
by multiplying both sides of \eqref{intro:eqn:dynamical_Maxwell} by $\ii \, W$ and restricting oneself to electromagnetic fields $\Psi = (\mathbf{E},\mathbf{H}) \in L^2(\R^3,\C^6)$ which satisfy \eqref{intro:eqn:source_Maxwell} in the distributional sense. Based on this precise formulation of the “quantum-light analogy” we can systematically adapt techniques from applied mathematics and quantum physics to classical electromagnetism. Here, the electromagnetic field $\Psi = (\mathbf{E},\mathbf{H})$ plays the role of the wave function and the \emph{Maxwell operator} 
\begin{align}
	M_w := W \, \Rot 
	= W \, \left (
	\begin{matrix}
		0 & + \ii \nabla^{\times} \\
		- \ii \nabla^{\times} & 0 \\
	\end{matrix}
	\right )
	\label{setup:eqn:Maxwell_operator}
\end{align}
takes the place of the Schrödinger operator. $\nabla^{\times} \mathbf{E} = \nabla \times \mathbf{E}$ is the curl for vector fields on $\R^3$, and we will frequently make use of this notation to connect the matrix 
\begin{align*}
	v^{\times} \psi := \left (
	\begin{matrix}
		0 & -v_3 & +v_2 \\
		+v_3 & 0 & -v_1 \\
		-v_2 & +v_1 & 0 \\
	\end{matrix}
	\right ) \left (
	\begin{matrix}
		\psi_1 \\
		\psi_2 \\
		\psi_3 \\
	\end{matrix}
	\right )
	= v \times \psi
\end{align*}
to any vectorial quantity $v$ such as the canonical basis vectors $e_j$, $ j = 1 , 2 , 3$, of $\C^3$. Moreover, the the Maxwell operator is selfadjoint \cite[Theorem~2.1]{DeNittis_Lein:adiabatic_periodic_Maxwell_PsiDO:2013} on the Hilbert space $\Hil_w$ one obtains by endowing the complex Banach space $L^2(\R^3,\C^6)$ with the weighted \emph{energy scalar product} $\bscpro{\Psi}{\Phi}_w := \bscpro{\Psi}{W^{-1} \Phi}$. Consequently, we are able to reach into the rich toolbox from the theory of selfadjoint operators. In particular, the time evolution group $\e^{- \ii t M_w}$ exists and is unitary with respect to $\scpro{\cdot \,}{\cdot}_w$. For the case of real material weights where $W$ commutes with complex conjugation $C$, the complexification of electromagnetic fields is just a matter of convenience, real electromagnetic fields are recovered by taking the real part of the solution afterwards (\cf \cite[Section~4]{DeNittis_Lein:sapt_photonic_crystals:2013}). 

On the level of operators, $C$ gives rise to an \emph{even particle-hole-type symmetry} because $C$ is anti-linear, $C^2 = + \id$ and it \emph{anti}commutes with the Maxwell operator, 
\begin{align}
	C \, M \, C = - M
	. 
	\label{setup:eqn:real_weights_C_PH_symmetry}
\end{align}
As particle-hole symmetries \emph{commute} with the time-evolution group $\e^{- \ii t M}$, we conclude that real fields remain real under the evolution, 
\begin{align}
	\bigl [ \e^{- \ii t M} , \Re \bigr ] = 0 
	. 
	\label{setup:eqn:time_evolution_Re_commute}
\end{align}
This leads to a symmetry in the band spectrum: if $\varphi_n(k)$ is an eigenfunction of $M(k)$ to $\omega_n(k)$, then $\overline{\varphi_n(-k)}$ is an eigenfunction of $M(k)$ to $-\omega_n(-k)$, \ie we obtain a \emph{pairing of frequency bands}
\begin{align}
	\bigl ( \varphi_n(k) \, , \, \omega_n(k) \bigr ) 
	\; \; 
	\longleftrightarrow
	\; \; 
	\bigl ( \overline{\varphi_n(-k)} \, , \, - \omega_n(-k) \bigr ) 
	. 
	\label{setup:eqn:band_pairing}
\end{align}
Hence, the frequency band spectrum of $M(k)$ is symmetric under inversion at $k = 0$ (\cf Figure~\ref{setup:figure:band_spectrum}). 

\begin{figure}
	\centering
	\resizebox{!}{50mm}{\includegraphics{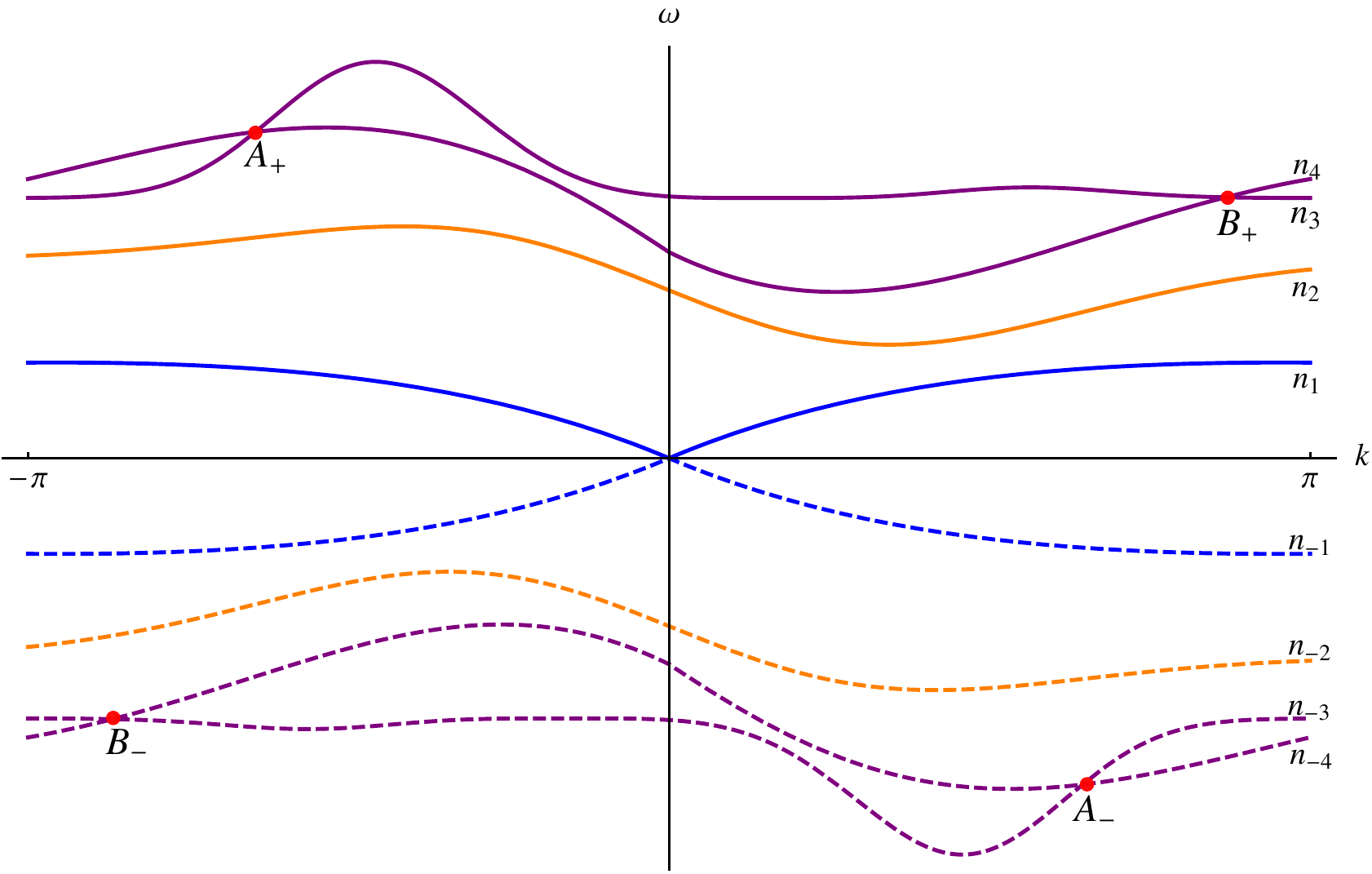}}
	\caption{One-dimensional representation of a frequency band spectrum of a photonic crystal. The particle-hole symmetry manifests itself as a point symmetry of the spectrum. Time-reversal symmetries, on the other hand, lead to the spectrum being even \ie $\omega_n(-k) = \omega_n(k)$ then holds for all frequency band functions. }
	\label{setup:figure:band_spectrum}
\end{figure}
%

\subsubsection{Adiabatically perturbed photonic crystals} 
\label{setup:real:photonic_crystals}
We are interested in the propagation of light in adiabatically perturbed photonic crystals where the periodic material weights are perturbed in a specific manner: 
\begin{assumption}[Slowly modulated weights]\label{setup:assumption:modulated_weights}
	Suppose the material weights are of the form $W_{\lambda}(x) = S^{-2}(\lambda x) \, W(x)$ where 
	\begin{enumerate}[(i)]
		\item the periodic contribution $W$ satisfies Assumption~\ref{setup:assumption:material_weights} and is periodic with respect to some lattice $\Gamma \cong \Z^3$, and 
		\item the slow modulation $S$ is either of the form $S(\lambda x) := \tau^{-1}(\lambda x)$ when $\chi \neq 0$ or 
		\begin{align}
			S(\lambda x) := \left (
			\begin{matrix}
				\tau_{\eps}^{-1}(\lambda x) \; \id_{\C^3} & 0 \\
				0 & \tau_{\mu}^{-1}(\lambda x) \; \id_{\C^3} \\
			\end{matrix}
			\right )
			\label{intro:eqn:modulated_weights_v1}
		\end{align}
		in case $\chi = 0$. 
	\end{enumerate}
	The functions $\tau , \tau_{\eps} , \tau_{\mu} \in \Cont^{\infty}_{\mathrm{b}}(\R^3)$ are always assumed to be positive, $\tau(0) = \tau_{\eps}(0) = \tau_{\mu}(0) = 1$ and bounded away from $0$ and $+\infty$. In case of modulation \eqref{intro:eqn:modulated_weights_v1}, one defines $\tau := \sqrt{\tau_{\eps} \, \tau_{\mu}}$. 
\end{assumption}
We will use the index $\lambda$ systematically, \eg $\Hil_{\lambda} := \Hil_{w_{\lambda}}$. Objects with the index $0$ denote the periodic case, and we can write $M_{\lambda} = S(\lambda \hat{x})^{-2} \, \Mper$ where $S(\lambda \hat{x})$ denotes the operator of multiplication by $S(\lambda x)$. We will use this notation for multiplication operators also for other variables.  

The periodic Maxwell operator $\Mper \cong \int_{\BZ}^{\oplus} \dd k \, \Mper(k)$, 
\begin{align*}
	\Mper(k) := W(\hat{y}) \, \left (
	\begin{matrix}
		0 & - (- \ii \nabla_y + k)^{\times} \\
		+ (- \ii \nabla_y + k)^{\times} & 0 \\
	\end{matrix}
	\right )
	, 
\end{align*}
fibers in crystal momentum $k \in \BZ$ ($\BZ \simeq \T^3$ being the Pontryagin dual of the lattice $\Gamma$, usually referred to as Brillouin zone) via the \emph{Zak transform} 
\begin{align}
	(\Zak \Psi)(k,y) := \sum_{\gamma \in \Gamma} \e^{- \ii k \cdot (y + \gamma)} \, \Psi(y + \gamma) 
	, 
	\label{setup:eqn:Zak_transform}
\end{align}
and apart from essential spectrum at $\omega = 0$ due to unphysical gradient fields, $\sigma \bigl ( \Mper(k) \bigr ) = \bigcup_{n \in \Z} \{ \omega_n(k) \}$ is purely discrete and consists of frequency bands \cite[Theorem~1.4]{DeNittis_Lein:adiabatic_periodic_Maxwell_PsiDO:2013}. With the exception of the ground state bands (which have a linear dispersion around $k = 0$ and $\omega = 0$), all Bloch functions $\varphi_n$ are locally analytic away from frequency band crossings. Note that unlike periodic Schrödinger operators, the Maxwell operator is not bounded from below. In fact, symmetries such as complex conjugation induce relations between bands of different signs \cite{DeNittis_Lein:symmetries_Maxwell:2014}: if complex conjugation $C$ commutes with the material weights (\ie $W$ is real), then the periodic Maxwell operator satisfies $C \, \Mper(k) \, C = - \Mper(-k)$. Consequently, if $\varphi_n(k)$ is an eigenfunction of $\Mper(k)$ to $\omega_n(k)$, then $C \varphi_n(-k)$ is an eigenfunction of $\Mper(k)$ to $- \omega_n(-k)$. Such pairings of \emph{twin bands} will become crucial to understanding the ray optics limit of real states, because $C \varphi_n(-k) \not\propto \varphi_n(k)$ implies these are eigenfunctions to distinct eigenvalues of $\Mper(k)$. Put another way, single bands cannot support real states (\cf discussion in \cite[Section~4.1]{DeNittis_Lein:sapt_photonic_crystals:2013}), a fact which will be discussed further in Section~\ref{setup:reduction_omega_g_0}. 

\subsubsection{Auxiliary representations} 
\label{setup:real:representations}
Our choice of representation exploits (i) the periodicity and (ii) gets rid of the $\lambda$-dependence of the Hilbert spaces. Just like in quantum mechanics, a change of representation is mitigated by a unitary map. The  Zak transform $\Zak : \Hil_{\lambda} \longrightarrow \Zak \Hil_{\lambda}$ defined in \eqref{setup:eqn:Zak_transform} above makes use of the periodicity. 

In a second step, we use the unitary $S(\ii \lambda \nabla_k) : \Zak \Hil_{\lambda} \longrightarrow \Zak \Hil_0$ to map the problem onto the (fibered) Hilbert space of the unperturbed, periodic system (\cf also \cite[Section~2.2]{DeNittis_Lein:adiabatic_periodic_Maxwell_PsiDO:2013}). And because the unperturbed weights $W$ are periodic, 
\begin{align*}
	\Zak \Hil_0 \simeq L^2(\BZ) \otimes \HperT
\end{align*}
decomposes into the “slow” space $L^2(\BZ)$ and the “fast” space $\HperT$ which is defined as $L^2(\T^3,\C^6)$ endowed with the scalar product 
\begin{align*}
	\scpro{\varphi}{\psi}_{\HperT} := \scpro{\varphi}{W^{-1} \psi}_{L^2(\T^3,\C^6)} 
	. 
\end{align*}
Alternatively, we could have reversed the order of the transformations because $\Zak \, S(\lambda \hat{x}) = S(\ii \lambda \nabla_k) \, \Zak$. 

\subsubsection{The Maxwell operator as a $\Psi$DO} 
\label{setup:real:Maxwell_operator_PsiDO}
The last ingredient we need is that the Maxwell operator 
\begin{align}
	M_{\lambda} = \Op_{\lambda}^{S \Zak}(\Msymb_{\lambda}) :\negmedspace &= \Zak^{-1} \, S(\ii \lambda \nabla_k)^{-1} \; \Op_{\lambda}(\Msymb_{\lambda}) \; S(\ii \lambda \nabla_k) \, \Zak 
	\label{setup:eqn:M_lambda_PsiDO}
	\\
	&= S(\lambda \hat{x})^{-1} \, \Zak^{-1} \; \Op_{\lambda}(\Msymb_{\lambda}) \; \Zak \, S(\lambda \hat{x}) 
	\notag 
\end{align}
can also be seen as a \emph{pseudodifferential} operator (\cf \cite[Theorem~1.3]{DeNittis_Lein:adiabatic_periodic_Maxwell_PsiDO:2013}) associated to 
\begin{align}
	\Msymb_{\lambda}(r,k) &= \Msymb_0(r,k) + \lambda \, \Msymb_1(r,k)
	\notag \\
	:\negmedspace &= \tau^2(r) \, \Mper(k) - \lambda \, \tau^2(r) \, \frac{\ii}{2} \,  \bigl ( \nabla_r \ln \tfrac{\tau_{\eps}}{\tau_{\mu}} \bigr )(r) \cdot \Sigma
	\, 
	\label{setup:eqn:symbol_Maxwellian}
\end{align}
where $\Sigma := \bigl ( \Sigma_1 , \Sigma_2 , \Sigma_3 \bigr )$ is an operator-valued vector with components 
\begin{align*}
	\Sigma_j := W \, \left (
	\begin{matrix}
		0 & e_j^{\times} \\
		e_j^{\times} & 0 \\
	\end{matrix}
	\right )
	. 
\end{align*}
The equivariance of the operator-valued function 
\begin{align}
	\Msymb_{\lambda}(r,k - \gamma^*) = \e^{+ \ii \gamma^* \cdot \hat{y}} \, \Msymb_{\lambda}(r,k) \, \e^{- \ii \gamma^* \cdot \hat{y}} 
	, 
	&&
	\forall r,k \in \R^3 
	, \; 
	\gamma^* \in \Gamma^*
	, 
	\label{setup:eqn:equivariance}
\end{align}
with respect to translations in the dual lattice $\Gamma^*$ ensures that its Weyl quantization 
\begin{align}
	\Op_{\lambda}(\Msymb_{\lambda}) &:= \frac{1}{(2\pi)^6} \int_{\R^3} \dd r \int_{\R^3} \dd k \int_{\R^3} \dd r' \int_{\R^3} \dd k' \; \e^{+ \ii (k' \cdot r - r' \cdot k)} \, 
	\cdot \notag\\
	&\qquad \qquad \qquad \qquad \qquad \qquad \cdot 
	\Msymb_{\lambda}(r,k) \; \e^{- \ii (k \cdot (\ii \lambda \nabla_k) - r \cdot \hat{k})} 
	\label{setup:eqn:definition_Op_lambda}
\end{align}
associated to the slow variables $r \mapsto \ii \lambda \nabla_k$ and $k \mapsto \hat{k}$ (multiplication by $k$) defines an equivariant operator on $\Zak \Hil_0$ (\cf Appendix~\ref{appendix:PsiDOs} and references therein). Note that while the bi-anisotropic tensor $\chi$ is absent in \cite{DeNittis_Lein:adiabatic_periodic_Maxwell_PsiDO:2013}, the results there immediately generalize: in case $\chi \neq 0$, the modulation $S(r) = \tau^{-1}(r)$ is scalar and a quick computation yields $\Msymb_{\lambda} = S^{-1} \Weyl \Mper(\, \cdot \,) \Weyl S^{-1} = \tau^2 \, \Mper(\, \cdot \,)$ agrees with \eqref{setup:eqn:symbol_Maxwellian} after setting $\tau_{\eps} = \tau_{\mu}$. 
\begin{remark}[Notation used here compared to \cite{DeNittis_Lein:adiabatic_periodic_Maxwell_PsiDO:2013,DeNittis_Lein:sapt_photonic_crystals:2013}]
	In an attempt to unburden the notation, we deviate from our earlier works. For instance, $M_{\lambda}$ as given by equation~\eqref{setup:eqn:Maxwell_operator} is a selfadjoint operator on $\Hil_{\lambda}$, and corresponds to $\mathbf{M}_{\lambda}$ in \cite{DeNittis_Lein:adiabatic_periodic_Maxwell_PsiDO:2013,DeNittis_Lein:sapt_photonic_crystals:2013}, while $\Msymb_{\lambda}$ from \eqref{setup:eqn:symbol_Maxwellian} above refers to the same symbol as in \cite[Corollary~4.3]{DeNittis_Lein:adiabatic_periodic_Maxwell_PsiDO:2013}. 
\end{remark}
%

\subsection{Materials with complex material weights} 
\label{setup:complex}
When the material weights are complex, additional considerations are necessary to connect mathematics and physics. While the details are somewhat tedious, they are crucial for the definition of \emph{physically meaningful} Maxwell equations in gyrotropic media. Later on in Section~\ref{setup:reduction_omega_g_0} we show how to do away with most of the notational baggage by writing $(\mathbf{E},\mathbf{H}) = \Re \Psi$ as the real part of a complex wave. Essentially, one has two choices: 
\begin{enumerate}[(i)]
	\item Keep Maxwell's equations as given by equations~\eqref{intro:eqn:Maxwell}, and accept that the solution $\bigl ( \mathbf{E}(t) , \mathbf{H}(t) \bigr )$ acquires a non-vanishing imaginary part even if the initial condition is real. 
	\item Alternatively, we modify equations~\eqref{intro:eqn:Maxwell} so as to ensure that solutions $\bigl ( \mathbf{E}(t) , \mathbf{H}(t) \bigr )$ remain real if the initial conditions are. 
\end{enumerate}
These two approaches lead to distinct physical predictions: according to our considerations in \cite{DeNittis_Lein:symmetries_Maxwell:2014} the presence of a chiral-type symmetry in approach~(i) predicts the existence of counter propagating edge modes in gyrotropic two-dimensional photonic crystals — in contradiction to what was observed in experiment \cite{Wang_et_al:unidirectional_backscattering_photonic_crystal:2009}. Nevertheless, approach~(i) was widely used \emph{implicitly} to describe materials with complex weights \cite[Section~7]{DeNittis_Lein:symmetries_electromagnetism:2016}. 

Approach~(ii) yields physically meaningful equations by baking a particle-hole symmetry akin to \eqref{setup:eqn:real_weights_C_PH_symmetry} into the model. While this seems \emph{ad hoc}, the equations we propose can be derived from the linear Maxwell equations (see \cite{DeNittis_Lein:symmetries_electromagnetism:2016} and references therein). 

The essential ingredient is that real electromagnetic fields 
\begin{align*}
	(\mathbf{E},\mathbf{H}) = \tfrac{1}{2} \bigl ( \Psi_+ + \Psi_- \bigr ) = \tfrac{1}{2} \bigl ( \Psi_+ + \overline{\Psi_+} \bigr )
\end{align*}
are necessarily a linear combination of \emph{complex} waves that come in \emph{positive-negative frequency pairs}. It turns out that this symmetry is preserved even when the weights $W_{+,\lambda} = W_{\lambda}$ are \emph{complex} (\ie the medium is \emph{gyrotropic}), because negative frequency complex waves are subjected to the complex conjugate weights $W_{-,\lambda} = \overline{W_{+,\lambda}}$. Put another way, there are \emph{two} sets of Maxwell equations of the form~\eqref{intro:eqn:Maxwell}, one for \emph{complex} waves for $\omega > 0$, the other for $\omega < 0$ with complex conjugate weights, namely 
\begin{subequations}
	\label{setup:eqn:generic_Maxwell_equations}
	\begin{align}
		W_{\pm,\lambda} \, \partial_t \Psi_{\pm}(t) &= \left (
		\begin{smallmatrix}
			0 & + \nabla^{\times} \\
			- \nabla^{\times} & 0 \\
		\end{smallmatrix}
		\right ) \, \Psi_{\pm}(t) 
		,
		&& \mbox{(dynamical eqns.)}
		\label{setup:eqn:generic_Maxwell_equations:dynamics}
		\\
		\mathrm{Div} \, W_{\pm,\lambda} \Psi_{\pm}(t) &= 0 
		. 
		&& \mbox{(no sources eqns.)}
		\label{setup:eqn:generic_Maxwell_equations:no_sources}
	\end{align}
\end{subequations}
This is a \emph{bona fide} extension of equations~\eqref{intro:eqn:Maxwell}, because if $W_+ = W_-$ then these two sets of equations coincide. The Maxwell operator 
\begin{align}
	M_{\lambda} := M_{+,\lambda} \oplus M_{-,\lambda} := \Mext_{+,\lambda} \big \vert_{\ran P_{+,\lambda}} \oplus \Mext_{-,\lambda} \big \vert_{\ran P_{-,\lambda}}
	\label{setup:eqn:generic_Maxwell_operator}
\end{align}
which enters the Schrödinger-type equation 
\begin{align}
	\ii \partial_t \Psi(t) = M_{\lambda} \Psi(t) 
	, 
	\qquad \qquad 
	\Psi(0) = \Phi \in \Hil_{\lambda} 
	, 
	\label{setup:eqn:generic_Maxwell_Schroedinger_equation}
\end{align}
is then the direct sum of positive/negative frequency contributions 
\begin{align}
	\Mext_{\pm,\lambda} := W_{\pm,\lambda} \, \Rot 
	\label{setup:eqn:extended_Maxwell_operator}
\end{align}
restricted to positive/negative frequency subspaces which are the ranges of the projections 
\begin{align*}
	P_{\pm,\lambda} := 1_{(0,+\infty)}(\pm \Mext_{\pm,\lambda}) 
\end{align*}
defined by functional calculus. The relevant Hilbert space 
\begin{align}
	\Hil_{\lambda} :\negmedspace &= \Hil_{+,\lambda} \oplus \Hil_{-,\lambda} 
	:= \ran P_{+,\lambda} \oplus \ran P_{-,\lambda} 
	\notag \\
	&\subset L^2_{W_{+,\lambda}}(\R^3,\C^6) \oplus L^2_{W_{-,\lambda}}(\R^3,\C^6) 
	\label{setup:eqn:gyrotropic_Hil_lambda}
\end{align}
is the direct sum of positive and negative frequency subspaces which inherit the scalar products from the suitably weighted $L^2$-spaces. Note that the condition $\omega \neq 0$ automatically implies that the fields are \emph{transversal}, \ie they satisfy the divergence free conditions~\eqref{setup:eqn:generic_Maxwell_equations:no_sources}. 

If we endow \eqref{setup:eqn:generic_Maxwell_operator} with the domain 
\begin{align}
	\mathcal{D}(M_{\lambda}) := \bigl ( P_{+,\lambda} \, \mathcal{D}(\Rot) \bigr ) \oplus \bigl ( P_{-,\lambda} \, \mathcal{D}(\Rot) \bigr )
	, 
	\label{setup:eqn:domain_gyrotropic_M_lambda}
\end{align}
then $M_{\lambda}$ defines a selfadjoint operator on $\Hil_{\lambda}$. 
\begin{definition}[Maxwell operator for gyrotropic media]\label{setup:defn:gyrotropic_Maxwell_operator}
	Suppose the modulated material weights $W_+ = W$ satisfy Assumption~\ref{setup:assumption:modulated_weights}. Then the Maxwell operator for gyrotropic media is given by \eqref{setup:eqn:generic_Maxwell_operator}, endowed with the domain~\eqref{setup:eqn:domain_gyrotropic_M_lambda}, and considered on the Hilbert space~\eqref{setup:eqn:gyrotropic_Hil_lambda}. 
\end{definition}
Baked into the construction is an even particle-hole symmetry given by 
\begin{align*}
	K = \sigma_1 \otimes C : \bigl ( \Psi_+ , \Psi_- \bigr ) \mapsto \bigl ( \overline{\Psi_-} , \overline{\Psi_+} \bigr ) 
\end{align*}
since this antiunitary operator anticommutes with $M_{\lambda}$, 
\begin{align}
	K \, M_{\lambda} \, K = - M_{\lambda}
	. 
	\label{setup:eqn:K_particle_hole_symmetry_M}
\end{align}
The presence of the symmetry $K$ means we still have a frequency band pairing~\eqref{setup:eqn:band_pairing}, and that $\e^{- \ii t M_{\lambda}}$ commutes with $\Re_K := \tfrac{1}{2} \bigl ( \id + K \bigr )$. On a physical level $K$ still translates complex conjugation of fields; this gives rise to a systematic identifcation of \emph{real, transversal} electromagnetic fields with complex fields via 
\begin{align}
	L^2(\R^3,\R^6) \ni 
	(\mathbf{E},\mathbf{H}) = \tfrac{1}{2} \bigl ( \Psi_+ + C \Psi_+ \bigr ) 
	\longleftrightarrow 
	(\Psi_+,C \Psi_+) = \Psi_{(\mathbf{E},\mathbf{H})} 
	\in \Hil_{\lambda} 
	, 
	\label{setup:eqn:canonical_identification_double}
\end{align}
and thus, \eqref{setup:eqn:K_particle_hole_symmetry_M} still implies that real fields remain real under the time evolution. This identification rests on the following 
\begin{lemma}
	Suppose Assumption~\ref{setup:assumption:material_weights} on the material weights holds true. Then the identification of the space $L^2_{\mathrm{trans}}(\R^3,\R^6)$ of real, transversal electromagnetic fields (\ie those satisfying~\eqref{setup:eqn:generic_Maxwell_equations:no_sources}) with 
	\begin{align}
		\mathrm{Eig}(K,+1) :\negmedspace &= \bigl \{ \Psi = (\Psi_+,\Psi_-) \in \Hil_{\lambda} \; \; \vert \; \; K \Psi = + \Psi \bigr \} 
		\notag \\
		&= \bigl \{ (\Psi_+ , C \Psi_+) \; \; \vert \; \; \Psi_+ \in \ran P_{+,\lambda} \bigr \} 
		\label{setup:eqn:eigenspace_K_real_fields}
	\end{align}
	via equation~\eqref{setup:eqn:canonical_identification_double} is a vector space isomorphism. 
\end{lemma}
\begin{proof}
	For non-gyrotropic media, the proof is easy, one has to use the symmetry $C \, P_{+,\lambda} \, C = P_{-,\lambda}$ between positive and negative frequency projections as well as $P_{+,\lambda} \, P_{-,\lambda} = 0$ which follows directly from functional calculus. 
	
	In case $W_- \neq W_+$, then $P_{\pm,\lambda}$ are defined via functional calculus for two \emph{different} operators, and there is no direct way to verify whether $P_{+,\lambda} \, P_{-,\lambda} = 0$. Nevertheless, the claim still holds true: First of all, the explicit characterization of the eigenspace of $K$ to $+1$ as given in \eqref{setup:eqn:eigenspace_K_real_fields} follows from direct computation. Moreover, we may view $\mathrm{Eig}(K,+1)$ as a vector space over $\R$ using the canonical identification of $\C \simeq \R^2$. Clearly, $C \, P_{+,\lambda} \, C = P_{-,\lambda}$ implies $\bigl ( P_{+,\lambda} (\mathbf{E},\mathbf{H}) , P_{-,\lambda} (\mathbf{E},\mathbf{H}) \bigr ) \in \mathrm{Eig}(K,+1)$. 
	
	To show that the association is bijective, we have to prove that $\Re \Psi_+ = 0$ implies $\Psi_+ = 0$ — only then are $\Re$ and $(\mathbf{E},\mathbf{H}) \mapsto \bigl ( P_{+,\lambda} (\mathbf{E},\mathbf{H}) , P_{-,\lambda} (\mathbf{E},\mathbf{H}) \bigr )$ inverses to one another. 
	Evidently, \emph{a priori} $\Re \ii \Psi_+ = 0$ just means that $\Psi_+$ is purely real (we choose to work with purely real vectors merely for notational convenience). We will show that the transversality condition does not allow for purely real or purely imaginary fields. Suppose $\Psi_+ = C \Psi_+ = \Psi_-$ is purely real. Then $\Psi_+ \in \ran P_{+,\lambda} \cap \ran P_{-,\lambda}$ lies in the intersection of positive and negative frequency spaces. 
	
	By definition, $M_{\pm,\lambda}$ are non-negative/non-positive operators, \ie for any $\Psi_{\pm}$ we deduce 
	\begin{align*}
		0 \leq \pm \, \bscpro{\Psi_{\pm}}{M_{\pm,\lambda} \, \Psi_{\pm}}_{L^2_{W_{\pm,\lambda}}(\R^3,\C^6)} 
		= \pm \, \bscpro{\Psi_{\pm}}{\Rot \, \Psi_{\pm}}_{L^2(\R^3,\C^6)}
		. 
	\end{align*}
	Even in case $\Psi_{\pm} \not\in P_{\pm,\lambda} \, \mathcal{D}(\Rot)$, we still retain the information on the difference in sign: we approximate $\Psi_{\pm}$ by cutting off high frequencies, $\Psi_{\pm,\varpi} := 1_{(0,\varpi)}(\pm \Mext_{\pm,\lambda}) \, \Psi_{\pm} \in P_{\pm,\lambda} \, \mathcal{D}(\Rot)$, $\varpi > 0$, and these cut off vectors are necessarily in the domain of $\Rot$. Evidently, for vectors $\Psi_{\pm} \not\in P_{\pm,\lambda} \, \mathcal{D}(\Rot)$ the expectation value $\bscpro{\Psi_{\pm,\varpi}}{\Rot \, \Psi_{\pm,\varpi}}_{L^2(\R^3,\C^6)}$ tends to $\pm \infty$ as $\varpi \rightarrow \infty$. So if $\Psi_+ = \Psi_-$, then necessarily $\bscpro{\Psi_+}{\Rot \, \Psi_+}_{L^2(\R^3,\C^6)} = 0$, and consequently, $\Psi_+ \in P_{\pm,\lambda} \, \mathcal{D}(\Rot)$ is in the domain and $\Psi_+ \perp \Rot \, \Psi_+$. There are two options now: either $\Psi_+ \in \ker \Rot$ or $\Psi_+ \in \ran 1_{(-\infty,0)}(\Mext_{+,\lambda})$. However, writing out the definition of $P_{+,\lambda}$ we get 
	\begin{align*}
		P_{+,\lambda} \, 1_{(-\infty,0)}(\Mext_{+,\lambda}) = 1_{(0,+\infty)}(\Mext_{+,\lambda}) \; 1_{(-\infty,0)}(\Mext_{+,\lambda}) = 0
		, 
	\end{align*}
	which means $\Psi_+ \in \ker \Rot$. But we know that $\ker \Rot$ consists of the (longitudinal) gradient fields \cite[Appendix~A.5]{DeNittis_Lein:adiabatic_periodic_Maxwell_PsiDO:2013}, and thus, $\ker \Rot \cap \ran P_{\pm,\lambda} = \{ 0 \}$. That means $\Psi_+ = 0$ and we have shown the lemma. 
\end{proof}
A second symmetry which will play an important role in our analysis later on is the \emph{grading} $\Gamma = \sigma_3 \otimes \id$ which commutes with $M_{\lambda}$, 
\begin{align*}
	\bigl [ M_{\lambda} , \Gamma \bigr ] = 0 
	. 
\end{align*}
Here, the eigenspaces of $\Gamma$ associated to $\pm 1$ are $\ran P_{\pm,\lambda}$, the spaces of complex waves with purely positive/negative frequencies. 

Moreover, just like in the non-gyrotropic case, the (real-valued) modulation $S(\lambda \hat{x})$, which acts “democratically” on the positive and negative frequency components, can be seen as a unitary operator $\Hil_{\lambda} \longrightarrow \Hil_0$ and connects $M_{\lambda}$ with $M_0$, 
\begin{align*}
	M_{\lambda} = S(\lambda \hat{x}) \, \bigl ( M_{+,0} \vert_{\ran P_{+,0}} \bigr ) \oplus S(\lambda \hat{x}) \, \bigl ( M_{-,0} \vert_{\ran P_{-,0}} \bigr ) 
	= S(\lambda \hat{x}) \, M_0 
	. 
\end{align*}
Lastly, let us mention that any and all of course applies also to \emph{non}-gyrotropic media, \ie if $W_+ = W_-$, then equations~\eqref{setup:eqn:generic_Maxwell_equations} coincides with equations~\eqref{intro:eqn:Maxwell}. 

\subsection{Real electromagnetic fields: reduction to complex waves with $\omega > 0$} 
\label{setup:reduction_omega_g_0}
The complexification of Maxwell's equations leads to a “doubling” of degrees of freedom (as $\C \simeq \R^2$). To undo this doubling, we will restrict our attention to \emph{complex} fields of \emph{positive} frequency. As a side benefit we are able to discard a lot of notational baggage. 

Any $\Psi_+ \in \Hil_{+,\lambda}$ defines \emph{two} real solustions, namely real and imaginary parts, 
\begin{subequations}
	\label{setup:eqn:real_states_linear_combination_even_bands}
	\begin{align}
		\bigl ( \mathbf{E}_{\Re} , \mathbf{H}_{\Re} \bigr ) &= \Re \, \Psi_+ = \tfrac{1}{2} \bigl ( \Psi_+ + C \Psi_+ \bigr ) 
		, 
		\\
		\bigl ( \mathbf{E}_{\Im} , \mathbf{H}_{\Im} \bigr ) &= \Im \, \Psi_+ = \tfrac{1}{\ii 2} \bigl ( \Psi_+ - C \Psi_+ \bigr ) 
		. 
	\end{align}
\end{subequations}
Then any linear combination of $\bigl ( \mathbf{E}_{\Re} , \mathbf{H}_{\Re} \bigr )$ and $\bigl ( \mathbf{E}_{\Im} , \mathbf{H}_{\Im} \bigr )$ with real coefficients $\alpha_{\Re} , \alpha_{\Im} \in \R$ can be expressed as the real part of a complex wave, 
\begin{align*}
	\alpha_{\Re} \, \bigl ( \mathbf{E}_{\Re} , \mathbf{H}_{\Re} \bigr ) + \alpha_{\Im} \, \bigl ( \mathbf{E}_{\Im} , \mathbf{H}_{\Im} \bigr ) = \Re \, \bigl ( (\alpha_{\Re} - \ii \alpha_{\Im}) \, \Psi_+ \bigr ) 
	. 
\end{align*}
This “phase locking” explains why all information is contained in $\Hil_{+,\lambda}$. Hence, we will identify the space of transversal \emph{real}-valued fields $L^2_{\mathrm{trans}}(\R^3,\R^6)$ with $\Hil_{+,\lambda}$ via 
\begin{align}
	P_{+,\lambda} : L^2_{\mathrm{trans}}(\R^3,\R^6) \longrightarrow \Hil_{+,\lambda} 
	\label{setup:eqn:identification_L2_real_L2_complex_omega_g_0}
\end{align}
and its inverse 
\begin{align*}
	\Re : \Hil_{+,\lambda} \longrightarrow L^2_{\mathrm{trans}}(\R^3,\R^6) 
	 . 
\end{align*}
What we are doing here is something completely standard and covered in every text book on electromagnetism, we are writing $(\mathbf{E},\mathbf{H})$ as the real part of a complex wave (see \eg \cite[equation~(6.128)]{Jackson:electrodynamics:1998}), 
\begin{align*}
	\bigl ( \mathbf{E}(t) , \mathbf{H}(t) \bigr ) &= \Re \, \bigl ( \e^{- \ii t M_{+,\lambda}} \, P_{+,\lambda} (\mathbf{E},\mathbf{H}) \bigr ) 
	. 
\end{align*}
From a mathematical perspective, this identification of vector spaces is a bit delicate, because $L^2_{\mathrm{trans}}(\R^3,\R^6)$ is a vector space over $\R$ while $\Hil_{+,\lambda}$ is a vector space over $\C$ — and tremendously useful because it allows us to adapt techniques initially developed for quantum mechanics to classical electromagnetism. For instance, the identification~\eqref{setup:eqn:identification_L2_real_L2_complex_omega_g_0} allows us to define and compute Chern classes (of vector bundles with \emph{complex} fibers) associated to (real) electromagnetic fields; We will explore this point in more detail in an upcoming publication \cite{DeNittis_Lein:symmetries_electromagnetism:2016}. 

Consequently, there is no need to work with direct sum spaces or distinguish between non-gyrotropic ($W = \overline{W}$) and gyrotropic ($W \neq \overline{W}$) materials. Instead, it suffices to study $M_{+,\lambda} = \Mext_{+,\lambda} \, P_{+,\lambda}$ on $\Hil_{+,\lambda}$; this operator in turn inherits all essential properties from $\Mext_{+,\lambda}$. For instance, $M_{+,\lambda} = \Op_{\lambda}^{S \Zak}(\Msymb_{\lambda}) \, P_{+,\lambda}$ can still be seen as a pseudodifferential operator associated to the symbol~\eqref{setup:eqn:symbol_Maxwellian} (where we set $M_0 = W_{+,0} \, \Rot$). While $P_{+,\lambda}$ is \emph{not} a pseudodifferential operator, we will only work with states associated to projections $\Pi_{+,\lambda}$ that \emph{are} $\Psi$DOs (\cf equation~\eqref{ray_optics:eqn:projection}) and satisfy 
\begin{align}
	\Pi_{+,\lambda} \, P_{+,\lambda} &= \Pi_{+,\lambda} + \order_{\norm{\cdot}}(\lambda^{\infty}) 
	= P_{+,\lambda} \, \Pi_{+,\lambda} + \order_{\norm{\cdot}}(\lambda^{\infty}) 
	. 
	\label{setup:eqn:relation_Pi_P}
\end{align}
That $P_{+,\lambda}$ is not a $\Psi$DO can be easily seen for the case $\lambda = 0$: then $k \mapsto M_{+,0}(k)$ is not analytic at $k = 0$ as we “lose” a two-dimensional subspace due to the ground state bands where $\omega_n(0) = 0$ (\cf discussion in \cite[Sections~3.2 and 3.3]{DeNittis_Lein:adiabatic_periodic_Maxwell_PsiDO:2013}). As explained in \cite[p.~230]{DeNittis_Lein:sapt_photonic_crystals:2013}, this essential difference between ground state and other bands is a physical one, and also here we will exclude ground state bands from our considerations (\cf Assumption~\ref{ray_optics:assumption:bands}). All other bands are well-behaved, though, and inherit the analyticity properties from $\Mext_{+,0}$, the operator discussed in \cite{DeNittis_Lein:adiabatic_periodic_Maxwell_PsiDO:2013,DeNittis_Lein:sapt_photonic_crystals:2013}. 

\section{The meaning of the ray optics limit} 
\label{ray_optics}
It is tempting to think that now that we have recast Maxwell's equations in Schrödinger form, the ray optics limit is just a matter of applying your semiclassical technique of choice. To the extend of mathematics, this may be true, but from a physical perspective, we have to take the differences between quantum mechanics and classical electromagnetism into account. Most importantly, the notion of “physical observable” is different. While \emph{quantum} observables are usually represented by selfadjoint operators, in electromagnetism they are suitable functions of the fields 
\begin{align*}
	\mathcal{F} : L^2(\R^3,\C^6) \longrightarrow \R 
	. 
\end{align*}
Examples of observables in electromagnetism include the \emph{energy density} 
\begin{align*}
	e_x[(\mathbf{E},\mathbf{H})] := \tfrac{1}{2} \, \bigl ( \mathbf{E}(x) , \mathbf{H}(x) \bigr ) \cdot W^{-1}(x) \bigl ( \mathbf{E}(x) , \mathbf{H}(x) \bigr )
	, 
\end{align*}
given for non-gyrotropic media here, the \emph{Poynting vector} 
\begin{align*}
	\mathcal{S}_x[(\mathbf{E},\mathbf{H})] := \mathbf{E}(x) \times \mathbf{H}(x) 
	, 
\end{align*}
even the fields themselves, \eg $\delta_x^{\mathbf{E}}(\mathbf{E},\mathbf{H}) := \mathbf{E}(x)$, as well as their local averages. This is in stark contrast to quantum mechanics where the wave function itself cannot be observed. At the end of the day, electromagnetism, even if written in the language of quantum mechanics, is still a classical field theory. 

Secondly, just as not every quantum observable has a semiclassical limit, not every observable in electromagnetism has a ray optics limit -- at least not via Theorem~\ref{ray_optics:thm:ray_optics}. Our goal is to derive a ray optics limit for quadratic observables which in the simplest case ($W = \overline{W}$ and $f = \bar{f}$) are of the form 
\begin{align}
	\mathcal{F}[(\mathbf{E},\mathbf{H})] = \bscpro{(\mathbf{E},\mathbf{H}) \,}{\, \Op_{\lambda}^{S \Zak}(f) \, (\mathbf{E},\mathbf{H})}_{L^2_{W_{\lambda}}(\R^3,\C^6)}
	\label{ray_optics:eqn:quadratic_observables_simple}
\end{align}
where the pseudodifferential operator $\Op_{\lambda}^{S \Zak}(f)$ is defined via equation~\eqref{setup:eqn:M_lambda_PsiDO}. Put another way, we consider the ray topics limit in the “Heisenberg picture” where we compare $\mathcal{F} \bigl [ \bigl ( \mathbf{E}(t) , \mathbf{H}(t) \bigr ) \bigr ]$ to the expectation value where $f$ is replaced by a suitable time-evolved observable $f_{\lambda}(t)$. We will make this precise in what follows. 

This is in contrast to previous approaches which implement “semiclassical wave packet techniques”, \eg a multiscale WKB ansatz \cite{Allaire_Palombaro_Rauch:diffractive_Bloch_wave_packets_Maxwell:2012} or wave packet techniques \cite{Bliokh_Bliokh:Berry_curvature_optical_Magnus_effect:2004,Onoda_Murakami_Nagaosa:Hall_effect_light:2004,Bliokh_et_al:spin_orbit_light:2015}. We believe our approach gives additional insights, because we not only give an \emph{electromagnetic} observable, but also the relevant \emph{ray optics observable}. This establishes relations akin to that between the quantum angular momentum operator $\hat{L} = \hat{x} \times (- \ii \eps \nabla)$ and the classical angular momentum $L(x,p) = x \times p$. While it makes no sense to claim a quadratic electromagnetic observable $\mathcal{F}$ of the form~\eqref{ray_optics:eqn:quadratic_observables_simple} is the “quantization” of the ray optics observable $f$, the roles are analogous: $f$ is the ray optics limit of the electromagnetic observable $\mathcal{F}$. Note that $f$ need not be a scalar function.

\subsection{A class of observables with a ray optics limit} 
\label{ray_optics:observables}
The reality of electromagnetic fields places a consistency condition on electromagnetic observables $\mathcal{F}$. Consider quadratic observables 
\begin{align}
	\mathcal{F}[(\mathbf{E},\mathbf{H})] = \bscpro{\Psi_{(\mathbf{E},\mathbf{H})}}{F \, \Psi_{(\mathbf{E},\mathbf{H})}}_{\lambda}
	\label{ray_optics:eqn:generic_observable}
\end{align}
defined in terms of some generic bounded operator 
\begin{align*}
	F = \left (
	\begin{matrix}
		F_{++} & F_{+-} \\
		F_{-+} & F_{--} \\
	\end{matrix}
	\right ) \in \mathcal{B}(\Hil_{\lambda}) 
	. 
\end{align*}
Here, the splitting of $F$ corresponds to the positive/negative frequency splitting, \eg $F_{+-}$ maps $\ran P_{\lambda,-}$ to $\ran P_{+,\lambda}$. Many physically relevant observables are of this type (for details see \eg \cite[Section~3.3]{Bliokh_Bekshaev_Nori:dual_electromagnetism:2013} and Section~\ref{ray_optics:examples}). 

To make sure $\mathcal{F}$ preserves the reality of electromagnetic waves, we need to impose 
\begin{align}
	\mathcal{F} \bigl [ K (\mathbf{E},\mathbf{H}) \bigr ] &= \overline{\mathcal{F}[(\mathbf{E},\mathbf{H})]}
	. 
	\label{ray_optics:eqn:two_symmetry_conditions:K}
\end{align}
Additionally, quadratic observables with a ray optics limit must satisfy a second condition, namely
\begin{align}
	\mathcal{F} \bigl [ \Gamma (\mathbf{E},\mathbf{H}) \bigr ] &= \mathcal{F}[(\mathbf{E},\mathbf{H})]
	, 
	\label{ray_optics:eqn:two_symmetry_conditions:Gamma}
\end{align}
because this condition is necessary to be able to reduce a genuine multiband problem to a single-band problem. We will explain this point in more detail below after studying the consequences of the presence of these symmetries. 
\begin{lemma}\label{ray_optics:lem:form_of_F}
	Suppose $\mathcal{F}$ defined through a selfadjoint $F \in \mathcal{B}(\Hil_{\lambda})$ via equation~\eqref{ray_optics:eqn:generic_observable} satisfies \eqref{ray_optics:eqn:two_symmetry_conditions:K}–\eqref{ray_optics:eqn:two_symmetry_conditions:Gamma}. Then $F$ is of the form $F = F_+ \oplus F_- = F_{++} \oplus C \, F_{++} \, C$. 
\end{lemma}
\begin{proof}
	A quick computation reveals that equations~\eqref{ray_optics:eqn:two_symmetry_conditions:K} and \eqref{ray_optics:eqn:two_symmetry_conditions:Gamma} translate to 
	\begin{align*}
		[ K , F ] &= 0 
		, 
		\\
		[ \Gamma , F ] &= 0 
		. 
	\end{align*}
	Let us start with the second condition: Writing $\Gamma = \sigma_3 \otimes \id$ immediately yields that the offdiagonal elements $F_{+-} = 0 = F_{-+}$ must vanish. The first commutator condition then relates $F_{--} = C \, F_{++} \, C$ to $F_{++}$. 
\end{proof}
To appreciate the role symmetry~\eqref{ray_optics:eqn:two_symmetry_conditions:Gamma} plays, we have to go back to the reality of electromagnetic fields: As explained in Section~\ref{setup:reduction_omega_g_0}, transverse electromagnetic waves $(\mathbf{E},\mathbf{H})$ are always linear combinations of an \emph{even} number of bands. Hence, even in the simplest case, electromagnetic fields are associated to $\omega(k)$ and its symmetric twin $-\omega(-k)$ (\cf equations~\eqref{setup:eqn:band_pairing} and \eqref{setup:eqn:real_states_linear_combination_even_bands}). The presence of this symmetry condition eliminates $F_{+-}$ and $F_{-+}$ so that $\mathcal{F}$ can be reduced to an expectation value with respect to the positive frequency contribution $\Psi_+ = P_{+,\lambda} (\mathbf{E},\mathbf{H})$ only. 
\begin{proposition}\label{ray_optics:prop:rewriting_electromagnetic_observables}
	Suppose $\mathcal{F}$ is a quadratic observable of the form~\eqref{ray_optics:eqn:generic_observable} associated to a bounded operator $F = F_+ \oplus C \, F_+ \, C \in \mathcal{B}(\Hil_{\lambda})$. Then $\mathcal{F}$ can be expressed as 
	\begin{align}
		\mathcal{F}[(\mathbf{E},\mathbf{H})] &= 2 \, \Re \, \bscpro{P_{+,\lambda} \, (\mathbf{E},\mathbf{H}) \,}{\, F_+ \, P_{+,\lambda} \, (\mathbf{E},\mathbf{H})}_{L^2_{W_{+,\lambda}}(\R^3,\C^6)} 
		\label{ray_optics:eqn:rewriting_F_to_Fplus}
		\\
		&= 2 \, \Re \, \bscpro{P_{+,\lambda} \, (\mathbf{E},\mathbf{H}) \,}{\, P_{+,\lambda} \, F_+ \, P_{+,\lambda} \, (\mathbf{E},\mathbf{H})}_{\Hil_{+,\lambda}} 
		. 
		\notag
	\end{align}
\end{proposition}
\begin{proof}
	For brevity, let us define $\Psi_+ := P_{+,\lambda} (\mathbf{E},\mathbf{H})$ so that \emph{real} states are of the form $(\mathbf{E},\mathbf{H}) \simeq (\Psi_+ , C \Psi_+)$. From $C \, P_{+,\lambda} \, C = P_{\lambda,-}$ and a quick computation we obtain 
	\begin{align*}
		\mathcal{F}[(\mathbf{E},\mathbf{H})] &= \bscpro{(\mathbf{E},\mathbf{H}) \,}{\, \bigl ( F_+ \oplus C \, F_+ \, C \bigr ) (\mathbf{E},\mathbf{H})}_{\lambda} 
		\\
		&= \bscpro{\Psi_+}{F_+ \Psi_+}_{L^2_{W_{+,\lambda}}(\R^3,\C^6)} + \bscpro{C \Psi_+}{C \, F_+ \Psi_+}_{L^2_{W_{\lambda,-}}(\R^3,\C^6)}
		\\
		&= 2 \, \Re \, \bscpro{P_{+,\lambda} (\mathbf{E},\mathbf{H}) \,}{\, F_+ \, P_{+,\lambda} (\mathbf{E},\mathbf{H})}_{L^2_{W_{+,\lambda}}(\R^3,\C^6)}
		. 
	\end{align*}
\end{proof}
\begin{remark}
	If we only imposed \eqref{ray_optics:eqn:two_symmetry_conditions:K}, then another term 
	\begin{align*}
		2 \, \Re_K \, \bscpro{\Psi_+}{C \, F_{-+} \Psi_+}_{L^2_{W_{+,\lambda}}(\R^3,\C^6)} = 2 \, \Re_K \, \bscpro{\Psi_+}{F_{+-} \, C \Psi_+}_{L^2_{W_{+,\lambda}}(\R^3,\C^6)}
	\end{align*}
	would appear that mixed positive and negative frequency contributions. While this term is still an expectation value, it is taken with respect to an \emph{anti}linear operator $C \, F_{-+} = F_{+-} \, C$. Perhaps an Egorov-type theorem can still be established in this case, but because complex conjugation $C$ only appears to either the left or the right, we do not see an easy way to translate this to the level of symbols as in \cite[Lemma~5]{DeNittis_Lein:sapt_photonic_crystals:2013}. 
\end{remark}
This reduction to positive frequencies allows us to give a simple definition of the relevant observables that holds for both, the non-gyrotropic \emph{and} gyrotropic case. More importantly, it reduces ray optics from a \emph{bona fide} multiband problem to a single band problem. 
\begin{definition}[Quadratic observables]\label{ray_optics:defn:admissible_observables}
	Suppose the electromagnetic observable 
	\begin{align}
		\mathcal{F}[(\mathbf{E},\mathbf{H})] 
		:\negmedspace &= \mathbb{E}_f[(\mathbf{E},\mathbf{H})] 
		\notag \\
		:\negmedspace &= 2 \, \Re \, \bscpro{P_{+,\lambda} (\mathbf{E},\mathbf{H}) \,}{\, \Op_{\lambda}^{S \Zak}(f) \, P_{+,\lambda} (\mathbf{E},\mathbf{H})}_{L^2_{W_{+,\lambda}}(\R^3,\C^6)}
		\label{ray_optics:eqn:observable_as_expectation_value}
	\end{align}
	is defined in terms of a $\Psi$DO associated to a function $f = f^*$ (\cf equation~\eqref{setup:eqn:M_lambda_PsiDO}). 
	\begin{enumerate}[(i)]
		\item We call $\mathcal{F}$ \emph{scalar} if $f \equiv f \otimes \id_{\HperT}$ and $f \in \Cont^{\infty}_{\mathrm{b}}(\R^6,\C)$ are periodic in $k$. 
		\item We call $\mathcal{F}$ \emph{non-scalar} if $f \in \Cont^{\infty}_{\mathrm{b}} \bigl ( \R^6 , \mathcal{B}(\HperT) \bigr )$ is an operator-valued function satisfying the equivariance condition~\eqref{setup:eqn:equivariance}. 
	\end{enumerate}
\end{definition}
The assumptions on $f$ ensure that $\Op_{\lambda}^{S \Zak}(f)$ defines a bounded selfadjoint operator on $L^2_{W_{+,\lambda}}(\R^3,\C^6)$ (\cf Section~\ref{setup:real:Maxwell_operator_PsiDO} and Appendix~\ref{appendix:PsiDOs} for details). 
\begin{remark}\label{ray_optics:remark:definition_ray_optics_observable}
	Note that in the definition of quadratic observables~\eqref{ray_optics:eqn:observable_as_expectation_value} we have used the scalar product on $L^2_{W_{+,\lambda}}(\R^3,\C^6)$ rather than $\Hil_{+,\lambda}$, because $\Op_{\lambda}^{S \Zak}(f)$ defines a bounded operator on $L^2_{W_{+,\lambda}}(\R^3,\C^6)$ so that 
	\begin{align*}
		&\bscpro{P_{+,\lambda} (\mathbf{E},\mathbf{H}) \,}{\, P_{+,\lambda} \, \Op_{\lambda}^{S \Zak}(f) \, P_{+,\lambda} (\mathbf{E},\mathbf{H})}_{\Hil_{+,\lambda}} =
		\\
		&\qquad \qquad \qquad 
		= \bscpro{P_{+,\lambda} (\mathbf{E},\mathbf{H}) \,}{\, \Op_{\lambda}^{S \Zak}(f) \, P_{+,\lambda} (\mathbf{E},\mathbf{H})}_{L^2_{W_{+,\lambda}}(\R^3,\C^6)} 
	\end{align*}
	holds. This allows us to omit one projection and simplify many arguments. 
\end{remark}
%

\subsection{A semiclassical approach to the ray optics limit} 
\label{ray_optics:main_theorem}
Now we come to the main course of the paper, a rigorous justification of ray optics via a semiclassical limit for observables of the form~\eqref{ray_optics:eqn:observable_as_expectation_value}. Roughly speaking, if the initial state is associated to a single, non-degenerate frequency band $\omega(k) \in \sigma \bigl ( \Mper(k) \bigr )$ which does not intersect or merge with other bands, the dispersion relation which enters in the ray optics equations is no longer $c \abs{k}$ but proportional to the frequency band function $\omega(k)$ to leading order. Let us be more precise and enumerate the conditions on the frequency band: 
\begin{assumption}\label{ray_optics:assumption:bands}
	Suppose $\omega(k)$ is a non-degenerate frequency band of $\Mper(k)$ with Bloch function $\varphi(k)$ that is isolated in the sense that 
	\begin{align}
		\inf_{k \in \BZ} \mathrm{dist} \Bigl ( \bigl \{ \omega(k) \bigr \} , \sigma \bigl ( \Mper(k) \bigr ) \setminus \bigl \{ \omega(k) \bigr \} \Bigr ) > 0 
		, 
		\label{ray_optics:eqn:gap_condition}
	\end{align}
	and that is \emph{not} a ground state band, \ie $\lim_{k \to 0} \omega_n(k) \neq 0$. 
\end{assumption}
Next, we need to clarify what we mean when we say “states associated to the frequency band $\omega$” in case the photonic crystal is perturbed. The perturbation deforms the subspace $\Zak^{-1} \ran \pi_0(\hat{k})$ where $\pi_0(k) := \sopro{\varphi(k)}{\varphi(k)}$, and the range of the \emph{superadiabatic projection}
\begin{align}
	\Pi_{\lambda} = \Op_{\lambda}^{S \Zak}(\pi_{\lambda}) + \order_{\norm{\cdot}}(\lambda^{\infty}) 
	= \Op_{\lambda}^{S \Zak}(\pi_0) + \order_{\norm{\cdot}}(\lambda)
	\label{ray_optics:eqn:projection}
\end{align}
takes its places. Note that even though $k \mapsto \varphi(k)$ need not be continuous (this is the case if the band is not topologically trivial), the gap condition ensures that $k \mapsto \pi_0(k)$ is necessarily analytic. Apart from being an orthogonal projection, its other defining property is 
\begin{align}
	\bigl [ M_{\lambda} , \Pi_{\lambda} \bigr ] = \order_{\norm{\cdot}}(\lambda^{\infty})
	. 
	\label{ray_optics:eqn:defining_relation_Pi_lambda}
\end{align}
The existence and explicit construction of $\Pi_{\lambda}$ relies on the gap condition~\eqref{ray_optics:eqn:gap_condition} and pseudo\-differential techniques (\cf \cite[Proposition~1]{DeNittis_Lein:sapt_photonic_crystals:2013}). Equation~\eqref{ray_optics:eqn:defining_relation_Pi_lambda} also implies that $\ran \Pi_{\lambda}$ is left invariant by the dynamics up to errors of arbitrarily small order in $\lambda$. 

For electromagnetic waves from the almost invariant subspace $\ran \Pi_{\lambda}$, we are going to rigorously justify the \emph{analog of the semiclassical limit for the Bloch electron} where the periodic structure of the ambient medium modifies the dispersion relation to 
\begin{align}
	\Omega &= \Omega_0 + \lambda \, \Omega_1 
	:= \tau^2 \, \omega - \lambda \, \tau^2 \, \mathcal{P} \cdot \nabla_r \ln \tfrac{\tau_{\eps}}{\tau_{\mu}}
	. 
	\label{ray_optics:eqn:ray_optics_Maxwellian}
\end{align}
To leading order, the frequency band function $\omega(k)$ is modulated by the perturbation $\tau(r)^2 = \tau_{\eps}(r) \, \tau_{\mu}(r)$, \ie the frequency depends on the change in the speed of light. The $\order(\lambda)$ term is sensitive to the details of the modulation, and only appears if $\eps$ and $\mu$ are modulated differently; moreover, it includes the imaginary part of the complex Poynting vector, 
%
\begin{align*}
	\mathcal{P}(k) := \Im \int_{\T^3} \dd y \, \overline{\varphi^E(k,y)} \times \varphi^H(k,y) 
	. 
\end{align*}
Note that this expression works for both types of perturbations, in case $\chi \neq 0$ we take $\tau_{\eps} = \tau_{\mu}$ and the last term vanishes. In addition to $\Omega_1$, there are also other $\order(\lambda)$ contributions to the ray optics equations as we will see below. 

Then quadratic observables $\mathcal{F}[(\mathbf{E},\mathbf{H})] = \mathbb{E}_f[(\mathbf{E},\mathbf{H})]$ from Definition~\ref{ray_optics:defn:admissible_observables} come in \emph{pairs}, and it is natural to compare $\mathcal{F} \bigl [ \bigl ( \mathbf{E}(t) , \mathbf{H}(t) \bigr ) \bigr ]$ to 
\begin{align*}
	\mathbb{E}_{f_{\lambda} \circ \Phi^{\lambda}_t} \bigl [ \bigl ( \mathbf{E}(0) , \mathbf{H}(0) \bigr ) \bigr ]
\end{align*}
where $f_{\lambda} = f_0 + \lambda \, f_1$ is connected to $f$ and transported along the \emph{ray optics flow} $\Phi^{\lambda}$. 
\begin{theorem}[The ray optics limit]\label{ray_optics:thm:ray_optics}
	Suppose Assumptions~\ref{setup:assumption:modulated_weights} and \ref{ray_optics:assumption:bands} hold true, and $\mathcal{F}$ is a quadratic observable associated to $f$ as in Definition~\ref{ray_optics:defn:admissible_observables}. Then we have a ray optics limit in the following sense: 
	\begin{enumerate}[(i)]
		\item For \emph{scalar} observables where $f \in \Cont^{\infty}_{\mathrm{b}}(\R^6,\C)$, the ray optics flow $\Phi^{\lambda}$ associated to the hamiltonian equations 
		\begin{align}
			\left (
			\begin{matrix}
				\dot{r} \\
				\dot{k} \\
			\end{matrix}
			\right ) &= \left (
			\begin{matrix}
				- \lambda \, \Xi & + \id \\
				- \id & 0 \\
			\end{matrix}
			\right ) \left (
			\begin{matrix}
				\nabla_r \Omega \\
				\nabla_k \Omega \\
			\end{matrix}
			\right )
			, 
			\label{ray_optics:eqn:ray_optics_equations_scalar}
		\end{align}
		which include the Berry curvature tensor $\Xi := \bigl ( \nabla_k \times \ii \sscpro{\varphi}{\nabla_k \varphi}_{\HperT} \bigr )^{\times}$ as part of the sympletic form, approximates the full light dynamics for $P_{+,\lambda} (\mathbf{E},\mathbf{H}) \in \ran P_{+,\lambda} \,  \Pi_{\lambda}$ and bounded times in the sense 
		\begin{align}
			\mathcal{F} \bigl [ \bigl ( \mathbf{E}(t),\mathbf{H}(t) \bigr ) \bigr ] &= \mathcal{F} \bigl [ \bigl ( \e^{- \ii \frac{t}{\lambda} M_{\lambda}} \Psi_{(\mathbf{E},\mathbf{H})} \bigr ) \bigr ]
			\notag \\
			&
			= \mathbb{E}_{f \circ \Phi^{\lambda}_t}[(\mathbf{E},\mathbf{H})]
			+ \order(\lambda^2)
			\label{ray_optics:eqn:ray_optics_limit_scalar}
			. 
		\end{align}
		\item For \emph{non-scalar} observables where $f \in \Cont^{\infty}_{\mathrm{b}} \bigl ( \R^6 , \mathcal{B}(\HperT) \bigr )$, the ray optics flow $\Phi^{\lambda}$ associated to the hamiltonian equations 
		\begin{align}
			\left (
			\begin{matrix}
				\dot{r} \\
				\dot{k} \\
			\end{matrix}
			\right ) &= \left (
			\begin{matrix}
				0 & + \id \\
				- \id & 0 \\
			\end{matrix}
			\right ) \left (
			\begin{matrix}
				\nabla_r \Omega \\
				\nabla_k \Omega \\
			\end{matrix}
			\right )
			\label{ray_optics:eqn:ray_optics_equations_non-scalar}
		\end{align}
		approximates the full light dynamics for $P_{+,\lambda} (\mathbf{E},\mathbf{H}) \in \ran P_{+,\lambda} \,  \Pi_{\lambda}$ and bounded times in the sense 
		\begin{align}
			\mathcal{F} \bigl [ \bigl ( \mathbf{E}(t),\mathbf{H}(t) \bigr ) \bigr ] &= \mathcal{F} \bigl [ \bigl ( \e^{- \ii \frac{t}{\lambda} M_{\lambda}} \Psi_{(\mathbf{E},\mathbf{H})} \bigr ) \bigr ]
			\notag \\
			&
			= \mathbb{E}_{f_{\mathrm{ro}} \circ \Phi^{\lambda}_t}[(\mathbf{E},\mathbf{H})]
			+ \order(\lambda^2)
			\label{ray_optics:eqn:ray_optics_limit_non-scalar}
		\end{align}
		where we have transported the modified non-scalar observable $f_{\mathrm{ro}} := \pi_{\lambda} \Weyl f \Weyl \pi_{\lambda} + \order(\lambda^2)$, defined in terms of $\pi_{\lambda}$ from equation~\eqref{ray_optics:eqn:projection}, along the flow $\Phi^{\lambda}$. Put another way, for non-scalar observables the effect of the projection does not modify the symplectic form of the ray optics equations but changes the function $f$ which defines the quadratic observable $\mathcal{F}$. 
	\end{enumerate}
\end{theorem}
The proof rests on an Egorov theorem; we will postpone this to Section~\ref{egorov}. 
\begin{remark}
	We can immediately extend Theorem~\ref{ray_optics:thm:ray_optics} to quadratic observables of the type 
	\begin{align}
		\mathcal{F}[\rho] = 2 \, \Re \, \trace_{L^2_{W_{+,\lambda}}(\R^3,\C^6)} \Bigl ( \rho \; \Op_{\lambda}^{S \Zak}(f) \Bigr )
		\label{ray_optics:eqn:observable_as_trace}
	\end{align}
	where $\rho = \rho^* \geq 0$ is a suitable trace-class operator that describes a mixture of different electromagnetic states. Although this generalization is physically relevant and meaningful, from a mathematical point of view the passage from \eqref{ray_optics:eqn:observable_as_expectation_value} to  \eqref{ray_optics:eqn:observable_as_trace} is totally trivial. 
\end{remark}
For scalar quadratic observables we can express \eqref{ray_optics:eqn:ray_optics_limit_scalar} as a phase space average of $f$ with respect to the Wigner transform. 
\begin{corollary}\label{ray_optics:cor:phase_space_average}
	Suppose we are in the setting of Theorem~\ref{ray_optics:thm:ray_optics}~(i) where $f$ is scalar. Then writing $\Psi_+ := P_{+,\lambda} (\mathbf{E},\mathbf{H})$ we can recast equation~\eqref{ray_optics:eqn:ray_optics_limit_scalar} as 
	\begin{align*}
		\mathcal{F} \bigl [ \bigl ( \mathbf{E}(t),\mathbf{H}(t) \bigr ) \bigr ] &= \int_{\R^3} \dd r \int_{\BZ} \dd k \, f \circ \Phi^{\lambda}_t(r,k) \; \mathrm{w}_{S(\lambda \hat{x}) \Psi_+}^{\mathrm{red}}(r,k) 
		\, + \\
		&\qquad \qquad
		+ \order \bigl ( \lambda^2 \, \snorm{(\mathbf{E},\mathbf{H})}_{\Hil_{\lambda}}^2 \bigr )
	\end{align*}
	where $\mathrm{w}_{S(\lambda \hat{x}) \Psi_+}^{\mathrm{red}}(r,k)$ denotes the reduced Wigner transform of $S(\lambda \hat{x}) \Psi \in L^2(\R^3,\C^6)$ given by 
	\begin{align}
		\mathrm{w}_{\Psi_+}^{\mathrm{red}}(r,k) :\negmedspace &= \frac{1}{(2 \pi \lambda)^3} \sum_{\gamma^* \in \Gamma^*} \int_{\R^3} \dd z \, \e^{+ \ii (k + \gamma^*) \cdot z} \, 
		\cdot \notag \\
		&\qquad \qquad \qquad \cdot 
		\Psi_+ \bigl ( \tfrac{r}{\lambda} + \tfrac{z}{2} \bigr ) \cdot W_+^{-1} \bigl ( \tfrac{r}{\lambda} + \tfrac{z}{2} \bigr ) \Psi_+ \bigl ( \tfrac{r}{\lambda} - \tfrac{z}{2} \bigr )  
		\label{ray_optics:eqn:Wigner_transform}
		. 
	\end{align}
\end{corollary}
We postpone a discussion of two observables with a ray optics limit, the local averages of the energy density and the Poynting vector, to Section~\ref{ray_optics:examples} below.

For the reader's convenience we have included a proof (\cf Appendix~\ref{appendix:reduced_Wigner_transform}) whose main purpose is to show how to correctly include the material weights $W_+^{-1}$ in the zone-folded Wigner transform. Note that even though equation~\ref{ray_optics:eqn:Wigner_transform} seems to be asymmetric, $W_+^{-1}$ is evaluated at $\nicefrac{r}{\lambda} + \nicefrac{z}{2}$, the equivalent expression~\eqref{appendix:reduced_Wigner_transform:eqn:symmetric_formula_reduced_Wigner_transform} for $\mathrm{w}_{\Psi_+}^{\mathrm{red}}$ is perfectly symmetric and allays those doubts. 
\begin{remark}
	Scalar observables have a somewhat simpler ray optics limit, because here a geometric correction, the Berry curvature, enters in the symplectic form. For non-scalar observables though, the Weyl commutator $\bigl [ f , \pi_{\lambda} \bigr ]_{\Weyl} = \order(1)$ is not small as 
	\begin{align}
		f \, \pi_0 &\neq \pi_0 \, f 
		\label{ray_optics:eqn:ray_optics_f_assumption} 
	\end{align}
	holds. Consequently, instead of getting an $\order(\lambda)$ correction in the symplectic form, we need to replace the function $f$ by 
	\begin{align}
		f_{\mathrm{ro}} &= \pi_{\lambda} \Weyl f \Weyl \pi_{\lambda} + \order(\lambda^2) 
		\notag \\
		&= \scpro{\varphi}{f \varphi}_{\HperT} \, \pi_0 + \lambda \, \Bigl ( \pi_1 \, f \, \pi_0 + \pi_0 \, f \, \pi_1 - \tfrac{\ii}{2} \bigl \{ \pi_0 , f \bigr \} \, \pi_0 - \tfrac{\ii}{2} \pi_0 \, \bigl \{ f , \pi_0 \bigr \} \Bigr ) 
		. 
		\label{ray_optics:eqn:f_ro_explicit}
	\end{align}
	Note that the term proportional to 
	\begin{align}
		\bigl \{ \pi_0 \vert f \vert \pi_0 \bigr \} := \sum_{j = 1}^3 \Bigl ( \partial_{k_j} \pi_0 \; f \; \partial_{r_j} \pi_0 - \partial_{r_j} \pi_0 \; f \; \partial_{k_j} \pi_0 \Bigr ) = 0 
		\label{ray_optics:eqn:Poisson_bracket_of_three}
	\end{align}
	vanishes identically as $\pi_0$ is a function of $k$ only. The crucial idea of Stiepan and Teufel \cite{Stiepan_Teufel:semiclassics_op_valued_symbols:2012} was to avoid including this $\order(\lambda)$ term by modifying the symplectic form. However, their derivation relies on $\bigl [ f , \pi_0 \bigr ] = 0$ \emph{and} 
	\begin{align*}
		\bigl \{ \pi_0 , f \bigr \} = - \bigl \{ f , \pi_0 \bigr \}
		. 
	\end{align*}
\end{remark}
\begin{lemma}\label{ray_optics:cor:explicit_expressions_pi_1_f_ro}
	The explicit expression for $\pi_{\lambda}$ in Theorem~\ref{ray_optics:thm:ray_optics} is $\pi_{\lambda} = \sopro{\varphi}{\varphi} + \lambda \, \pi_1 + \order(\lambda^2)$ with 
	\begin{align*}
		\pi_1 &=  \Bigl ( - \tfrac{\ii}{2} \nabla_{r} \ln \tfrac{\tau_{\eps}}{\tau_{\mu}} \cdot \sopro{\varphi}{ \Sigma \varphi} + \ii \, \nabla_{r} \ln \tau\cdot \sopro{\varphi}{\nabla_{k} \varphi} \, \bigl ( \Mper(\, \cdot \,) + \omega \bigr ) \Bigr ) \, R_{\omega}^{\perp}
		\; + 
		\\
		&\qquad +
		\, \mathrm{adjoint}
	\end{align*}
	where $R_{\omega}^{\perp}(k) := \pi_0^{\perp} \, \bigl ( \Mper(k) - \omega(k) \bigr )^{-1} \, \pi_0^{\perp}$ is the reduced resolvent of the periodic Maxwell operator. The modified ray optics observable $f_{\mathrm{ro}}$ computes to 
	\begin{align*}
		f_{\mathrm{ro}} = \bscpro{\varphi}{f \varphi}_{\HperT} \, \pi_0 &+ \lambda \, \Bigl ( \bscpro{\varphi \,}{\, \bigl [ f , \pi_1 \bigr ]_+ \varphi}_{\HperT} - \tfrac{\ii}{2} \, \bscpro{\varphi}{\bigl [\nabla_{k} \pi_0 \, , \nabla_{r} f] \varphi}_{\HperT} \Bigr ) \, \pi_0 
		\Bigr .  \\
		& + 
		\lambda \, \Bigl ( 
		\bscpro{\varphi}{f \varphi}_{\HperT} \, \pi_1 
		- \tfrac{\ii}{2} \bscpro{\varphi \,}{\nabla_r f \varphi}_{\HperT} \cdot \bigl [ \nabla_{k} \pi_0 \, , \, \pi_0 \bigr ]
		\Bigr ) 		
	\end{align*}
	where $\bigl [ f , \pi_1 \bigr ]_+ := f \, \pi_1 + \pi_1 \, f$ and $\bigl [ \nabla_k \pi_0 \, , \nabla_r f \bigr ] := \nabla_k \pi_0\cdot \nabla_r f \,  - \nabla_r f \cdot \nabla_k \pi_0$. We point out that 
	unlike the first two terms in the above for $f_{\mathrm{ro}}$ which are proportional to $\pi_0$, the third term is completely offdiagonal with respect to $\pi_0$.
\end{lemma}
The interested reader may find the computation in Appendix~\ref{appendix:computations}.
\begin{remark}\label{ray_optics:rem:specialization_f_ro_selfadjoint}
	Under certain circumstances, we can give more explicit expressions for $f_{\mathrm{ro}}$. Provided 
	\begin{align*}
		f &= \pi_0 \, f \, \pi_0 + \pi_0^{\perp} \, f \, \pi_0^{\perp}
		= \bscpro{\varphi}{f \varphi}_{\HperT} \, \pi_0 + \pi_0^{\perp} \, f \, \pi_0^{\perp}
	\end{align*}
	is block-diagonal with respect to the decomposition induced by $\pi_0$, for example, we can simplify the term involving the anti-commutator $[f , \pi_1]_+$: Because we take the expectation value with respect to $\varphi$, only the block-diagonal part of the anticommutator actually contributes. So if $f$ is block-diagonal, the offdiagonal part of $\pi_1$ only contributes to the offdiagonal part of $[f , \pi_1]_+$, and therefore we can replace $\pi_1$ by $\pi_0 \, \pi_1 \, \pi_0$. 
	%
	%
	In general, however, this is not true as $f$ need not be block-diagonal. 
	
	If $f$ takes values in the selfadjoint operators, the above expression for $f_{\mathrm{ro}}$ simplifies to 
	\begin{align*}
		f_{\mathrm{ro}} &= \bscpro{\varphi}{f \varphi}_{\HperT} \, \pi_0 
		+ \\
		&\qquad 
		+ \lambda \, \Bigl ( 
		2 \, \Re \bscpro{f \varphi \,}{\, \pi_1 \varphi}_{\HperT}
		- \bscpro{\varphi}{\nabla_r f \varphi}_{\HperT} \cdot \mathcal{A} 
		\Bigr . \\
		&\qquad \qquad \quad \Bigl . 
		- \, \Im \bscpro{\varphi}{\nabla_r f \cdot \nabla_k \varphi}_{\HperT}
		\Bigr ) \, \pi_0 
		+ \\
		&\qquad 
		+ \lambda \, \Bigl ( \bscpro{\varphi}{f \varphi}_{\HperT} \, \pi_1 
		+ \bscpro{\varphi}{\nabla_r f \varphi}_{\HperT} \cdot \bigl [ \nabla_k \pi_0 \, , \, \pi_0 \bigr ]
		\Bigr ) 
	\end{align*}
	where $\mathcal{A} := \ii \bscpro{\varphi}{\nabla_k \varphi}_{\HperT}$ is the vector associated with the Berry connection. Typically, the selfadjoint observables of interest are of the form $f = \rho \, W I$ where $\rho$ is a scalar function which localizes on a domain $\Lambda \subset \R^3$ and $I$ is a suitable $6 \times 6$ hermitian matrix. If we think of $\rho$ as a smoothened version of the characteristic function $1_{\Lambda}$, then $\nabla_r f(r) \approx n(r) \, W I$ where $n(r) \in \R^3$ is the external normal to $r \in \partial \Lambda$, and $n(r)=0$ whenever $r \not\in \partial \Lambda$. With this in mind, we can distinguish “bulk”-type contributions to $f_{\mathrm{ro}}$ that are proportional to $\rho$, 
	\begin{align*}
		\bscpro{\varphi}{I \varphi}_{L^2(\T^3,\C^6)} \, \pi_0 + \lambda \, 2 \, \Re \bscpro{ \varphi \,}{\,I \pi_1 \varphi}_{L^2(\T^3,\C^6)} \,\pi_0 +  \bscpro{\varphi}{I \varphi}_{L^2(\T^3,\C^6)} \, \pi_1
		, 
	\end{align*}
	and an $\order(\lambda)$ part of “boundary” type which is localized around $\partial \Lambda$, 
	\begin{align*}
		&- \bscpro{\varphi}{I \varphi}_{L^2(\T^3,\C^6)} \, \bigl( n \cdot \mathcal{A} \bigr ) \, \pi_0 - \bigl ( n \cdot \mathcal{A}^I \bigr )\, \pi_0 
		\, + \\
		&\qquad \qquad \qquad 
		+ \bscpro{\varphi}{I \varphi}_{L^2(\T^3,\C^6)} \, \bigl ( n \cdot \bigl [ \nabla_k \pi_0 \, , \, \pi_0 \bigr ] \bigr )
		,
	\end{align*}
	where $ \mathcal{A}^I := \Im \bscpro{\varphi}{I \nabla_k \varphi}_{L^2(\T^3,\C^6)}$. 
\end{remark}
%

\subsection{The ray optics limit for certain observables} 
\label{ray_optics:examples}
Our main result, Theorem~\ref{ray_optics:thm:ray_optics}, applies directly to a number of physical observables, and we will discuss the local field energy as well as the local average of the Poynting vector in detail. Other examples include local averages of the quadratic components of the fields, the components of the Maxwell-Minkowski stress tensor and Minkowski's electromagnetic momentum (see also \cite[Section~3.3]{Bliokh_Bekshaev_Nori:dual_electromagnetism:2013} for other \emph{in vacuo} observables). 

Throughout this subsection, we abbreviate $\Psi_+ := P_{+,\lambda} (\mathbf{E},\mathbf{H})$ and use $\Psi_+(t) = \e^{- \ii t M_{+,\lambda}} \, \Psi_+$.

\subsubsection{The local field energy} 
\label{ray_optics:examples:energy}
The local field energy is an example of a \emph{scalar} quadratic observable: while Egorov-type theorems do not allow one to infer information on the \emph{pointwise} behavior of the local energy density 
\begin{align*}
	e_x[(\mathbf{E},\mathbf{H})] := \tfrac{1}{2} \Psi_+(x) \cdot W_{+,\lambda}^{-1}(x) \Psi_+(x)
	, 
\end{align*}
it does apply to local averages. Pick any closed set $\Lambda \subset \R^3$ of positive Lebesgue measure. Next, we choose a smoothened characteristic function $\rho \in \Cont^{\infty}_{\mathrm{b}}(\R^3,\R)$, meaning $\rho \vert_{\Lambda} = 1$ and $\rho$ vanishes on $\R^3 \setminus \Lambda^{\delta}$ for some $\delta > 0$ where 
\begin{align*}
	\Lambda^{\delta} := \Bigl \{ r \in \R^3 \; \; \big \vert \; \; \mathrm{dist} (r,\Lambda) < \delta \Bigr \} 
\end{align*}
is a “thickened” version of the set $\Lambda$. Then 
\begin{align*}
	\mathcal{E}_{\rho}[(\mathbf{E},\mathbf{H})] :\negmedspace &= 2 \, \Re \int_{\R^3} \dd x \, \rho(\lambda x) \, e_x[(\mathbf{E},\mathbf{H})] 
	= \Re \, \bscpro{\Psi_+}{\rho(\lambda \hat{x}) \Psi_+}_{L^2_{W_{+,\lambda}}(\R^3,\C^6)} 
	\\
	&
	= \Re \, \bscpro{\Psi_+}{\Op_{\lambda}^{S \Zak}(\rho) \Psi_+}_{L^2_{W_{+,\lambda}}(\R^3,\C^6)} 
\end{align*}
is in good approximation the field energy contained inside of the stretched domain 
\begin{align*}
	\Lambda_{\lambda} := \bigl \{ x \in \R^3 \; \; \vert \; \; \lambda x \in \Lambda \bigr \} 
\end{align*}
provided the “thickness” $\delta$ of the transition layer where $\rho \rightarrow 0$ is small. With this proviso, we will call $\mathcal{E}_{\rho}$ the field energy localized in the volume $\Lambda_{\lambda}$. 

Clearly, $\rho$ defines the scalar, quadratic observable $\mathcal{E}_{\rho}$, and thus, Theorem~\ref{ray_optics:thm:ray_optics} and Corollary~\ref{ray_optics:cor:phase_space_average} apply: for $\Psi_+ \in \ran \Pi_{\lambda}$ we can approximate  
\begin{align}
	\mathcal{E}_{\rho} \bigl [ \bigl ( \mathbf{E}(t) , \mathbf{H}(t) \bigr ) \bigr ] &= \Re \, \bscpro{\Psi_+}{\Op_{\lambda}^{S \Zak} \bigl ( \rho \circ \Phi^{\lambda}_t \bigr ) \Psi_+}_{L^2_{W_{+,\lambda}}(\R^3,\C^6)} + \order(\lambda^2) 
	\label{ray_optics:eqn:ray_optics_limit_energy}
	\\
	&= \Re \int_{\R^3} \dd r \int_{\R^3} \dd k \, \rho \circ \Phi^{\lambda}_t(r,k) \; \mathrm{w}_{\Psi_+}(r,k) + \order(\lambda^2) 
	\notag 
	. 
\end{align}
with the help of the ray optics flow associated to \eqref{ray_optics:eqn:ray_optics_equations_scalar}. 

\subsubsection{The Poynting vector} 
\label{ray_optics:examples:poynting}
In the theory of electromagnetism the Poynting vector 
\begin{align*}
	\mathcal{S}_x[(\mathbf{E},\mathbf{H})] := \tfrac{1}{2} \, \overline{\psi_+^E(x)} \times \psi_+^H(x) 
\end{align*}
is proportional (up to a factor $c^2$) to the \emph{Abraham momentum density}. Indeed, it appears in the local energy conservation law (\cf \cite[equation~(38)]{Bergmann:gyrotropic_Maxwell:1982})
\begin{align}
	\partial_t e_x \bigl [ \bigl ( \mathbf{E}(t) , \mathbf{H}(t) \bigr ) \bigr ] + \nabla_x \cdot \mathcal{S}_x \bigl [ \bigl ( \mathbf{E}(t) , \mathbf{H}(t) \bigr ) \bigr ] = 0 
	\label{ray_optics:eqn:energy_momentum_conservation_law}
\end{align}
as the balancing term to the energy flux. Surprisingly, the three components of $\mathcal{S}$ are linked to what would be called the “current operator” in quantum mechanics, 
\begin{align}
	j_n :\negmedspace &= \tfrac{\ii}{\lambda} \bigl [ M_{\lambda} , \lambda \hat{x}_n \bigr ] 
	= S^{-2}(\lambda \hat{x}) \, W \, \left (
	\begin{matrix}
		0 & - e_n^{\times} \\
		+ e_n^{\times} & 0 \\
	\end{matrix}
	\right ) 
	=: \Op_{\lambda}^{S \Zak}(s_n)
	\label{ray_optics:eqn:current_operator}
	, 
\end{align}
with 
\begin{align*}
	s_n := S^{-1} \, W \, \left (
	\begin{matrix}
		0 & - e_n^{\times} \\
		+ e_n^{\times} & 0 \\
	\end{matrix}
	\right ) \, S^{-1}
	, 
\end{align*}
and a quick computation reveals 
\begin{align*}
	2 \, \Re \, \int_{\R^3} \dd x \, \mathcal{S}_{x,n}[(\mathbf{E},\mathbf{H})] &= 2 \, \Re \, \bscpro{\Psi_+ \, }{ \, j_n \Psi_+}_{L^2_{W_{+,\lambda}}(\R^3,\C^6)} 
	\\
	&= \Re \, \bscpro{\Psi_+}{\Op_{\lambda}^{S \Zak}(s_n) \Psi_+}_{L^2_{W_{+,\lambda}}(\R^3,\C^6)} 
	. 
\end{align*}
As one can see right away, even when the perturbation is scalar $s_n$ defines a \emph{non-scalar observable}. 

There are in fact \emph{two} interesting quantities connected to the Poynting vector, the \emph{net flux} across $\partial \Lambda_{\lambda}$ as well as the local average of $\mathcal{S}$ across $\Lambda_{\lambda}$. The first can be accessed via $\mathcal{E}_{\rho}$ with the help of local energy conservation~\eqref{ray_optics:eqn:energy_momentum_conservation_law}: Taking the time-derivative of $\mathcal{E}_{\rho} \bigl [ \bigl ( \Psi_+(t) \bigr ) \bigr ]$ approximately yields the net momentum flux over the “surface” $\partial \Lambda_{\lambda}$, 
\begin{align*}
	\frac{\dd}{\dd t} \mathcal{E}_{\rho} \bigl [ \bigl ( \mathbf{E}(t) , \mathbf{H}(t) \bigr ) \bigr ] &= \Re \, \scpro{\Psi_+ \, }{ \, \frac{\dd}{\dd t} \Bigl ( \e^{+ \ii \frac{t}{\lambda} M_{\lambda}} \, \rho(\lambda \hat{x}) \, \e^{- \ii \frac{t}{\lambda} M_{\lambda}} \Bigr ) \Psi_+}_{L^2_{W_{+,\lambda}}(\R^3,\C^6)} 
	\\
	&= \Re \, \Bscpro{\Psi_+ \, }{ \, \nabla \rho \bigl ( \lambda \hat{x}(t) \bigr ) \cdot j(t) \, \Psi_+}_{L^2_{W_{+,\lambda}}(\R^3,\C^6)}
	, 
\end{align*}
because the support of the derivative $\mathrm{supp} \, \nabla \rho \subseteq \overline{\Lambda_{\lambda}^{\delta} \setminus \Lambda_{\lambda}}$ is contained in the “boundary layer” of $\Lambda_{\lambda}^{\delta}$ for $\delta$ sufficiently small which is a thickened version of the boundary $\partial \Lambda_{\lambda}$. 

Now the ray optics limit for \emph{scalar} observables applies to the energy contained in $\Lambda_{\lambda}$. If we assume that the time derivative of the error term in equation~\eqref{ray_optics:eqn:ray_optics_limit_energy} is \emph{still} of $\order(\lambda^2)$, then we can find a semiclassical expression for the net energy flow, 
\begin{align}
	\frac{\dd}{\dd t} \mathcal{E}_{\rho} \bigl [ \bigl ( \mathbf{E}(t) , \mathbf{H}(t) \bigr ) \bigr ] &= \Re \, \Bscpro{\Psi_+ \, }{ \, \nabla \rho \bigl ( \lambda \hat{x}(t) \bigr ) \cdot j(t) \, \Psi_+}_{L^2_{W_{+,\lambda}}(\R^3,\C^6)}
	\notag \\
	&= \Re \, \Bscpro{\Psi_+ \, }{ \, \Op_{\lambda}^{S \Zak} \Bigl ( \tfrac{\dd}{\dd t} \rho \circ \Phi^{\lambda}_t \Bigr ) \Psi_+}_{L^2_{W_{+,\lambda}}(\R^3,\C^6)} + \order(\lambda^2) 
	\notag \\
	&= \Re \, \Bscpro{\Psi_+ \, }{ \, \Op_{\lambda}^{S \Zak} \Bigl ( \nabla_r \rho \bigl ( r(t) \bigr ) \cdot \dot{r}(t) \Bigr ) \Psi_+}_{L^2_{W_{+,\lambda}}(\R^3,\C^6)} 
	\, + \notag \\
	&\qquad 
	+ \order(\lambda^2) 
	. 
	\label{ray_optics:eqn:time_derivative_local_energy_ray_optics_limit}
\end{align}
Using formal arguments, we see that these results are consistent with the local energy conservation law~\eqref{ray_optics:eqn:energy_momentum_conservation_law}: 
\begin{align*}
	\frac{\dd}{\dd t} \mathcal{E}_{\rho} &\bigl [ \bigl ( \mathbf{E}(t) , \mathbf{H}(t) \bigr ) \bigr ] \approx 
	\\
	&\approx
	2 \, \Re \, \int_{\Lambda} \dd x \, \partial_t e \bigl ( x , \Psi_+(t) \bigr ) 
	= - 2 \, \Re \, \int_{\Lambda} \dd x \, \nabla_x \cdot \mathcal{S} \bigl ( x , \Psi_+(t) \bigr ) 
	\\
	&= - 2 \, \Re \, \int_{\partial \Lambda} \dd \eta(x) \cdot \mathcal{S} \bigl ( x , \Psi_+(t) \bigr ) 
	\approx 2 \, \Re \, \int_{\R^3} \dd x \, \nabla_x \rho(\lambda x) \cdot \mathcal{S} \bigl ( x , \Psi_+(t) \bigr ) 
	\\
	&
	= \Re \, \Bscpro{\Psi_+ \, }{ \, \nabla \rho \bigl ( \lambda \hat{x}(t) \bigr ) \cdot j(t) \, \Psi_+}_{L^2_{W_{+,\lambda}}(\R^3,\C^6)}
\end{align*}
where $\dd \eta(x)$ is the measure on $\partial \Lambda$ with surface normal pointing outwards. Let us point out that to make this “heuristic” argument rigorous, a more in-depth analysis of the error term is necessary; But this is beyond the scope of this paper. 

The field momentum inside of $\Lambda$ is accessible via Theorem~\ref{ray_optics:thm:ray_optics}~(ii), 
\begin{align*}
	\mathcal{S}_{\rho,n} \bigl [ \bigl ( \mathbf{E}(t) , \mathbf{H}(t) \bigr ) \bigr ] :& \negmedspace = \Re \, \Bscpro{\Psi_+(t) \,}{\, \Op_{\lambda}^{S \Zak}(\rho \, s_n) \, \Psi_+(t)}_{L^2_{W_{+,\lambda}}(\R^3,\C^6)} 
	\\
	&= \Re \, \Bscpro{\Psi_+ \,}{\, \Op_{\lambda}^{S \Zak} \bigl ( f_{\mathrm{ro}} \circ \Phi_t^{\lambda} ) \, \Psi_+}_{L^2_{W_{+,\lambda}}(\R^3,\C^6)} + \order(\lambda^2)
	, 
\end{align*}
although we need to replace $f := \rho \, s_n$ with $f_{\mathrm{ro}} = \pi_{\lambda} \Weyl f \Weyl \pi_{\lambda} + \order(\lambda^2)$ and the flow $\Phi^{\lambda}$ of \eqref{ray_optics:eqn:ray_optics_equations_scalar} by that associated to the ray optics equations~\eqref{ray_optics:eqn:ray_optics_equations_non-scalar} which omit the Berry curvature in the symplectic form. Instead, several terms that are linked to the geometry of the Bloch bundle appear at $\order(\lambda)$ in $f_{\mathrm{ro}}$. 

\subsubsection{Other quadratic observables relevant in electrodynamics} 
\label{ray_optics:examples:others}
At least four more observables, all of them non-scalar, fit into the category of quadratic observables once they are localized by a smoothened characteristic function $\rho$. We leave the details such as finding the appropriate operator-valued function to the reader. 

The \emph{averaged quadratic component of the electric field}
\begin{align*}
	\abs{E_{\rho,n}}^2 := \int_{\R^3} \dd x \, \rho(\lambda x) \, \sabs{E_n(x)}^2
\end{align*}
and a similar expression for the magnetic field falls into the category set forth by Definition~\ref{ray_optics:defn:admissible_observables}. 

Apart from the local averages of the Poynting vector $\mathcal{S}(x,\Psi_+)$ and of the related \emph{Abraham} momentum density $\mathcal{G}^{\mathrm{A}}(x,\Psi_+) := c^{-2} \, \mathcal{S}(x,\Psi_+)$, for the case $\chi = 0$ there is a second momentum observable in electromagnetism, the \emph{Minkowski} momentum density
\begin{align*}
	\mathcal{G}^{\mathrm{M}}_x[(\mathbf{E},\mathbf{H})] := \frac{1}{2} \, \tau^{-4}(\lambda x) \, \bigl ( \overline{\eps(x) \psi_+^E(x)} \bigr ) \times \bigl ( \mu(x) \psi_+^H(x) \bigr ) 
	.
\end{align*}
The relation between the Abraham and the  Minkowski momentum densities as well as their physical interpretation are delicate topics in classical electrodynamics known as the \emph{Abraham-Minkowski controversy} (see \eg \cite{Pfeifer_Nieminen_et_al:momentum_electromagnetic_wave:2007}). 

Similarly, local averages of the components of 
\emph{angular momentum} (defined with respect to either Abraham or Minkowski momentum density)
\begin{align*}
	\mathcal{L}_{\rho}^{\mathrm{A}/\mathrm{M}}[(\mathbf{E},\mathbf{H})] := 2 \, \Re \, \int_{\R^3} \dd x \, \rho(\lambda x) \; \Bigl ( x \times \mathcal{G}^{\mathrm{A}/\mathrm{M}}(x,\Psi_+) \Bigr )
\end{align*}
as well as the components of the components of the \emph{Maxwell stress tensor} (for $\chi = 0$)
\begin{align*}
	\mathcal{T}_{\rho}^{j,n}[(\mathbf{E},\mathbf{H})] :\negmedspace &= 2 \, \Re \, \int_{\R^3} \dd x \, \rho(\lambda x) \, \Bigl ( \tau_{\eps}^{-2}(\lambda x) \; \overline{\psi_{+,j}^E(x)} \; \bigl ( \eps(x) \psi_+^E(x) \bigr )_n 
	+ \Bigr . \\
	&\qquad \qquad \quad \Bigl . 
	+ \, \tau_{\mu}^{-2}(\lambda x) \; \overline{\psi_{+,j}^H(x)} \; \bigl ( \mu(x) \psi_+^H(x) \bigr )_n - \delta_{j,n} \, e(x,\Psi_+) \Bigr ) 
\end{align*}
are other examples of quadratic observables covered by Theorem~\ref{ray_optics:thm:ray_optics}. To each one of those quadratic observables one can associate a symbol similar to the form considered in Remark~\ref{ray_optics:rem:specialization_f_ro_selfadjoint} as the reader can easily verify.

\section{An Egorov-type theorem} 
\label{egorov}
The main ingredients in the proof of the ray optics limit, Theorem~\ref{ray_optics:thm:ray_optics}, are two Egorov theorems, one for scalar and one for non-scalar observables. We first treat the scalar case in detail, and then proceed to the non-scalar case where we only discuss the necessary modifications.

\subsection{Simplifying the notation} 
\label{egorov:simplifying}
To make the formulae easier on the eyes, we will take some steps to unburden to notation: 
\begin{enumerate}[(i)]
	\item We will systematically drop the index “$+$”, \eg $W_+$ becomes $W$. 
	\item Instead of $M_{+,\lambda}$ on $\Hil_{+,\lambda}$ we will consider $M_{\lambda} := S(\lambda \hat{x}) \, W_+ \, \Rot$ on \emph{all} of $\Hil_{\lambda} := L^2_{W_{+,\lambda}}(\R^3,\C^6)$; the restriction to the positive frequency subspace $\Hil_{+,\lambda}$ will be implemented by sandwiching operators in between the projection $\Pi_{\lambda}$ constructed in \cite[Proposition~1]{DeNittis_Lein:sapt_photonic_crystals:2013} associated to the chosen positive frequency band. $\Pi_{\lambda}$ automatically satisfies equation~\eqref{setup:eqn:relation_Pi_P} by construction, and hence, up to $\order(\lambda^{\infty})$ consist only of $\omega > 0$ states. 
	\item For any $F \in \mathcal{B} \bigl ( L^2_{W_{+,\lambda}}(\R^3,\C^6) \bigr )$ we can view $P_{+,\lambda} \, F \, P_{+,\lambda}$ as a bounded operator on $\Hil_{+,\lambda}$, and the simple estimate 
	\begin{align*}
		\bnorm{P_{+,\lambda} \, F \, P_{+,\lambda}}_{\mathcal{B}(\Hil_{+,\lambda})}
		\leq \snorm{F}_{\mathcal{B}(L^2_{W_{+,\lambda}}(\R^3,\C^6))}
		,
	\end{align*}
	allows us to push operator norm estimates from $\mathcal{B} \bigl ( L^2_{W_{+,\lambda}}(\R^3,\C^6) \bigr )$ to $\mathcal{B}(\Hil_{+,\lambda})$. 
\end{enumerate}
%

\subsection{An Egorov theorem for scalar observables} 
\label{ray_optics:scalar}
For this simpler class of \emph{scalar} observables, we can directly apply the results of Stiepan and Teufel \cite{Stiepan_Teufel:semiclassics_op_valued_symbols:2012}. The main technical advantage of their technique compared to earlier works such as \cite{PST:effective_dynamics_Bloch:2003,DeNittis_Lein:Bloch_electron:2009} is that they do \emph{not} need to assume the triviality of the Bloch bundle. 
\begin{proposition}[Egorov theorem for scalar observables]\label{ray_optics:prop:ray_optics_limit}
	Suppose we are in the setting of Theorem~\ref{ray_optics:thm:ray_optics}~(i). Then for any scalar observable associated to $f \in \Cont^{\infty}_{\mathrm{b}}(\R^6,\C)$ which is periodic in $k$, the full light dynamics can be approximated by ray optics for bounded times, \ie for all $T > 0$ we have 
	\begin{align}
		&\sup_{t \in [-T,+T]} \Bnorm{\Pi_{\lambda} \Bigl ( \e^{+ \ii \frac{t}{\lambda} M_{\lambda}} \, \Op_{\lambda}^{S \Zak}(f) \, \e^{- \ii \frac{t}{\lambda} M_{\lambda}} - \Op_{\lambda}^{S \Zak} \bigl ( f \circ \Phi^{\lambda}_t \bigr ) \Bigr ) \Pi_{\lambda}}_{\mathcal{B}(\Hil_{\lambda})} 
		= \notag \\
		&\qquad \qquad 
		= \order(\lambda^2) 
		. 
		\label{ray_optics:eqn:ray_optics_Egorov}
	\end{align}
\end{proposition}
To help separate computations from technical arguments, we start with the following 
\begin{lemma}\label{ray_optics:eqn:computation_ray_optics_Maxwellian}
	Suppose we are in the setting of Proposition~\ref{ray_optics:prop:ray_optics_limit}. Then in both cases ($\chi = 0$ or $\chi \neq 0$ and $\tau_{\eps} = \tau_{\mu}$) the dispersion relation $\Omega$ characterized by  
	\begin{align}
		\pi_{\lambda} \Weyl \bigl ( \Omega - \Msymb_{\lambda} \bigr ) \Weyl \pi_{\lambda} = \order(\lambda^2) 
		\label{ray_optics:eqn:defining_relation_M_ro}
	\end{align}
	computes to be \eqref{ray_optics:eqn:ray_optics_Maxwellian}. 
\end{lemma}
\begin{proof}
	While the final result holds true for both cases, $\chi = 0$ and $\chi \neq 0$, we detail the computations for $\chi = 0$ where electric permittivity and magnetic permeability may be scaled separately. In case $\chi \neq 0$ we set $\tau_{\eps} = \tau_{\mu}$, and all terms which contain gradients of the ratio $\nicefrac{\tau_{\eps}}{\tau_{\mu}} = 1$ vanish. The explicit expression for the dispersion relation 
	%
	\begin{align}
		\Omega(r,k) = \trace_{\HperT} \bigl ( \Msymb_{\lambda}(r,k) \, \pi_0(k) \bigr ) + \lambda \, \Omega_{\partial \mathcal{M}}
		\label{ray_optics:eqn:M_ro}
	\end{align}
	%
	is determined by equations~(17) and (18) in \cite{Stiepan_Teufel:semiclassics_op_valued_symbols:2012}, and consists of two parts. The first contribution is the expectation value of the symbol. The second, $\Omega_{\partial \mathcal{M}} := \trace_{\HperT} \bigl ( \bigl \{ \pi_0 \vert \Msymb_0 \vert \pi_0 \bigr \} \bigr ) = 0$, vanishes in our case for the same reason as in equation~\eqref{ray_optics:eqn:Poisson_bracket_of_three} -- $\pi_0(k)$ depends only on crystal momentum. 
	
	The trace terms are merely a fancy way to write the expectation value with respect to $\varphi$. Clearly, the leading-order term 
	\begin{align*}
		\Omega_0(r,k) &= \trace_{\HperT} \bigl ( \Msymb_0(r,k) \, \pi_0(k) \bigr ) 
		\\
		&
		= \scpro{\varphi(k)}{\tau^2(r) \, \Mper(k) \varphi(k)}_{\HperT} 
		= \tau^2(r) \, \omega(k)
	\end{align*}
	is just the band function scaled by $\tau$. For the sub-leading term we first compute 
	\begin{align*}
		&\scpro{\left (
		\begin{matrix}
			\varphi^E(k) \\
			\varphi^H(k) \\
		\end{matrix}
		\right ) \; }{ \; W \, \left (
		\begin{matrix}
			0 & e_j^{\times} \\
			e_j^{\times} & 0 \\
		\end{matrix}
		\right ) \left (
		\begin{matrix}
			\varphi^E(k) \\
			\varphi^H(k) \\
		\end{matrix}
		\right )}_{\HperT}
		= \\
		&\qquad \qquad \qquad 
		= \scpro{\left (
		\begin{matrix}
			\varphi^E(k) \\
			\varphi^H(k) \\
		\end{matrix}
		\right ) \; }{ \; \left (
		\begin{matrix}
			e_j \times \varphi^H(k) \\
			e_j \times \varphi^E(k) \\
		\end{matrix}
		\right )}_{L^2(\T^3,\C^6)}
		\\
		&\qquad \qquad \qquad = - \int_{\T^3} \dd y \, e_j \cdot \Bigl ( \overline{\varphi^E(k,y)} \times \varphi^H(k,y) - \varphi^E(k,y) \times \overline{\varphi^H(k,y)} \Bigr ) 
		\\
		&\qquad \qquad \qquad = - \ii \, 2 \, \mathcal{P}_j(k) 
		. 
	\end{align*}
	This now yields 
	\begin{align*}
		\Omega_1 &= \trace \bigl ( \Msymb_1 \, \pi_0 \bigr ) 
		= - \tau^2 \, \frac{\ii}{2} \, \sum_{j = 1}^3 \partial_{r_j} \ln \tfrac{\tau_{\eps}}{\tau_{\mu}} \; 
		\scpro{\left (
		\begin{matrix}
			\varphi^E \\
			\varphi^H \\
		\end{matrix}
		\right )}{W \, \left (
		\begin{matrix}
			0 & e_j^{\times} \\
			e_j^{\times} & 0 \\
		\end{matrix}
		\right ) \left (
		\begin{matrix}
			\varphi^E \\
			\varphi^H \\
		\end{matrix}
		\right )}_{\HperT}
		\\
		&= - \tau^2 \, \mathcal{P} \cdot \nabla_r \ln \tfrac{\tau_{\eps}}{\tau_{\mu}}
		. 
	\end{align*}
	When $\chi \neq 0$ the perturbation $S(r) = \tau^{-1}(r)$ is scalar. That means the last term in $\Msymb_{\lambda}(r,k) = \tau^2(r) \, \Mper(k) + 0$ vanishes. Seeing as $\Msymb_1$ does not enter in the computation of the term given by \cite[equation~(18)]{Stiepan_Teufel:semiclassics_op_valued_symbols:2012}, we immediately deduce $\Omega_1 = 0$. 
\end{proof}
\begin{proof}[Proposition~\ref{ray_optics:prop:ray_optics_limit}]
	The modifications to the proofs in \cite{Stiepan_Teufel:semiclassics_op_valued_symbols:2012} are of purely technical nature. Nevertheless, for the benefit of the reader we will sketch the general strategy of Stiepan and Teufel's work, and explain the necessary modifications.

	\paragraph{Notation} 
	\label{par:notation}
	Given the quantum mechanical context their notation is different and clashes with ours: Stiepan and Teufel consider a hamiltonian (operator) $\hat{H}$ with symbol $H = H_0 + \eps \, H_1$ which corresponds to the Maxwell operator $M_{\lambda}$ and its symbol $\Msymb_{\lambda} = \Msymb_0 + \lambda \, \Msymb_1$. The relevant symbol classes such as $\SemiHoermeq{m}{0} \bigl ( \mathcal{B}(\mathfrak{h}_1,\mathfrak{h}_2) \bigr )$ are defined in Definition~\ref{appendix:PsiDOs:defn:semiclassical_symbols}. The analog of the semiclassical hamiltonian $h = h_0 + \eps \, h_1$ is the dispersion relation $\Omega$, and to avoid a notational clash we have renamed the components of the  extended Berry curvature as given by \cite[equation~(23)]{Stiepan_Teufel:semiclassics_op_valued_symbols:2012} to $\Xi^{kk}$, $\Xi^{rk}$, $\Xi^{kr}$ and $\Xi^{rr}$. At this point we have already obtained the explicit expressions of the dispersion relation in Lemma~\ref{ray_optics:eqn:computation_ray_optics_Maxwellian}. We need to verify that Proposition~2, Proposition~3 and Theorem~2 in \cite{Stiepan_Teufel:semiclassics_op_valued_symbols:2012} can be extended to the case of the slowly modulated periodic Maxwell operator. 

	\paragraph{Facts on the Maxwell operator and the superadiabatic projection} 
	\label{par:facts_on_the_maxwell_operator_and_the_superadiabatic_projection}
	First, the Maxwell operator is unbounded and defined in terms of an equivariant symbol 
	\begin{align*}
		\Msymb_{\lambda} \in \SemiHoermeq{1}{1} \bigl ( \mathcal{B} \bigl ( \domainT,L^2(\T^3,\C^6) \bigr ) \bigr )
	\end{align*}
	where $\domainT$ defined in \cite[equation~(32)]{DeNittis_Lein:adiabatic_periodic_Maxwell_PsiDO:2013} is the domain of the periodic Maxwell operator $\Mper(k)$ and $\Msymb_{\lambda}$ is given by equation~\eqref{setup:eqn:symbol_Maxwellian} (\cf \cite[Corollary~4.3]{DeNittis_Lein:adiabatic_periodic_Maxwell_PsiDO:2013}). Moreover, from \cite[Proposition~1]{DeNittis_Lein:sapt_photonic_crystals:2013} we know the superadiabatic projection $\Pi_{\lambda} = \Op_{\lambda}^{S \Zak}(\pi_{\lambda}) + \order_{\norm{\cdot}}(\lambda^{\infty})$ associated to an isolated band exists and is $\order(\lambda^{\infty})$-close in norm to a $\Psi$DO with symbol 
	\begin{align*}
		\pi_{\lambda} \in \SemiHoermeq{0}{0} \bigl ( \mathcal{B} \bigl ( L^2(\T^3,\C^6) \bigr ) \bigr ) \cap \SemiHoermeq{1}{0} \bigl ( \mathcal{B} \bigl (L^2(\T^3,\C^6) , \domainT \bigr ) \bigr ) 
		. 
	\end{align*}
	As explained in Appendix~\ref{appendix:PsiDOs} equivariance is preserved by the Weyl product. Moreover, all of the error terms below are in $\SemiHoermeq{0}{0} \bigl ( \mathcal{B} \bigl ( L^2(\T^3,\C^6) \bigr ) \bigr )$. 

	\paragraph{Step 1: Pull the projection into the commutator} 
	\label{par:step_1_pull_the_projection_into_the_commutator}
	A simple computation yields
	\begin{align*}
		\pi_{\lambda} \Weyl \bigl [ \Msymb_{\lambda} , f \bigr ]_{\Weyl} \Weyl \pi_{\lambda} = \Bigl [ \pi_{\lambda} \Weyl \Msymb_{\lambda} \Weyl \pi_{\lambda} \; , \; \pi_{\lambda} \Weyl f \Weyl \pi_{\lambda} \Bigr ]_{\Weyl} + \order(\lambda^{\infty})
		, 
	\end{align*}
	and all we need to check is that all the terms are in $\SemiHoermeq{0}{0} \bigl ( \mathcal{B} \bigl ( L^2(\T^3,\C^6) \bigr ) \bigr )$ which then quantize to bounded operators by a variant of the Caldéron-Vaillancourt theorem (\cf the discussion in \cite[Section~4.1]{DeNittis_Lein:adiabatic_periodic_Maxwell_PsiDO:2013} and \cite[Proposition~B.5]{Teufel:adiabatic_perturbation_theory:2003}): \emph{a priori} the left-hand side is an element of the space $\SemiHoermeq{2}{0} \bigl ( \mathcal{B} \bigl ( L^2(\T^3,\C^6) \bigr ) \bigr )$ by the composition properties of symbols, but the equivariance condition implies that for any $m > 0$ we have in fact 
	\begin{align*}
		\SemiHoermeq{m}{0} \bigl ( \mathcal{B} \bigl ( L^2(\T^3,\C^6) \bigr ) \bigr ) = \SemiHoermeq{0}{0} \bigl ( \mathcal{B} \bigl ( L^2(\T^3,\C^6) \bigr ) \bigr )
		. 
	\end{align*}
	%

	\paragraph{Step 2: Replace $\Msymb_{\lambda}$ by $\Omega$} 
	\label{par:step_2_replace_msymb__lambda_by_msymb__mathrm_ro}
	Adapting the arguments from \cite[Proposition~2]{Stiepan_Teufel:semiclassics_op_valued_symbols:2012} readily yields 
	\begin{align}
		\pi_{\lambda} \Weyl \bigl ( \Msymb_{\lambda} - \Omega \bigr ) \Weyl \pi_{\lambda} &=
		\pi_0 \, \Bigl ( \pi_{\lambda} \Weyl \bigl ( \Msymb_{\lambda} - \Omega \bigr ) \Weyl \pi_{\lambda} \Bigr ) \, \pi_0 + \order(\lambda^3) = \order(\lambda^2) 
		. 
		\label{ray_optics:eqn:replace_M_lambda_M_ro}
	\end{align}
	As argued in Step~1 above, left- and right-hand side are in the symbol space $\SemiHoermeq{0}{0} \bigl ( \mathcal{B} \bigl ( L^2(\T^3,\C^6) \bigr ) \bigr )$. The computation (after making the necessary changes in notation) is identical, and one gets \eqref{ray_optics:eqn:replace_M_lambda_M_ro}. Consequently, we obtain 
	\begin{align*}
		\Bigl [ \pi_{\lambda} \Weyl \Msymb_{\lambda} \Weyl \pi_{\lambda} \; , \; \pi_{\lambda} \Weyl f \Weyl \pi_{\lambda} \Bigr ]_{\Weyl} &= \Bigl [ \pi_{\lambda} \Weyl \Omega \Weyl \pi_{\lambda} \; , \; \pi_{\lambda} \Weyl f \Weyl \pi_{\lambda} \Bigr ]_{\Weyl} + 
		\\
		&\qquad + 
		\Bigl [ \pi_{\lambda} \Weyl \bigl ( \Msymb_{\lambda} - \Omega \bigr ) \Weyl \pi_{\lambda} \; , \; \pi_{\lambda} \Weyl f \Weyl \pi_{\lambda} \Bigr ]_{\Weyl}
		\\
		&= \Bigl [ \pi_{\lambda} \Weyl \Omega \Weyl \pi_{\lambda} \; , \; \pi_{\lambda} \Weyl f \Weyl \pi_{\lambda} \Bigr ]_{\Weyl} + \order(\lambda^3)
		. 
	\end{align*}
	%

	\paragraph{Step 3: Pull the projection out of the commutator} 
	\label{par:step_3_pull_the_projection_out_of_the_commutator}
	Then after replacing $\Msymb_{\lambda}$ with the dispersion relation $\Omega$ we pull the projection back out of the commutator, 
	\begin{align}
		\Bigl [ \pi_{\lambda} \Weyl \Omega \Weyl \pi_{\lambda} \; , \; \pi_{\lambda} \Weyl f \Weyl \pi_{\lambda} \Bigr ]_{\Weyl} &= \pi_{\lambda} \Weyl \bigl [ \Omega \, , \, f \bigr ]_{\Weyl} \Weyl \pi_{\lambda} 
		\label{ray_optics:eqn:proof_ray_optics_step_3}
		+ \\
		&\qquad
		- \lambda^2 \, \ii \, \pi_{\lambda} \Weyl \Bigl [ \tfrac{\ii}{\lambda} \bigl [ \Omega \, , \, \pi_{\lambda} \bigr ]_{\Weyl} \; , \; \tfrac{\ii}{\lambda} \bigl [ f \, , \, \pi_{\lambda} \bigr ]_{\Weyl} \Bigr ]_{\Weyl} \Weyl \pi_{\lambda} 
		, 
		\notag
	\end{align}
	although at the expense of an extra $\order(\lambda^2)$ term. Note that the equality is exact. 

	\paragraph{Step 4: Approximate commutator with $\lambda$-corrected Poisson bracket} 
	\label{par:step_4_approximate_commutator_with_eps_corrected_poisson_bracket}
	Now we develop all Moyal commutators in $\lambda$, keeping only terms up to $\order(\lambda^2)$: since $\Omega$ and $f$ are scalar, the even powers in the Moyal commutator 
	\begin{align*}
		\bigl [ \Omega \, , \, f \bigr ]_{\Weyl} = - \lambda \, \ii \, \bigl \{ \Omega \, , \, f \bigr \} + \order(\lambda^3)
	\end{align*}
	vanish. For the other two commutators, it suffices to keep only the leading-order term. Thus, after replacing $\Msymb_{\lambda}$ by $\Omega$ in $\pi_{\lambda} \Weyl \bigl [ \Msymb_{\lambda} , f \bigr ]_{\Weyl} \Weyl \pi_{\lambda}$, and replacing the Moyal commutators with Poisson brackets at the expense of an $\order(\lambda^2)$ error, we can write
	\begin{align}
		\tfrac{\ii}{\lambda} \pi_{\lambda} \Weyl \bigl [ \Msymb_{\lambda} \, , \, f \bigr ]_{\Weyl} \Weyl \pi_{\lambda} &= \pi_{\lambda} \Weyl \Bigl ( \bigl \{ \Omega \, , \, f \bigr \} - \lambda \, \ii \Bigl [ \bigl \{ \Omega \, , \, \pi_0 \bigr \} , \bigl \{ f \, , \, \pi_0 \bigr \} \Bigr ] \Bigr ) \Weyl \pi_{\lambda} + \order(\lambda^2) 
		\notag \\
		&
		= \pi_{\lambda} \Weyl \bigl \{ \Omega \, , \, f \bigr \}_{\lambda} \Weyl \pi_{\lambda} + \order(\lambda^2)
		\label{ray_optics:eqn:Moyal_commutator_equal_lambda_Poisson_bracket}
	\end{align}
	in terms of a $\lambda$-corrected Poisson bracket 
	\begin{align}
		\bigl \{ \Omega \, , \, f \bigr \}_{\lambda} := X_{\Omega} \cdot \nabla f
		:= \left (
		\begin{matrix}
			- \lambda \, \Xi^{kk} & + \id + \lambda \, \Xi^{kr} \\
			- \id + \lambda \, \Xi^{rk} & - \lambda \, \Xi^{rr} \\
		\end{matrix}
		\right ) \left (
		\begin{matrix}
			\nabla_r \Omega \\
			\nabla_k \Omega \\
		\end{matrix}
		\right ) \cdot \left (
		\begin{matrix}
			\nabla_r f \\
			\nabla_k f \\
		\end{matrix}
		\right )
		. 
		\label{ray_optics:eqn:lambda_Poisson_bracket}
	\end{align}
	The explicit formula for the modified symplectic form (whose $\order(\lambda)$ contribution is also called \emph{extended Berry curvature}), \cite[equation~(23)]{Stiepan_Teufel:semiclassics_op_valued_symbols:2012}, simplifies tremendously since $\pi_0$ depends only on $k$: the terms which involve derivatives of $\pi_0$ with respect to $r$ vanish, \ie $\Xi^{rr} = 0$ and $\Xi^{rk} = 0 = \Xi^{kr}$. Thus, only the ordinary Berry curvature survives and we obtain the usual Berry curvature for the remaining contribution, $\Xi^{kk} = \Xi$. 

	\paragraph{Step 5: A Duhamel argument} 
	\label{par:step_5_duhamel_argument}
	The ray optics equations~\eqref{ray_optics:eqn:ray_optics_equations_scalar} which define the flow $\Phi^{\lambda}$ can alternatively be written as
	\begin{align*}
		\dot{r}_j &= \bigl \{ \Omega , r_j \bigr \}_{\lambda}
		\\
		\dot{k}_j &= \bigl \{ \Omega , k_j \bigr \}_{\lambda}
	\end{align*}
	where $\{ \, \cdot \, , \, \cdot \, \}_{\lambda}$ is the Poisson bracket defined in equation~\eqref{ray_optics:eqn:lambda_Poisson_bracket} above. Thus, observables evolve according to $\frac{\dd}{\dd t} f \circ \Phi^{\lambda}_t = \bigl \{ \Omega , f \circ \Phi^{\lambda}_t \bigr \}_{\lambda}$. 
	Since the components of the hamiltonian vector field $X_{\Omega}$ are bounded functions with bounded derivatives to any order, the Picard-Lindelöf theorem tells us that the associated ray optics flow $\Phi^{\lambda}$ exists globally in time, and has bounded derivatives to any order (see \eg \cite[Lemma~IV.9]{Robert:tour_semiclassique:1987}). Consequently, also the time-evolved observable $f \circ \Phi^{\lambda}_t$ is a symbol in $\SemiHoermper{0}{0}(\C)$. 
	
	Now the claim follows from a standard Duhamel argument: the difference in time evolutions can be related to the Moyal commutator on the left-hand side of \eqref{ray_optics:eqn:Moyal_commutator_equal_lambda_Poisson_bracket}, 
	\begin{align*}
		\Pi_{\lambda} \, \Bigl ( &\e^{+ \ii \frac{t}{\lambda} M_{\lambda}} \, \Op_{\lambda}^{\Zak}(f) \, \e^{- \ii \frac{t}{\lambda} M_{\lambda}} - \Op_{\lambda}^{\Zak} \bigl ( f \circ \Phi^{\lambda}_t \bigr ) \Bigr ) \, \Pi_{\lambda}
		= \\
		&= \int_0^t \dd s \, \frac{\dd}{\dd s} \Pi_{\lambda} \, \Bigl ( \e^{+ \ii \frac{s}{\lambda} M_{\lambda}} \, \Op_{\lambda}^{\Zak} \bigl ( f \circ \Phi^{\lambda}_{t-s} \bigr ) \, \e^{- \ii \frac{s}{\lambda} M_{\lambda}} \Bigr ) \, \Pi_{\lambda}
		\\
		&= \int_0^t \dd s \, \e^{+ \ii \frac{s}{\lambda} M_{\lambda}} \, \Pi_{\lambda} \, \Op_{\lambda}^{\Zak} \Bigl ( \tfrac{\ii}{\lambda} \bigl \{ \Omega \, , \, f \circ \Phi^{\lambda}_{t-s} \bigr \}_{\lambda} - \tfrac{\dd}{\dd t} f \circ \Phi^{\lambda}_{t-s} \Bigr ) \, \Pi_{\lambda} \, \e^{- \ii \frac{s}{\lambda} M_{\lambda}} 
		\, + \\
		&\qquad 
		+ \order_{\norm{\cdot}}(\lambda^2)
		. 
	\end{align*}
	This concludes the proof. 
\end{proof}
%

\subsection{The case of non-scalar observables} 
\label{ray_optics:non-scalar}
Even if observables are not scalar, one can still derive an Egorov theorem by slightly modifying the proof of Proposition~\ref{ray_optics:prop:ray_optics_limit}. Here, the main idea is to evolve $f_{\mathrm{ro}}$ which is obtained by truncating the expansion of $\pi_{\lambda} \Weyl f \Weyl \pi_{\lambda}$ after the first order. 
\begin{proposition}[Egorov theorem for non-scalar observables]\label{ray_optics:prop:ray_optics_limit_non-scalar}
	Suppose we are in the setting of Theorem~\ref{ray_optics:thm:ray_optics}~(ii). Then for all $f \in \Cont^{\infty}_{\mathrm{b}} \bigl ( \R^6 , \mathcal{B}(\HperT) \bigr )$ satisfying the equivariance condition~\eqref{setup:eqn:equivariance} the full light dynamics can be approximated by ray optics for bounded times, \ie for all $T > 0$ we have 
	\begin{align}
		&\sup_{t \in [-T,+T]} \Bnorm{\Pi_{\lambda} \Bigl ( \e^{+ \ii \frac{t}{\lambda} M_{\lambda}} \, \Op_{\lambda}^{S \Zak}(f) \, \e^{- \ii \frac{t}{\lambda} M_{\lambda}} - \Op_{\lambda}^{S \Zak} \bigl ( f_{\mathrm{ro}} \circ \Phi^{\lambda}_t \bigr ) \Bigr ) \, \Pi_{\lambda}}_{\mathcal{B}(\Hil_{\lambda})} 
		= \notag \\
		&\qquad \qquad 
		= \order(\lambda^2)
		. 
		\label{ray_optics:eqn:ray_optics_Egorov}
	\end{align}
\end{proposition}
\begin{proof}
	Up until Step 3 the proof can be taken verbatim from that of Proposition~\ref{ray_optics:prop:ray_optics_limit}. Instead of proceeding as in equation~\eqref{ray_optics:eqn:proof_ray_optics_step_3} in Step~4, we replace $\pi_{\lambda} \Weyl f \Weyl \pi_{\lambda}$ with $\pi_{\lambda} \Weyl f_{\mathrm{ro}} \Weyl \pi_{\lambda}$. While the two agree up to $\order(\lambda^2)$, just like in equation~\eqref{ray_optics:eqn:replace_M_lambda_M_ro} the $\order(\lambda^2)$ term commutes with $\pi_0$, and thus, the error we introduce in 
	\begin{align*}
		\Bigl [ \pi_{\lambda} \Weyl \Omega \Weyl \pi_{\lambda} \; , \; \pi_{\lambda} \Weyl f \Weyl \pi_{\lambda} \Bigr ]_{\Weyl} &= \pi_{\lambda} \Weyl \Bigl [ \Omega \, , \, \pi_{\lambda} \Weyl f \Weyl \pi_{\lambda} \Bigr ]_{\Weyl} \Weyl \pi_{\lambda} + \order(\lambda^{\infty})
		\\
		&= \pi_{\lambda} \Weyl \Bigl [ \Omega \, , \, f_{\mathrm{ro}} \Bigr ]_{\Weyl} \Weyl \pi_{\lambda} + \order(\lambda^3)
	\end{align*}
	is in fact $\order(\lambda^3)$. The double commutator term in \eqref{ray_optics:eqn:proof_ray_optics_step_3} is zero as 
	\begin{align*}
		\Bigl [ \pi_{\lambda} \Weyl f \Weyl \pi_{\lambda} \; , \; \pi_{\lambda} \Bigr ]_{\Weyl} = \order(\lambda^{\infty})
	\end{align*}
	vanishes to any order. That means there are no $\order(\lambda)$ which modify the symplectic form either, and we have to replace $\{ \, \cdot \, , \, \cdot \, \}_{\lambda}$ with the usual Poisson bracket in equation~\eqref{ray_optics:eqn:Moyal_commutator_equal_lambda_Poisson_bracket} and Step~5 of the proof. Consequently, the resulting ray optics equations are \eqref{ray_optics:eqn:ray_optics_equations_non-scalar} which compared to \eqref{ray_optics:eqn:ray_optics_equations_scalar} are missing the Berry curvature in the symplectic form. This finishes the proof. 
\end{proof}
%

\subsection{Proof of Theorem~\ref{ray_optics:thm:ray_optics}} 
\label{ray_optics:aux:proof_ray_optics}
With these intermediate results in hand, the proof of the ray optics limit is straightforward.%
\begin{proof}[Theorem~\ref{ray_optics:thm:ray_optics}]
	We revert to the notation of Section~\ref{setup:complex} and add the index “$+$” back to the notation. Given that $\Psi_+ := P_{+,\lambda} (\mathbf{E},\mathbf{H}) \in \ran P_{+,\lambda} \, \Pi_{+,\lambda}$ holds and that $\Pi_{\lambda}$ satisfies equation~\eqref{setup:eqn:relation_Pi_P}, we can insert $\Pi_{+,\lambda}$ free of charge, 
	\begin{align}
		\mathcal{F} &\bigl [ \bigl ( \mathbf{E}(t) , \mathbf{H}(t) \bigr ) \bigr ] = 
		\notag \\
		&= 2 \, \Re \, \Bscpro{\Psi_+ \, }{ \, \e^{+ \ii \frac{t}{\lambda} M_{+,\lambda}} \, \Op_{\lambda}^{S \Zak}(f) \, \e^{- \ii \frac{t}{\lambda} M_{+,\lambda}} \Psi_+}_{L^2_{W_{+,\lambda}}(\R^3,\C^6)}
		\notag \\
		&= 2 \, \Re \, \Bscpro{\Psi_+ \, }{ \, \Pi_{+,\lambda} \, \e^{+ \ii \frac{t}{\lambda} M_{+,\lambda}} \, \Op_{\lambda}^{S \Zak}(f) \, \e^{- \ii \frac{t}{\lambda} M_{+,\lambda}} \, \Pi_{+,\lambda} \Psi_+}_{L^2_{W_{+,\lambda}}(\R^3,\C^6)} 
		\, + \notag \\
		&\qquad
		+ \order(\lambda^{\infty}) 
		. 
		\label{ray_optics:eqn:scalar_product_pi_Op_f_pi}
	\end{align}
	Suppose $f$ is scalar, then the claim follows from Proposition~\ref{ray_optics:prop:ray_optics_limit}. Similarly, Proposition~\ref{ray_optics:prop:ray_optics_limit_non-scalar} implies part~(ii) for non-scalar $f$. 
\end{proof}
%

\section{Quantum-light analogies} 
\label{discussion}
The premise of this article was to rigorously establish the quantum-light analogy between semiclassics for the Bloch electron and ray optics in photonic crystals. However, we need to clearly distinguish between analogies in the mathematical structures and similarities in the physics of crystalline solids and photonic crystals.

\subsection{Comparison of semiclassics and ray optics} 
\label{discussion:quantum_light_analogy}
From the perspective of mathematics it is not at all surprising that the semiclassical equations 
\begin{subequations}
	\label{discussion:eqn:semiclassical_equations_Bloch_electron}
	\begin{align}
		\dot{r} &= + \nabla_k h - \lambda \, \Xi \, \dot{k} 
		\\
		\dot{k} &= - \nabla_r h + \dot{r} \times \mathbf{B} 
	\end{align}
\end{subequations}
for a Bloch electron subjected to an external electromagnetic field $\bigl ( - \nabla_r \phi , \mathbf{B} \bigr )$ indeed resemble equation~\eqref{ray_optics:eqn:ray_optics_equations_scalar} where the semiclassical hamiltonian 
\begin{align*}
	h(r,k) &= \bigl ( E_n(k) + \phi(r) \bigr ) + \order(\lambda)
\end{align*}
takes the place of the dispersion relation~\eqref{ray_optics:eqn:ray_optics_Maxwellian} (see \cite{PST:effective_dynamics_Bloch:2003} and references therein for details). The presence of the \emph{anomalous velocity term} $\Xi \, \dot{k}$ in the ray optics equations was key in the early works \cite{Onoda_Murakami_Nagaosa:Hall_effect_light:2004,Raghu_Haldane:quantum_Hall_effect_photonic_crystals:2008} to anticipate topologically protected edge modes in photonic crystals. In fact, \cite{Onoda_Murakami_Nagaosa:geometrics_optical_wave-packets:2006,Raghu_Haldane:quantum_Hall_effect_photonic_crystals:2008,Esposito_Gerace:photonic_crystals_broken_TR_symmetry:2013} all contain the same semiclassical argument showing the quantization of the transverse conductivity for the \emph{quantum} system: in case the Bloch electron is subjected to a constant electromagnetic field and the magnetic flux through the unit cell is rational, the effect of $\mathbf{B}$ can be subsumed by using \emph{magnetic} Bloch bands, and the average current carried by a \emph{filled band}
\begin{align}
	j &= \int_{\BZ} \dd k \, \dot{r} 
	= \int_{\BZ} \dd k \, \bigl ( \nabla_k E_n(k) - \eps \, \Xi(k) \, \mathbf{E} \bigr ) 
	= \eps \, c \times \mathbf{E} 
	\label{discussion:eqn:transverse_conductivity}
\end{align}
is proportional to the antisymmetric matrix $c = \frac{1}{2\pi} \int_{\BZ} \dd k \, \nabla_k \times \mathcal{A}(k)$ made up of the first Chern numbers and the electric field $\mathbf{E}$. While suggestive the argument does \emph{not} work for photonic crystals for reasons that are important and independent of finding a photonic analog of the transverse conductivity.

\paragraph{Typical states} 
\label{par:typical_states}
The leading-order term in \eqref{discussion:eqn:transverse_conductivity} vanishes because the band is completely filled. Such states are typical for semiconductors and isolators where the Fermi energy $E_{\mathrm{F}}$ lies in a gap. Even when one includes finite-temperature effects, these are typically seen as perturbations of the (zero temperature) \emph{Fermi projection} 
\begin{align*}
	P_{\mathrm{F}} &= 1_{(-\infty,E_{\mathrm{F}}]}(H)
	. 
\end{align*}
However, the Maxwell equations describe classical waves, and there is no exclusion principle which forbids us to populate the same frequency band more than once. 

Experiments usually rely on a laser to selectively populate a frequency band. Thus, states are typically peaked around some $k_0 \in \BZ$ and a frequency $\omega_0$, one may think of a laser beam which impinges on the surface of a photonic crystal: the frequency of the laser light fixes the spectral region, and the angle with respect to the surface normal determines $k_0$. A fully filled band would correspond to a carefully concocted cocktail of light moving in all different directions at specific frequencies, something that seems to be much harder to achieve if at all possible. Hence, we have to take the Brillouin zone average with respect to the \emph{reduced} Wigner transform 
\begin{align*}
	\mathrm{w}_{\Psi_+}^{\mathrm{red}}(r,k) := \sum_{\gamma^* \in \Gamma^*} \mathrm{w}_{\Psi_+} \bigl ( r , k + \gamma^* \bigr ) 
\end{align*}
obtained by zone folding the usual Wigner transform $\mathrm{w}_{\Psi_+}$, and $\mathrm{w}_{\Psi_+}^{\mathrm{red}}(r,k)$ is now peaked around $k_0$ rather than constant in $k$. 

\paragraph{Observables} 
\label{par:observables}
We have consciously avoided to call $\Op_{\lambda}(f)$ the (Weyl) \emph{quantization} of the classical observable $f$, as the operator $\Op_{\lambda}(f)$ is not an observable in classical electromagnetism -- those are functionals of the fields. While this distinction may seem pedantic and unnecessary, it is crucial if one wants to imbue expressions such as 
\begin{align*}
	\int_{\R^3} \dd r \int_{\BZ} \dd k \, f(t,r,k) \; w_{\Psi_+}^{\mathrm{red}}(r,k) 
\end{align*}
with physical meaning. In fact, depending on the type of observable, scalar or non-scalar, we have \emph{two different} ray optics equations to choose from. For instance, our discussion in Section~\ref{ray_optics:examples:poynting} explains that only the net energy flux across a surface uses $\dot{r} = + \nabla_k \Omega - \lambda \, \Xi \, \dot{k}$, local averages of the Poynting vector require one to use \emph{simpler} ray optics equations which \emph{omit the anomalous velocity term} at the expense of having to insert a more complicated function $f_{\mathrm{ro}} = \bscpro{\varphi}{f \varphi}_{\HperT} + \order(\lambda)$ into the integral over phase space. 
\medskip

\noindent
All in all, while the hamiltonian equations \eqref{discussion:eqn:semiclassical_equations_Bloch_electron} and \eqref{ray_optics:eqn:ray_optics_equations_scalar} look very similar on the surface, the physics they describe is very different. The presence of the anomalous velocity term incorrectly suggests one is able to repeat the arguments of \eqref{discussion:eqn:transverse_conductivity}: ignoring that completely filled frequency bands are hard to come by \emph{and} that it is unclear what physical quantity the Brillouin zone average of 
\begin{align*}
	\dot{r} = + \nabla_k \Omega + \lambda \, \omega \, \Xi \, \nabla_r (\tau^2) + \order(\lambda^2) 
\end{align*}
corresponds to, it still would not lead to an expression proportional to $c$.

\paragraph{Designing an experiment to probe the $\order(\lambda)$ effects} 
\label{par:designing_an_experiment_to_probe_order_lambda_effects}
Nevertheless, the $\order(\lambda)$ contributions to the ray optics equations contain interesting physics, and the question comes to mind whether it is possible to engineer an experiment where these effects are particularly strong. The reason the leading-order term in \eqref{discussion:eqn:transverse_conductivity} is identically $0$ is the complete filling of the energy band. In photonics, we can turn this premise on its head, instead of indiscriminately exciting a whole band, we can populate states with pin point accuracy. We propose to use states in the \emph{slow} or \emph{frozen mode regime} (see \eg \cite{Figotin_Vitebskiy:slow_light_cubic_quartic:2006}): Here, we are interested in critical points of the frequency band function where in addition to $\nabla_k \omega_n(k_0) = 0$ \emph{at least} also the second-order derivatives vanish, $\mathrm{Hess} \, \omega_n(k_0) = 0$. To see why, one needs to consider the density of states (DOS) $D(\omega)$ -- a quantity which is well-defined because away from $0$, the spectrum of periodic Maxwell operators is believed to be absolutely continuous (proven under additional regularity assumptions on the material weights in \cite{Morame:spectrum_purely_ac_periodic_Maxwell:2000,Suslina:ac_spectra_periodic_operators:2000,Kuchment_Levendorskii:spectrum_periodic_elliptic_operators:2001}). For simplicity, let us assume that in the vicinity of $k_0$, the frequency band behaves as 
\begin{align*}
	\omega_n(k) = \omega_0 + a \, \bigl ( k - k_0 \bigr )^p + \order \bigl ( (k - k_0)^{p+1} \bigr )
\end{align*}
for some integer $p \geq 2$. Then a simple scaling argument yields that the contribution of the band $\omega_n$ near $\omega_0$ to the DOS is 
\begin{align*}
	D(\omega) \approx b \, \bigl ( \omega - \omega_0 \bigr )^{\frac{3}{p} - 1}
\end{align*}
where the factor $3$ stems from the dimension of the ambient space $\R^3$. For generic critical points $p = 2$ and $D(\omega)$ vanishes at $\omega_0$ -- there are \emph{no states to populate}. The additional condition $\mathrm{Hess} \, \omega(k_0) = 0$ implies $p \geq 3$, and the DOS either remains non-zero and finite at $\omega_0$ ($p = 3$) or diverges ($p \geq 4$). These heuristic considerations allow us to conclude that for $p = 3$ the leading-order term 
\begin{align*}
	\int_{\BZ} \dd k \, \nabla_k \omega(k) \, w_{\Psi_+}^{\mathrm{red}}(r,k) \approx 0 
	, 
\end{align*}
vanishes, but there are sufficiently many states to excite because $D(\omega_0) \neq 0$. 

\subsection{Comparison to previous results} 
\label{discussion:previous_results}
Let us close with a comparison of our ray optics equations to previous results; for simplicity, we will adapt the notation used in these papers to make it consistent with ours. The equations Raghu and Haldane arrived at by analogy in \cite{Raghu_Haldane:quantum_Hall_effect_photonic_crystals:2008}, 
\begin{align*}
	\dot{r} &= + \nabla_k \bigl ( \tau_{\eps} \, \tau_{\mu} \, \omega \bigr ) - \lambda \, \Xi \, \dot{k} 
	, 
	\\
	\dot{k} &= - \nabla_r \bigl ( \tau_{\eps} \, \tau_{\mu} \, \omega \bigr ) 
	, 
\end{align*}
are missing the Rammal-Wilkinson-type term $\tau_{\eps} \, \tau_{\mu} \, \mathcal{P} \cdot \nabla_r \ln \frac{\tau_{\eps}}{\tau_{\mu}}$ which contributes to the dispersion to subleading order. Hence, their result is accurate if the perturbation acts on $\eps$ and $\mu$ in exactly the same way, \ie $\tau_{\eps} = \tau_{\mu}$. This is, however, atypical as $\mu = \mu_{\mathrm{vac}}$ usually does not appreciably vary from its vacuum value in many materials. Esposito and Gerace have been able to derive only the equation for $\dot{r}$ via standard perturbation theory \cite{Esposito_Gerace:photonic_crystals_broken_TR_symmetry:2013}. 

The equations of motion in the third work of note by Onoda, Murakami and Nagaosa \cite{Onoda_Murakami_Nagaosa:geometrics_optical_wave-packets:2006} include an additional spin degree of freedom $z$ to cover the case of a single, $n$-fold degenerate band. Their dispersion relation 
\begin{align*}
	\widetilde{\Omega} &= \tau_{\eps} \, \tau_{\mu} \, \omega + \bscpro{z \,}{\, \widetilde{\Omega}_1 z}_{\C^n} 
\end{align*}
\emph{does} include an extra term that is the expectation value of 
\begin{align*}
	\widetilde{\Omega}_1 = - \tfrac{1}{2} \, \tau_{\eps} \, \tau_{\mu} \, \omega \, \bigl ( \widetilde{\mathcal{A}}^E - \widetilde{\mathcal{A}}^H \bigr ) \cdot \nabla_r \ln \tfrac{\tau_{\eps}}{\tau_{\mu}}
\end{align*}
with respect to spin; this extra term is defined in terms of the difference of electric and magnetic “Berry” connections, 
\begin{align*}
	\widetilde{\mathcal{A}}_{jn}^{E,H}(k) := \ii \, \bscpro{\widetilde{\varphi}_j^{E,H}(k) \,}{\nabla_k \widetilde{\varphi}_l^{E,H}(k)}_{\HperT}
	, 
\end{align*}
where $\widetilde{\varphi_j}^{E,H}$ are normalized such that $\bnorm{\widetilde{\varphi}_j^E}_{L^2_{\eps}(\T^3,\C^3)} = 1$ and similarly for the magnetic component. In addition Onoda et al define the “Berry” connection $\widetilde{\mathcal{A}} = \tfrac{1}{2} \bigl ( \widetilde{\mathcal{A}}^E + \widetilde{\mathcal{A}}^H \bigr )$ that is the average of electric and magnetic contributions. While up to the factor of $\nicefrac{1}{2}$ this seems to coincide with the usual Berry connection
\begin{align*}
	\mathcal{A}_{jn} = \ii \, \bscpro{\varphi_j}{\nabla_k \varphi_l}_{\HperT}
	= \ii \, \bscpro{\varphi_j^E}{\nabla_k \varphi_l^E}_{L^2_{\eps}(\T^3,\C^3)} + \ii \, \bscpro{\varphi_j^H}{\nabla_k \varphi_l^H}_{L^2_{\mu}(\T^3,\C^3)}
	, 
\end{align*}
their similarities are deceiving: in general the field energy stored in the electric and magnetic components of a Bloch mode need not be the same, and thus, the normalization factors $\bnorm{\varphi_j^E(k)}_{L^2_{\eps}(\T^3,\C^3)} \neq \bnorm{\varphi_j^H(k)}_{L^2_{\mu}(\T^3,\C^3)}$ of both contributions are different. In that situation the vector bundle associated to the projection
\begin{align}
	\widetilde{\pi}_0(k) &= \sum_{j = 1}^n \frac{1}{4} \ket{\left (
	\begin{matrix}
		\snorm{\varphi_j^E(k)}_{\eps}^{-1} \, \varphi_j^E \\
		\snorm{\varphi_j^H(k)}_{\mu}^{-1} \, \varphi_j^H \\
	\end{matrix}
	\right )} \negmedspace \bra{\left (
	\begin{matrix}
		\snorm{\varphi_j^E(k)}_{\eps}^{-1} \, \varphi_j^E \\
		\snorm{\varphi_j^H(k)}_{\mu}^{-1} \, \varphi_j^H \\
	\end{matrix}
	\right )}
	\label{discussion:eqn:projection_alternative_Bloch_bundle}
\end{align}
whose connection and curvature tensor are $\widetilde{\mathcal{A}}$ and $\widetilde{\Xi} = \dd \widetilde{\mathcal{A}} + \ii \, [\widetilde{\mathcal{A}} , \widetilde{\mathcal{A}}]$, respectively, is distinct from the standard Bloch vector bundle associated to $\pi_0(k) = \sum_{j = 1}^n \sopro{\varphi_j(k)}{\varphi_j(k)} \neq \widetilde{\pi}_0(k)$ endowed with the standard Berry connection $\mathcal{A}$. 

Consequently, it is not possible to relate Onoda et al's equations of motion 
\begin{subequations}\label{discussion:eqn:Onoda_Murakami_Nagaosa_ray_optics_equations}
	\begin{align}
		\dot{r} &= + \nabla_k \widetilde{\Omega} + \bscpro{z \,}{\, \widetilde{\Xi} \, z}_{\C^n} \, \dot{k} - \ii \, \bscpro{z \,}{\, \bigl [ \, \widetilde{\Omega}_1 \; , \, \tfrac{1}{2} \bigl ( \widetilde{\mathcal{A}}^E - \widetilde{\mathcal{A}}^H \bigr ) \bigr ] \, z}_{\C^n} 
		\\
		\dot{k} &= - \nabla_r \widetilde{\Omega} 
		\\
		\dot{z} &= - \ii \, \Bigl ( \mbox{$\sum_{j = 1}^3$} \dot{k}_j \, \widetilde{\mathcal{A}}_j + \widetilde{\Omega}_1 \Bigr ) \, z
	\end{align}
\end{subequations}
to the topology of the standard Bloch bundle associated to $\pi_0$, but instead to the bundle defined through \eqref{discussion:eqn:projection_alternative_Bloch_bundle}. Even just on the level of ODEs, if $\Xi$ and $\widetilde{\Xi}$ are different from one another, they and their respective flows differ on $\order(\lambda)$. What is more, some of the terms in the equations are gauge-dependent. Hence, even though \eqref{discussion:eqn:Onoda_Murakami_Nagaosa_ray_optics_equations} somewhat resemble the other ray optics equations at first glance, it is actually difficult to compare Onoda et al's ray optics equations. 

Lastly, the only other rigorous result \cite{Allaire_Palombaro_Rauch:diffractive_Bloch_wave_packets_Maxwell:2012} considers perturbations of order $\order(\lambda^2)$. For such perturbations, the $\order(\lambda)$ corrections to the leading-order ray optics equations necessarily vanish, and their work does not offer any insight into what the correct subleading terms are. 
\medskip

\noindent
In summary, \emph{none} of the previous works agree beyond leading order with each other as well as our result. (With leading order we mean the contributions beyond those involving $\tau_{\eps} \, \tau_{\mu} \, \omega$ in case the perturbation parameter is not made explicit.) Therefore, if we replaced the flow $\Phi^{\lambda}$ in, say, equation~\eqref{ray_optics:eqn:ray_optics_limit_scalar} with the flow associated to one of the equations above (which differ by $\order(\lambda)$), a Grönwall argument tells us that the magnitude of the error in \eqref{ray_optics:eqn:ray_optics_limit_scalar} will no longer be $\order(\lambda^2)$ but $\order(\lambda)$. 

The second important difference between this work and previous results is that we do \emph{not} assume to the initial states to be a wave packet — even though in optics states that are well-localized in real and momentum space are much more ubiquitous than in condensed matter physics. Hence, our work contains \emph{the only rigorous results} in this direction that include $\order(\lambda)$ corrections, and therefore we have settled the question of what the correct form of the ray optics equations is conclusively. 
%
\begin{appendix}
	\section{Pseudodifferential calculus for equivariant operator-valued symbols} 
	\label{appendix:PsiDOs}
	The main point of \cite{DeNittis_Lein:adiabatic_periodic_Maxwell_PsiDO:2013} was to explain how $M_{\lambda} = \Op_{\lambda}^{S \Zak}(\Msymb_{\lambda})$ can be understood as the pseudodifferential operator associated to the semiclassical symbol~\eqref{setup:eqn:symbol_Maxwellian}. We content ourselves giving only the necessary definitions and refer the interested reader to \cite[Section~4]{DeNittis_Lein:adiabatic_periodic_Maxwell_PsiDO:2013} and references therein. Simply put, $\Op_{\lambda}$ maps $r$ onto $\ii \lambda \nabla_k$ and $k$ onto the multiplication operator $\hat{k}$. The formal expression~\eqref{setup:eqn:definition_Op_lambda} for $\Op_{\lambda}(f)$ needs to be interpreted properly: Assume $\mathfrak{h}_1$ and $\mathfrak{h}_2$ are Banach or Hilbert spaces; in our applications, they stand for $L^2(\T^3,\C^6)$, $\HperT$ and $\domainT$. 
	We recall that $\HperT$ is $L^2(\T^3,\C^6)$ with scalar product weighted by $W^{-1}$ and $\domainT \subseteq \HperT$ is the domain of the unperturbed fibered Maxwell operator endowed with the graph norm. On all these spaces the actions of the multiplication operators $\e^{\pm \ii \gamma^* \cdot \hat{y}}$ are well-defined.
	A function $f \in \Cont^{\infty} \bigl ( \R^6 , \mathcal{B}(\mathfrak{h}_1,\mathfrak{h}_2) \bigr )$ is called \emph{equivariant} if and only if 
	\begin{align}
		f(r,k - \gamma^*) = \e^{+ \ii \gamma^* \cdot \hat{y}} \, f(r,k) \, \e^{- \ii \gamma^* \cdot \hat{y}} 
		\label{appendix:PsiDOs:eqn:equivariance}
	\end{align}
	holds for all $(r,k) \in \R^6$ and $\gamma^* \in \Gamma^*$. Operator-valued Hörmander symbols 
	\begin{align*}
		\Hoer{m} \bigl ( \mathcal{B}(\mathfrak{h}_1,\mathfrak{h}_2) \bigr ) :\negmedspace &= \left \{ 
			f \in \Cont^{\infty} \bigl ( \R^6,\mathcal{B}(\mathfrak{h}_1,\mathfrak{h}_2) \bigr )
			\; \; \big \vert \; \; 
			\forall \alpha , \beta \in \N_0^3 : \snorm{f}_{m , \alpha , \beta} < \infty
			\right \}
	\end{align*}
	of order $m \in \R$ and type $\rho \in [0,1]$ are defined through the usual seminorms 
	\begin{align*}
		\snorm{f}_{m , \alpha , \beta} := \sup_{(r,k) \in \R^6} \left ( \sqrt{1 + k^2}^{\; -m + \sabs{\beta} \rho} \, \bnorm{\partial_r^{\alpha} \partial_k^{\beta} f(r,k)}_{\mathcal{B}(\mathfrak{h}_1,\mathfrak{h}_2)} \right )
		, 
		&&
		\alpha , \beta \in \N_0^3 
		, 
	\end{align*}
	where $\N_0 := \N \cup \{ 0 \}$. The class of symbols $\Hoerm{m}{\rho} \bigl ( \mathcal{B}(\mathfrak{h}_1,\mathfrak{h}_2) \bigr )$ which satisfy the equivariance condition~\eqref{setup:eqn:equivariance} are denoted with $\Hoermeq{m}{\rho} \bigl ( \mathcal{B}(\mathfrak{h}_1,\mathfrak{h}_2) \bigr )$; similarly, $\Hoermper{m}{\rho} \bigl ( \mathcal{B}(\mathfrak{h}_1,\mathfrak{h}_2) \bigr )$ is the class of $\Gamma^*$-periodic symbols, $f(r,k - \gamma^*) = f(r,k)$. Lastly, we introduce the notion of 
	\begin{definition}[Semiclassical symbols]\label{appendix:PsiDOs:defn:semiclassical_symbols}
		Assume $\mathfrak{h}_j$, $j = 1 , 2$, are Banach spaces as above.	A map $f : [0,\lambda_0) \longrightarrow \Hoereq{m} \bigl ( \mathcal{B}(\mathfrak{h}_1,\mathfrak{h}_2) \bigr )$, $\lambda \mapsto f$, is called an \emph{equivariant semiclassical symbol} of order $m \in \R$ and weight $\rho \in [0,1]$, that is $f \in \SemiHoereq{m} \bigl ( \mathcal{B}(\mathfrak{h}_1,\mathfrak{h}_2) \bigr )$, iff there exists a sequence $\{ f_n \}_{n \in \N_0}$, $f_n \in \Hoereq{m - n \rho} \bigl ( \mathcal{B}(\mathfrak{h}_1,\mathfrak{h}_2) \bigr )$, such that for all $N \in \N_0$, one has 
		\begin{align*}
			\lambda^{-N} \left ( f - \sum_{n = 0}^{N - 1} \lambda^n \, f_n \right ) \in \Hoereq{m - N \rho} \bigl ( \mathcal{B}(\mathfrak{h}_1,\mathfrak{h}_2) \bigr ) 
		\end{align*}
		uniformly in $\lambda$ in the sense that for any $a , b \in \N_0^3$, there exist constants $C_{a b} > 0$ so that 
		\begin{align*}
			\norm{f - \sum_{n = 0}^{N - 1} \lambda^n \, f_n}_{m , \alpha , \beta} \leq C_{\alpha , \beta} \, \lambda^N 
		\end{align*}
		holds for all $\lambda \in [0,\lambda_0)$. 
	
		The Fréchet space of \emph{periodic semiclassical symbols} $\SemiHoerper{m} \bigl ( \mathcal{B}(\mathfrak{h}_1,\mathfrak{h}_2) \bigr )$ is defined analogously. 
	\end{definition}
	For such symbols $\Op_{\lambda}(f) : \mathcal{S}'_{\mathrm{eq}}(\R^3,\mathfrak{h}_1) \longrightarrow \mathcal{S}'_{\mathrm{eq}}(\R^3,\mathfrak{h}_2)$ makes sense as a linear, continuous map between equivariant distributions (\cf \cite[p.~90]{DeNittis_Lein:adiabatic_periodic_Maxwell_PsiDO:2013}), and under certain conditions on $f$ the restriction of $\Op_{\lambda}(f)$ to $L^2_{\mathrm{eq}}(\R^3,\mathfrak{h}_1) \subset \mathcal{S}'_{\mathrm{eq}}(\R^3,\mathfrak{h}_1)$ maps into 
	\begin{align*}
		L^2_{\mathrm{eq}}(\R^3,\mathfrak{h}_2) :\negmedspace &= \Bigl \{ \Psi \in L^2_{\mathrm{loc}}(\R^3,\mathfrak{h}_2) \; \; \big \vert \; \; 
		\Bigr . \\
		&\qquad \qquad \Bigl . 
		\Psi(k - \gamma^*) = \e^{+ \ii \gamma^* \cdot \hat{y}} \, \Psi(k) \; \mbox{a.~e.~$\forall \gamma^* \in \Gamma^*$} \Bigr \} 
		. 
	\end{align*}
	Another building block of pseudodifferential calculus is the Moyal product $\Weyl$ implicitly defined through $\Op_{\lambda}(f \Weyl g) := \Op_{\lambda}(f) \, \Op_{\lambda}(g)$. It defines a bilinear continuous map 
	\begin{align*}
		\sharp : \Hoereq{m_1} \bigl ( \mathcal{B}(\mathfrak{h}_1,\mathfrak{h}_2) \bigr ) \times \Hoereq{m_2} \bigl ( \mathcal{B}(\mathfrak{h}_2,\mathfrak{h}_3) \bigr ) \longrightarrow \Hoereq{m_1 + m_2} \bigl ( \mathcal{B}(\mathfrak{h}_1,\mathfrak{h}_3) \bigr ) 
	\end{align*}
	which has an asymptotic expansion 
	\begin{align}
		&f \Weyl g\, \asymp \sum_{n = 0}^{\infty} \lambda^n \, (f \Weyl g)_{(n)} 
		= f \, g - \lambda \, \tfrac{\ii}{2} \{ f , g \} + \order(\lambda^2) 
		\label{setup:eqn:Moyal_asymp_expansion}
	\end{align}
	where $\{ f , g \} := \sum_{j = 1}^3 \bigl ( \partial_{k_j} f \; \partial_{r_j} g - \partial_{r_j} f \; \partial_{k_j} g \bigr )$ is the usual Poisson bracket. Each term $(f \Weyl g)_{(n)}(r,k)$ is a sum of products of derivatives of $f$ and $g$ evaluated at $(r,k)$. 

	For technical reasons, we need to distinguish between the oscillatory integral $f \Weyl g$ and the formal sum $\sum_{n = 0}^{\infty} \lambda^n \, (f \Weyl g)_{(n)}$ when constructing the local Moyal resolvent. To simplify notation though, we will denote the formal sum on the right-hand side with the same symbol $f \Weyl g$. Given that the terms $(f \Weyl g)_{(n)}$ are purely local, the formal sum $\sum_{n = 0}^{\infty} \lambda^n \, (f \Weyl g)_{(n)}$ also makes sense if $f$ and $g$ are defined only on some common, open subset of $\R^6$.

	\section{Derivation of reduced Wigner transform} 
	\label{appendix:reduced_Wigner_transform}
	\begin{proof}[Corollary \ref{ray_optics:cor:phase_space_average}]
		The main purpose here is to correctly include the material weights $W^{-1}$ in the expression for the reduced Wigner transform, and since the technical details justifying all expressions are covered in the proof of \cite[Corollary~2]{Panati_Teufel:propagation_Wigner_functions_Bloch_electron:2004}, we will omit them. 
	
		From the definition of $\Op_{\lambda}^{S \Zak}$, we deduce that the expectation value can also be computed with respect to the $\Hil_0$ scalar product, 
		\begin{align*}
			&\Bscpro{\Psi_+ \,}{\, \Op_{\lambda}^{S \Zak}(f) \Psi_+}_{L^2_{W_{+,\lambda}}(\R^3,\C^6)} 
			= \\
			&\qquad \qquad 
			= \Bscpro{\Psi_+ \,}{\, S(\lambda \hat{x})^{-1} \, \Zak^{-1} \, \Op_{\lambda}(f) \, \Zak \, S(\lambda \hat{x}) \Psi_+}_{L^2_{W_{+,\lambda}}(\R^3,\C^6)} 
			\\
			&\qquad \qquad 
			= \Bscpro{S(\lambda \hat{x}) \Psi_+ \,}{\, \Zak^{-1} \, \Op_{\lambda}(f) \, \Zak \, S(\lambda \hat{x}) \Psi_+}_{L^2_{W_+}(\R^3,\C^6)} 
			. 
		\end{align*}
		The straight-forward equality 
		\begin{align*}
			\Zak^{-1} \, f(\hat{k}) \, \Zak = \Zak^{-1} \, f(- \ii \nabla_y + \hat{k}) \, \Zak 
			= f(-\ii \nabla_x)
		\end{align*}
		for $r$-independent, scalar, $\Gamma^*$-periodic functions extends to pseudodifferential operators defined by scalar, $\Gamma^*$-periodic $\Cont^{\infty}_{\mathrm{b}}$-functions \cite[Proposition~5]{PST:effective_dynamics_Bloch:2003}, 
		\begin{align*}
			\Zak^{-1} \, f \bigl ( \ii \lambda \nabla_k , \hat{k} \bigr ) \, \Zak &= f \bigl ( \lambda \hat{x} , - \ii \nabla_x \bigr ) 
			. 
		\end{align*}
		Here, $f \bigl ( \ii \lambda \nabla_k , \hat{k} \bigr ) := \Op_{\lambda}(f)$ and $f \bigl ( \lambda \hat{x} , - \ii \nabla_x \bigr )$ is the ordinary Weyl quantization of $f$ obtained after replacing $\ii \lambda \nabla_k$ with $\lambda \hat{x}$ and $\hat{k}$ with $- \ii \nabla_x$ in equation~\eqref{setup:eqn:definition_Op_lambda}. Consequently, the expectation value reduces to 
		\begin{align}
			\ldots &= \Bscpro{S(\lambda \hat{x}) \Psi_+ \,}{\, \Zak^{-1} \, \Op_{\lambda}(f) \, \Zak \, S(\lambda \hat{x}) \Psi_+}_{L^2_{W_+}(\R^3,\C^6)} 
			\notag \\
			&= \Bscpro{S(\lambda \hat{x}) \Psi_+ \,}{\, f \bigl ( \lambda \hat{x} , - \ii \nabla_x \bigr ) \, S(\lambda \hat{x}) \Psi_+}_{L^2_{W_+}(\R^3,\C^6)} 
			. 
			\label{ray_optics:eqn:exp_val_Hil_lambda_Hil_0}
		\end{align}
		Next, we will derive the explicit expression for the reduced Wigner transform: we first perform a simple change of variables, 
		\begin{align*}
			&\Bscpro{\Psi_+}{f \bigl ( \lambda \hat{x} , - \ii \nabla_x \bigr ) \Psi_+}_{L^2_{W_+}(\R^3,\C^6)} 
			= \\
			&\qquad = \int_{\R^3} \dd x \, \Psi_+(x) \cdot W_+^{-1}(x) \bigl ( f \bigl ( \lambda \hat{x} , - \ii \nabla_x \bigr ) \Psi_+ \bigr )(x) 
			\\
			&\qquad = \int_{\R^3} \dd x \frac{1}{(2\pi)^3} \int_{\R^3} \dd y \int_{\R^3} \dd \eta \; \e^{- \ii \eta \cdot (y - x)} \, 
			\cdot \\
			&\qquad \qquad \qquad \qquad \qquad \qquad \qquad \quad \cdot \, 
			\Psi_+(x) \cdot W_+^{-1}(x) \, f \bigl ( \tfrac{\lambda}{2}(x + y) , \eta \bigr ) \Psi_+(y) 
			\\
			&\qquad = \frac{1}{(2 \pi \lambda)^3} \int_{\R^3} \dd r \int_{\R^3} \dd z \int_{\R^3} \dd \eta \; \e^{+ \ii \eta \cdot z} \, f(r,\eta) \, 
			\cdot \\
			&\qquad \qquad \qquad \qquad \qquad \qquad \qquad \quad \cdot \, 
			\Psi_+ \bigl ( \tfrac{r}{\lambda} + \tfrac{z}{2} \bigr ) \cdot W_+^{-1} \bigl ( \tfrac{r}{\lambda} + \tfrac{z}{2} \bigr ) \, \Psi_+ \bigl ( \tfrac{r}{\lambda} - \tfrac{z}{2} \bigr ) 
			,
		\end{align*}
		and then use the $\Gamma^*$-periodicity of $f$, 
		\begin{align}
			\ldots &= \frac{1}{(2 \pi \lambda)^3} \int_{\R^3} \dd r \int_{\R^3} \dd z \, \sum_{\gamma^* \in \Gamma^*} \int_{\BZ} \dd k \; \e^{+ \ii (k + \gamma^*) \cdot z} \, f \bigl ( r , k + \gamma^* \bigr ) 
			\cdot \notag \\
			&\qquad \qquad \qquad \qquad \qquad \qquad \qquad \cdot 
			\Psi_+ \bigl ( \tfrac{r}{\lambda} + \tfrac{z}{2} \bigr ) \cdot W_+^{-1} \bigl ( \tfrac{r}{\lambda} + \tfrac{z}{2} \bigr ) \, \Psi_+ \bigl ( \tfrac{r}{\lambda} - \tfrac{z}{2} \bigr ) 
			\notag \\
			&= \lambda^{-3} \, \sum_{\gamma \in \Gamma} \int_{\R^3} \dd r \int_{\BZ} \dd k \; \e^{+ \ii k \cdot \gamma} \, f(r,k) \, 
			\cdot \notag \\
			&\qquad \qquad \qquad \qquad \qquad \qquad \qquad \cdot 
			\Psi_+ \bigl ( \tfrac{r}{\lambda} + \tfrac{\gamma}{2} \bigr ) \cdot W_+^{-1} \bigl ( \tfrac{r}{\lambda} + \tfrac{\gamma}{2} \bigr ) \, \Psi_+ \bigl ( \tfrac{r}{\lambda} - \tfrac{\gamma}{2} \bigr ) 
			. 
			\label{appendix:reduced_Wigner_transform:eqn:symmetric_formula_reduced_Wigner_transform}
		\end{align}
		Factoring out $\int_{\R^3} \dd r \, \int_{\BZ} \dd k \, f(r,k)$ and tracing the arguments on the bottom of p.~9 in \cite{Panati_Teufel:propagation_Wigner_functions_Bloch_electron:2004} yields that the remaining expression coincides with $\mathrm{w}_{\Psi_+}^{\mathrm{red}}$, 
		\begin{align*}
			\Bscpro{\Psi}{f \bigl ( \lambda \hat{x} , - \ii \nabla_x \bigr ) \Psi}_{L^2_{W_+}(\R^3,\C^6)} &= \int_{\R^3} \dd r \int_{\BZ} \dd k \, f(r,k) \; \mathrm{w}_{\Psi_+}^{\mathrm{red}}(r,k)
			. 
		\end{align*}
		Replacing $f$ with $f \circ \Phi_t^{\lambda}$ and $\Psi_+$ with $S(\lambda \hat{x}) \Psi_+$ in \eqref{ray_optics:eqn:exp_val_Hil_lambda_Hil_0} then yields the claim. 
	\end{proof}
	%

	\section{Computation of $\pi_1$ and $f_{\mathrm{ro}}$} 
	\label{appendix:computations}
	\begin{proof}[Lemma~\ref{ray_optics:cor:explicit_expressions_pi_1_f_ro}]
		\textsf{\bfseries Moyal projection} 
		\label{par:moyal_projection}
		The terms of the projection are computed order-by-order from the \emph{projection} and \emph{commutation defects} which are responsible for the block-diagonal and block-offdiagonal contributions, respectively (\cf \cite[equations~(4)--(8)]{PST:sapt:2002}). As $\pi_0$ is a function of $k$ only, the projection defect
		\begin{align*}
			\lambda \, G_1 + \order(\lambda^2) &= \pi_0 \Weyl \pi_0 - \pi_0 
			= 0 
		\end{align*}
		vanishes, and thus, also $\pi_1^{\mathrm{d}} = 0$. 
		
		The offidagonal term is derived from the commutation defect
		\begin{align*}
			\lambda \, &F_1 + \order(\lambda^2) = \bigl [ \Msymb_{\lambda} , \pi_0 \bigr ]_{\Weyl} 
			\\
			&= [ \Msymb_0 , \pi_0 ] + \lambda \, \Bigl ( \bigl [ \Msymb_1 , \pi_0 \bigr ] - \tfrac{\ii}{2} \bigl \{ \Msymb_0 , \pi_0 \bigr \} + \tfrac{\ii}{2} \bigl \{ \pi_0 , \Msymb_0 \bigr \} \Bigr ) + \order(\lambda^2)
			\\
			&= \lambda \, \Bigl ( [ \Msymb_1 , \pi_0 ] + \tfrac{\ii}{2} \nabla_r \bigl ( \tau^2 \, \Mper(\, \cdot \,) \bigr ) \, \cdot \nabla_k \pi_0 + \tfrac{\ii}{2} \nabla_k \pi_0 \cdot \nabla_r \bigl ( \tau^2 \, \Mper(\, \cdot \,) \bigr ) \Bigr ) 
			\, + \\
			&\qquad 
			+ \order(\lambda^2) 
			\\
			&= \lambda \, \, \tau^2 \, \frac{\ii}{2} \sum_{j = 1}^3 \Bigl ( - \partial_{r_j} \ln \tfrac{\tau_{\eps}}{\tau_{\mu}} \; \bigl [ \Sigma_j \, , \, \pi_0 \bigr ] + 2 \, \partial_{r_j} \ln \tau \, \bigl [ \Mper(\, \cdot \,) \, , \, \partial_{k_j} \pi_0 \bigr ]_+ \Bigr ) + \order(\lambda^2) 
		\end{align*}
		where we have used the abbreviation 
		\begin{align*}
			\Sigma_j := W \, \left (
			\begin{matrix}
				0 & e_j^{\times} \\
				e_j^{\times} & 0 \\
			\end{matrix}
			\right )
			. 
		\end{align*}
		The offdiagonal part is now the sum of two terms, 
		\begin{align*}
			\pi_1^{\mathrm{od}} &= \pi_0 \, F_1 \, \pi_0^{\perp} \, \bigl ( \Msymb_0 - \tau^2 \, \omega \bigr )^{-1} \, \pi_0^{\perp} + \pi_0^{\perp} \, \bigl ( \Msymb_0 - \tau^2 \, \omega \bigr )^{-1} \, \pi_0^{\perp} \, F_1 \, \pi_0 
			\\
			&= \tau^{-2} \, \pi_0 \, F_1 \, \pi_0^{\perp} \, \bigl ( \Mper(\, \cdot \,) - \omega \bigr )^{-1} \, \pi_0^{\perp} + \tau^{-2} \, \pi_0^{\perp} \, \bigl ( \Mper(\, \cdot \,) - 
			\omega \bigr )^{-1} \, \pi_0^{\perp} \, F_1 \, \pi_0 
			, 
		\end{align*}
		and because the second term is the adjoint of the first, it suffices to look at only one of them. We first need to figure out the offdiagonal parts of $F_1$, and because it is purely offdiagonal, $F_1 = \pi_0 \, F_1 \, \pi_0^{\perp} + \pi_0^{\perp} \, F_1 \, \pi_0$, we can leave out one of the projections: 
		\begin{align*}
			\pi_0 \, F_1 &= \tau^2 \, \frac{\ii}{2} \sum_{j = 1}^3 \Bigl ( - \partial_{r_j} \ln \tfrac{\tau_{\eps}}{\tau_{\mu}} \; \pi_0 \, \Sigma_j + 2 \, \partial_{r_j} \ln \tau \, \bigl [ \Mper(\, \cdot \,) \, , \, \pi_0 \, \partial_{k_j} \pi_0 \bigr ]_+ \Bigr )
		\end{align*}
		Let us compute each bit in turn: the $\Sigma_j$ define selfadjoint operators on $\HperT$, and hence,  
		\begin{align*}
			\pi_0 \, \Sigma_j &= \sopro{\varphi}{\Sigma_j \varphi}
			. 
		\end{align*}
		while the term involving the anticommutator 
		\begin{align*}
			\bigl [ \Mper(\, \cdot \,) \, , \, \pi_0 \, \partial_{k_j} \pi_0 \bigr ]_+ &= \pi_0 \, \partial_{k_j} \pi_0 \, \bigl ( \Mper(\, \cdot \,) + \omega \bigr ) 
			\\
			&=  \sopro{\varphi}{\varphi} \, \bigl ( \sopro{\partial_{k_j} \varphi}{\varphi} + \sopro{\varphi}{\partial_{k_j} \varphi} \bigr ) \, \bigl ( \Mper(\, \cdot \,) + \omega \bigr )
			\\
			&= \bigl ( \sopro{\varphi}{\partial_{k_j} \varphi} + \bscpro{\varphi}{\partial_{k_j} \varphi}_{\HperT} \, \pi_0 \bigr ) \, \bigl ( \Mper(\, \cdot \,) + \omega \bigr ) 
		\end{align*}
		Putting everything together, we obtain 
		\begin{align*}
			\pi_0 \, \pi_1 \, \pi_0^{\perp} &= \sum_{j = 1}^3 \Bigl ( - \tfrac{\ii}{2} \, \partial_{r_j} \ln \tfrac{\tau_{\eps}}{\tau_{\mu}} \; \sopro{\varphi}{\Sigma_j \varphi} + \ii \, \partial_{r_j} \ln \tau \, \bigl ( \sopro{\varphi}{\partial_{k_j} \varphi} 
			+ \bigr . \Bigr . \\
			&\qquad \qquad \Bigl . \bigl . 
			+ \bscpro{\varphi}{\partial_{k_j} \varphi}_{\HperT} \, \pi_0 \bigr ) \, \bigl ( \Mper(\, \cdot \,) + \omega \bigr )  \Bigr ) \, \pi_0^{\perp} \, \bigl ( \Mper(\, \cdot \,) - \omega \bigr )^{-1} \, \pi_0^{\perp}
			\\
			&= \sum_{j = 1}^3 \Bigl ( - \tfrac{\ii}{2} \, \partial_{r_j} \ln \tfrac{\tau_{\eps}}{\tau_{\mu}} \; \sopro{\varphi}{\Sigma_j \varphi} + \ii \, \partial_{r_j} \ln \tau \, \sopro{\varphi}{\partial_{k_j} \varphi} \, \bigl ( \Mper(\, \cdot \,) + \omega \bigr )  \Bigr ) 
			\cdot \\
			&\qquad \qquad \cdot 
			\pi_0^{\perp} \, \bigl ( \Mper(\, \cdot \,) - \omega \bigr )^{-1} \, \pi_0^{\perp}
		\end{align*}
		for one of the two contributions to $\pi_1 = \pi_0 \, \pi_1 \, \pi_0^{\perp} + \bigl ( \pi_0 \, \pi_1 \, \pi_0^{\perp} \bigr )^*$. 
		\medskip
		
		\noindent
		\textsf{\bfseries Ray optics observable} 
		\label{par:ray_optics_observable}
		There are two types of terms in \eqref{ray_optics:eqn:f_ro_explicit}, two terms involving $\pi_1$ and two with Poisson brackets. Let us start with the former: since $\pi_1$ is completely offdiagonal, we can compute the sum of the first two terms as 
		\begin{align*}
			\pi_1 \, f \, \pi_0 + \pi_0 \, f \, \pi_1 &= 
			  \pi_0^{\perp} \, \pi_1 \, \pi_0 \, f \, \pi_0 
			+ \pi_0 \, f \, \pi_0^{\perp} \, \pi_1 \, \pi_0 \, 
			+ \\
			&\qquad 
			+ \pi_0 \, \pi_1 \, \pi_0^{\perp} \, \, f \, \pi_0 
			+ \pi_0 \, f \, \pi_0 \, \pi_1 \, \pi_0^{\perp}
			\\
			&= \Bigl ( \bscpro{\varphi \,}{\, \pi_1 \, \pi_0^{\perp} \, f \varphi}_{\HperT} + \bscpro{\varphi \,}{\, f \, \pi_0^{\perp} \, \pi_1 \varphi}_{\HperT} \Bigr ) \, \pi_0 
			+ \\
			&\qquad 
			+ \bscpro{\varphi}{f \varphi}_{\HperT} \, \bigl ( \pi_0 \, \pi_1 \, \pi_0^{\perp} + \pi_0^{\perp} \, \pi_1 \, \pi_0 \bigr ) 
			\\
			&= \Bigl ( \bscpro{\pi_0^{\perp} \, \pi_1 \varphi \,}{\, f \varphi}_{\HperT} + \bscpro{\varphi \,}{\, f \, \pi_0^{\perp} \, \pi_1 \varphi}_{\HperT} \Bigr ) \, \pi_0 + \bscpro{\varphi}{f \varphi}_{\HperT} \, \pi_1 
			\\
			&= \bscpro{\varphi \,}{\bigl [ f , \pi_1 \bigr ]_+ \, \varphi}_{\HperT} \, \pi_0 + \bscpro{\varphi}{f \varphi}_{\HperT} \, \pi_1 
			. 
		\end{align*}
		The only terms that remain are the two Poisson brackets, 
		\begin{align*}
			\bigl \{ \pi_0 , f \bigr \} \, \pi_0 + \pi_0 \, \bigl \{ f , \pi_0 \bigr \} = \sum_{j = 1}^3 \Bigl ( \partial_{k_j} \pi_0 \, \partial_{r_j} f \, \pi_0 - \pi_0 \, \partial_{r_j} f \, \partial_{k_j} \pi_0 \Bigr ) 
			. 
		\end{align*}
		We insert 
		\begin{align*}
			\partial_{k_j} \pi_0 &= \pi_0 \, \partial_{k_j} \pi_0 \, \pi_0^{\perp} + \pi_0^{\perp} \, \partial_{k_j} \pi_0 \, \pi_0
			= \pi_0 \, \partial_{k_j} \pi_0 + \partial_{k_j} \pi_0 \, \pi_0
		\end{align*}
		into the above and compute 
		\begin{align*}
			\ldots &= \pi_0 \, \partial_{k_j} \pi_0  \, \partial_{r_j} f \, \pi_0 + \partial_{k_j} \pi_0 \, \pi_0 \, \partial_{r_j} f \, \pi_0 - \pi_0 \, \partial_{r_j} f \, \pi_0 \, \partial_{k_j} \pi_0 
			\, + \\
			&\qquad 
			- \, \pi_0 \, \partial_{r_j} f \, \partial_{k_j} \pi_0 \, \pi_0 
			\\
			&= \Bigl ( \bscpro{\varphi}{\partial_{k_j} \pi_0 \, \partial_{r_j} f \varphi}_{\HperT} - \bscpro{\varphi}{\partial_{r_j} f \, \partial_{k_j} \pi_0 \varphi}_{\HperT} \Bigr ) \, \pi_0 + 
			\\
			&\qquad + 
			\bscpro{\varphi}{\partial_{r_j} f \varphi}_{\HperT} \, \partial_{k_j} \pi_0 \, \pi_0 - \bscpro{\varphi}{\partial_{r_j} f \varphi}_{\HperT} \, \pi_0 \, \partial_{k_j} \pi_0 
			\\
			&= \Bscpro{\varphi}{\bigl [ \partial_{k_j} \pi_0 \,  , \, \partial_{r_j} f \bigr ] \varphi}_{\HperT} \, \pi_0 + 
			\bscpro{\varphi}{\partial_{r_j} f \varphi}_{\HperT} \, \bigl [ \partial_{k_j} \pi_0 \, , \, \pi_0 \bigr ]
		\end{align*}
		where we have omitted the sum for brevity. Thus, the ray optics observable computes to 
		\begin{align*}
			f_{\mathrm{ro}} &= \bscpro{\varphi}{f \varphi}_{\HperT} \, \pi_0 + \lambda \, \Bigl ( \bscpro{\varphi \,}{\, \bigl [ f , \pi_1 \bigr ]_+ \varphi}_{\HperT} - \tfrac{\ii}{2} \, \bscpro{\varphi}{\bigl [ \nabla_k \pi_0 \, , \nabla_r f \bigr ] \varphi}_{\HperT} \Bigr ) \, \pi_0 
			\Bigr . + \\
			&\qquad + 
			\lambda \, \Bigl ( 
			\bscpro{\varphi}{f \varphi}_{\HperT} \, \pi_1 
			- \tfrac{\ii}{2} \bscpro{\varphi \,}{\nabla_r f \varphi}_{\HperT} \cdot \bigl [ \nabla_k \pi_0 \, , \, \pi_0 \bigr ]
			\Bigr ) 
		\end{align*}
		where by definition $\bigl [ \nabla_k \pi_0 \, , \nabla_r f \bigr ] := \nabla_k \pi_0 \cdot \nabla_r f - \nabla_r f \cdot \nabla_k \pi_0$. To obtain a simplified expression in case $f = f^*$ takes values in the selfadjoint operators, we note that $f = f^*$ implies $\bigl ( \partial_{r_j} f \bigr )^* = \partial_{r_j} f$, and consequently, we obtain 
		\begin{align*}
			&\Bscpro{\varphi}{\bigl [ \partial_{k_j} \pi_0 \,  , \, \partial_{r_j} f \bigr ] \varphi}_{\HperT} 
			= \\
			&\qquad \qquad 
			= \Bscpro{\varphi}{\bigl [ \sopro{\partial_{k_j} \varphi}{\varphi} \,  , \, \partial_{r_j} f \bigr ] \varphi}_{\HperT} + \Bscpro{\varphi}{\bigl [ \sopro{\varphi}{\partial_{k_j} \varphi} \,  , \, \partial_{r_j} f \bigr ] \varphi}_{\HperT}
			\\
			&\qquad \qquad 
			= \bscpro{\varphi}{\partial_{k_j} \varphi}_{\HperT} \, \bscpro{\varphi}{\partial_{r_j} f \varphi}_{\HperT} 
			- \bscpro{\varphi}{\partial_{r_j} f \, \partial_{k_j} \varphi}_{\HperT} 
			+ \\
			&\qquad \qquad \qquad 
			+ \bscpro{\partial_{k_j} \varphi}{\partial_{r_j} f \varphi}_{\HperT} 
			- \bscpro{\varphi}{\partial_{r_j} f \varphi}_{\HperT} \, \bscpro{\partial_{k_j} \varphi}{\varphi}_{\HperT} 
			\\
			&\qquad \qquad 
			= 2 \, \bscpro{\varphi}{\partial_{k_j} \varphi}_{\HperT} \, \bscpro{\varphi}{\partial_{r_j} f \varphi}_{\HperT} - \ii \, 2 \, \Im \bscpro{\varphi}{\partial_{r_j} f \, \partial_{k_j} \varphi}_{\HperT}
			. 
		\end{align*}
		Thus, the commutator terms sum up to 
		\begin{align*}
			&\bigl \{ \pi_0 , f \bigr \} \, \pi_0 + \pi_0 \, \bigl \{ f , \pi_0 \bigr \} 
			= \\
			&\qquad \qquad 
			= \Bigl ( 2 \, \bscpro{\varphi}{\partial_{k_j} \varphi}_{\HperT} \, \bscpro{\varphi}{\partial_{r_j} f \varphi}_{\HperT} - \ii \, 2 \, \Im \bscpro{\varphi}{\partial_{r_j} f \, \partial_{k_j} \varphi}_{\HperT} \Bigr ) \, \pi_0 
			\, + \\
			&\qquad \qquad \qquad + 
			\bscpro{\varphi}{\partial_{r_j} f \varphi}_{\HperT} \, \bigl [ \partial_{k_j} \pi_0 \, , \, \pi_0 \bigr ]
			, 
		\end{align*}
		and overall, we yield the desired expression for 
		\begin{align*}
			f_{\mathrm{ro}} 
			&= \bscpro{\varphi}{f \varphi}_{\HperT} \, \pi_0 + \lambda \, \Bigl (
			\bscpro{\varphi \,}{\bigl [ f , \pi_1 \bigr ]_+ \, \varphi}_{\HperT} \, \pi_0 
			+ \bscpro{\varphi}{f \varphi}_{\HperT} \, \pi_1 
			+ \Bigr . 
			\\
			&\qquad 
			\Bigl . 
			- \tfrac{\ii}{2} \Bigl ( 2 \, \bscpro{\varphi}{\nabla_k \varphi}_{\HperT} \cdot \bscpro{\varphi}{\nabla_r f \varphi}_{\HperT} 
			- \ii \, 2 \, \Im \bscpro{\varphi}{\nabla_r f \cdot \nabla_k \varphi}_{\HperT} \Bigr ) \, \pi_0 
			+ \Bigr . \\
			&\qquad 
			\Bigl . 
			+ \bscpro{\varphi}{\nabla_r f \varphi}_{\HperT} \cdot \bigl [ \nabla_k \pi_0 \, , \, \pi_0 \bigr ]
			\Bigr ) 
			\\
			&= \bscpro{\varphi}{f \varphi}_{\HperT} \, \pi_0 
			+ \\
			&\qquad 
			+ \lambda \, \Bigl ( 
			2 \, \Re \bscpro{f \varphi \,}{\, \pi_1 \varphi}_{\HperT}
			- \bscpro{\varphi}{\nabla_r f \varphi}_{\HperT} \cdot \mathcal{A} 
			\Bigr . \\
			&\qquad \qquad \quad \Bigl . 
			- \, \Im \bscpro{\varphi}{\nabla_r f \cdot \nabla_k \varphi}_{\HperT}
			\Bigr ) \, \pi_0 
			\, + \\
			&\qquad 
			+ \lambda \, \Bigl ( \bscpro{\varphi}{f \varphi}_{\HperT} \, \pi_1 
			+ \bscpro{\varphi}{\nabla_r f \varphi}_{\HperT} \cdot \bigl [ \nabla_k \pi_0 \, , \, \pi_0 \bigr ]
			\Bigr ) 
			.  
		\end{align*}
		%
	\end{proof}
	%
\end{appendix}
%

\printbibliography

\end{document}